\documentclass[manuscript, screen]{acmart}

\usepackage{lipsum}
\usepackage{amsfonts}
\usepackage{graphicx}
\usepackage[normalem]{ulem}
\usepackage{diagbox}
\usepackage{comment}
\usepackage{enumerate}
\usepackage{subfig}
\usepackage[algo2e,linesnumbered,ruled,vlined,boxed]{algorithm2e}
\newlength{\commentWidth}
\let\oldnl\nl
\newcommand{\nonl}{\renewcommand{\nl}{\let\nl\oldnl}}
\DontPrintSemicolon
\SetKwInOut{Input}{Input}\SetKwInOut{Output}{Output}
\usepackage{textcase}

\AtEndPreamble{%
\theoremstyle{acmdefinition}
\newtheorem{remark}[theorem]{Remark}}

\DeclareMathOperator{\poly}{poly}

\DeclareMathOperator{\width}{width}
\DeclareMathOperator{\ent}{ent}
\DeclareMathOperator{\rev}{rev}

\let\epsilon\varepsilon
\let\eps\varepsilon

\newcommand{\NN}{\ensuremath{\mathbb{N}}}

\usepackage{xspace}

\usepackage{tikz}
\usetikzlibrary{decorations.pathreplacing}
\usetikzlibrary{plotmarks}
\usetikzlibrary{positioning,automata,arrows}
\usetikzlibrary{shapes.geometric}
\usetikzlibrary{decorations.markings}
\usetikzlibrary{positioning, shapes, arrows}

\PassOptionsToPackage{usenames,dvipsnames,svgnames}{xcolor}
\definecolor{orange}{RGB}{235,90,0}
\definecolor{darkorange}{RGB}{175,30,0}
\definecolor{turkis}{RGB}{131,182,182}
\definecolor{darkturkis}{RGB}{31,82,82}
\definecolor{green}{RGB}{102,180,0}
\definecolor{darkgreen}{RGB}{51,90,0}
\definecolor{myblue}{RGB}{0,0,213}
\definecolor{mydarkblue}{RGB}{0,0,100}
\definecolor{mybrightblue}{HTML}{74B0E4}
\definecolor{mybrighterblue}{HTML}{B3EAFA}
\definecolor{lila}{RGB}{102,0,102}
\definecolor{darkred}{RGB}{139,0,0}
\definecolor{darkyellow}{RGB}{188,135,2}
\definecolor{brightgray}{RGB}{200,200,200}
\definecolor{darkgray}{RGB}{50,50,50}
\definecolor{amaranth}{rgb}{0.9, 0.17, 0.31}
\definecolor{alizarin}{rgb}{0.82, 0.1, 0.26}
\definecolor{amber}{rgb}{1.0, 0.75, 0.0}
\definecolor{green(ryb)}{rgb}{0.4, 0.69, 0.2}
\definecolor{hanblue}{rgb}{0.27, 0.42, 0.81}
\definecolor{grannysmithapple}{rgb}{0.66, 0.89, 0.63}

\usepackage{environ}
\NewEnviron{cproblem}[1]{%
\begin{center}\fbox{\parbox{0.6\columnwidth}{%
    {\centering\scshape #1\par}%
    \parskip=1ex
    \everypar{\hangindent=1em}%
    \BODY
}}\end{center}}

\AtBeginDocument{%
  }

\begin{document}

\title{On the Complexity of Computing the Co-lexicographic Width of a Regular Language}


\author{Ruben Becker}
\email{rubensimon.becker@unive.it}
\affiliation{%
  \institution{Ca' Foscari University of Venice}
  \city{Venice}
  \country{Italy}}
  
\author{Davide Cenzato}
\email{davide.cenzato@unive.it}
\affiliation{%
  \institution{Ca' Foscari University of Venice}
  \city{Venice}
  \country{Italy}}

\author{Sung-Hwan Kim}
\email{sunghwan.kim@unive.it}
\affiliation{%
  \institution{Ca' Foscari University of Venice}
  \city{Venice}
  \country{Italy}}

\author{Tomasz Kociumaka}
\email{tomasz.kociuma@mpi-inf.mpg.de}
\affiliation{%
  \institution{Max Planck Institute for Informatics, Saarland Informatics Campus}
  \city{Saarbr\"ucken}
  \country{Germany}
}

\author{Bojana Kodric}
\email{bojana.kodric@unive.it}
\affiliation{%
  \institution{Ca' Foscari University of Venice}
  \city{Venice}
  \country{Italy}}

\author{Alberto Policriti}
\email{alberto.policriti@uniud.it}
\affiliation{%
  \institution{University of Udine}
  \city{Udine}
  \country{Italy}}

\author{Nicola Prezza}
\email{nicola.prezza@unive.it}
\affiliation{%
  \institution{Ca' Foscari University of Venice}
  \city{Venice}
  \country{Italy}}

\renewcommand{\shortauthors}{Becker et al.}

\begin{abstract}

    Co-lex partial orders were recently introduced in (Cotumaccio et al., SODA 2021 and Journal of the ACM 2023) as a powerful tool to index finite state automata, with applications to regular expression matching. They generalize Wheeler orders (Gagie et al., Theoretical Computer Science 2017) and naturally reflect the co-lexicographic order of the strings labeling source-to-node paths in the automaton. 
    Briefly, the co-lex width $p$ of a finite-state automaton measures how \emph{sortable} its states are with respect to the co-lexicographic order among the strings they accept. Automata of co-lex width $p$ can be compressed to $O(\log p)$ bits per edge and admit regular expression matching algorithms running in time proportional to $p^2$ per matched character. 

    The \emph{deterministic co-lex width} of a regular language $\mathcal L$ is the smallest width of such a co-lex order, among all DFAs recognizing $\mathcal L$.
    Since languages of small co-lex width admit efficient and elegant solutions to hard problems such as automata compression and pattern matching in the substring closure of the language, computing the co-lex width of a language is relevant in applications requiring efficient solutions to those problems.
    The paper introducing co-lex orders determined that the deterministic co-lex width $p$ of a language $\mathcal L$ can be computed in time proportional to $m^{O(p)}$, given as input any DFA $\mathcal A$ for $\mathcal L$, of size (number of transitions) $m = |\mathcal A|$. Despite this complexity being polynomial for constant values of $p$ (in particular Wheeler languages, for which $p=1$), the constant in the exponent of this running time is large and the exact complexity of the problem is still not known. 
    
    In this paper, using new techniques, we show that it is possible to
    decide in $O(m^p)$ time if the deterministic co-lex width of the language recognized by a given minimum DFA is strictly smaller than some integer $p\ge 2$. 
    We complement this upper bound with a matching conditional lower bound based on the Strong Exponential Time Hypothesis.
    The problem is known to be PSPACE-complete when the input is an NFA (D'Agostino et al., Theoretical Computer Science 2023); thus, together with that result, our paper essentially settles the complexity of the problem.
\end{abstract}

\ccsdesc[500]{Theory of computation~Regular languages}
\ccsdesc[500]{Theory of computation~Design and analysis of algorithms}
\ccsdesc[500]{Theory of computation~Problems, reductions and completeness}
\ccsdesc[500]{Theory of computation~Sorting and searching}

\keywords{Sorting, Indexing,  Regular Languages, Wheeler automata, Parameterized Complexity}


\maketitle

\section{Introduction}
Wheeler automata were introduced by Gagie et al.~\cite{gagie:tcs17:wheeler} as a natural generalization of prefix-sorting techniques --- standing at the core of the most successful string processing algorithms --- to labeled graphs. Informally speaking, an automaton on alphabet $\Sigma$ is Wheeler if the co-lexicographic (\emph{colex} for short) order of the strings labeling source-to-states paths yields a \emph{total} order of the states.
As shown by Gagie et al.~\cite{gagie:tcs17:wheeler}, Wheeler automata can be encoded in just $O(\log|\Sigma|)$ bits per edge and they support efficient membership and \emph{pattern matching} queries. More precisely, finding all nodes reached by a path (starting in any node) labeled with a given query string can be done in linear time. These properties make them a powerful tool in applications such as regular expression matching and bioinformatics; in the latter, one popular way to cope with the rapidly-increasing number of available fully-sequenced genomes, is to encode them in a pangenome graph: aligning short DNA sequences allows one to discover whether the sequences at hand contain variants recorded (as sub-paths) in the graph~\cite{pangenomeGraphs}. 

Wheeler languages --- that is, regular languages recognized by Wheeler automata --- were later studied by  Alanko et al.\ in \cite{alanko:iac21:wheeler}. In that paper, the authors showed that Wheeler DFAs and Wheeler NFAs have the same expressive power: they recognize the same subset of the regular languages. As a matter of fact, the class of Wheeler languages proved to possess several other remarkable properties, in addition to representing the class of regular languages for which efficient indexing data structures exist. For instance, such languages can be characterized very elegantly with a convex version of the Myhill-Nerode theorem, and the smallest DFAs and NFAs for such languages have asymptotically the same number of states.

The main drawback of Wheeler automata and languages is that they represent a relatively sparse family with respect to the set of all automata/regular languages. In other words, very few automata admit a total order of their states reflecting the co-lexicographic order of strings that can be read on the automaton's walks. As shown by Cotumaccio et al.\ in a line of recent works \cite{cotumaccio:soda21:psortable,Cotumaccio22,CotumaccioJACM23}, a very natural solution to this issue is to drop the totality requirement and look at \emph{colex partial orders} (the formal definition is given in Definition \ref{def: colex order}). While the co-lexicographic width can be defined for general automata, in the deterministic case that we consider in this paper (i.e.\ DFAs), such orders have a very natural interpretation (see Figure \ref{fig:dfawidth}): consider the (possibly infinite) set $I_v$ of strings labeling all walks starting in the source node and ending in node $v$ of a DFA $\mathcal A$. Associate each state $v$ with the (open) interval $\mathcal I(v) = (\inf I_v,\sup I_v)$ on the co-lexicographically-sorted set of finite strings $\Sigma^*$. The colex order $<$ of $\mathcal A$'s states is then the natural partial order of such intervals $(\inf I_v,\sup I_v)$, where $u<v$ if and only if $\sup I_u \preceq \inf I_v$. The \emph{width} $p$ of such a partial order --- deemed the \emph{colex width} of the DFA --- is the cardinality of the largest set of mutually-overlapping such intervals. 
As shown in the foundational work on partial colex orders \cite{CotumaccioJACM23}, the colex width $p$ parameterizes several hard problems on automata: (i)~membership of a string in $\mathcal L(\mathcal A)$ and in its substring closure can be determined in time proportional to $p^2$ per matched character, (ii)~any NFA of width $p$ with $n$ states admits an equivalent DFA with at most $2^p(n-p+1)$ states, and (iii)~any automaton of width $p$ can be encoded in just $O(\log p + \log|\Sigma|)$ bits per transition.

\begin{figure}[ht!]
    \centering
    \subfloat[DFA $\mathcal{A}(=\mathcal{A}_{min})$]{
  \begin{tikzpicture}[
    x=2.0cm,y=1.5cm,
    v/.style={circle, draw=black, minimum size=7mm},
    ]
    
    \node  (v0) at (-0.5,1) {};
    \node[v]  (v1) at (0,1) {$v_1$};
    \node[v]  (v2) at (1,1) {$v_2$};
    \node[v]  (v3) at (2,1) {$v_3$};
    \node[v]  (v4) at (0,0) {$v_4$};
    \node[v]  (v5) at (1,0) {$v_5$};
    \node[v,accepting]  (v6) at (2,0) {$v_6$};

    \path [-stealth, thick]
        (v0) edge (v1)
        (v1) edge [above] node {0} (v2)
        (v1) edge [left] node {1} (v4)
        (v2) edge [above] node {1} (v3)
        (v2) edge [right] node {0} (v5)
        (v3) edge [bend left] node[right] {0} (v6)
        (v4) edge [above,left,pos=0.6] node {1} (v2)
        (v5) edge [below] node {0} (v4)
        (v5) edge [below] node {1} (v6)
        (v6) edge [bend left] node[left] {1} (v3)
        (v6) edge [loop right] node {0} (v6)
        ;
    
  \end{tikzpicture}
}
\subfloat[Infima and suprema]{
  \footnotesize
  \begin{tabular}[b]{l|r|r}
    \hline
    $v$ & $\inf I_v$ & $\sup I_v$ \\\hline
    1 & $\epsilon$ & $\epsilon$ \\
    2 & $0$ & $11$ \\
    3 & $01$ & $111$ \\
    4 & $000$ & $1$ \\
    5 & $00$ & $110$ \\
    6 & $\cdots0000$ & $1101$ \\
    \hline
  \end{tabular}
}

\subfloat[Interval representation]{
  \footnotesize
  \begin{tikzpicture}[
    x=0.9cm,y=-0.25cm,
    v/.style={circle, draw={black}, fill={white}, inner sep=1pt},
    vv/.style={circle, draw={black}, fill={white}, inner sep=1pt},
    w/.style={circle, fill={red}, inner sep=1pt},
    ]
    
    \node    (hl) at (0,0) {}; 
    \node    (hr) at (12,0) {}; 
    \node[label={[label distance=0cm,text depth=1em,rotate=-90] left: $\epsilon$}] at (1,0) {};
    \node[label={[label distance=0cm,text depth=1em,rotate=-90] left: $0$}] at (2,0) {};
    \node[label={[label distance=0cm,text depth=1em,rotate=-90] left: $00$}] at (3,0) {};
    \node[label={[label distance=0cm,text depth=1em,rotate=-90] left: $000$}] at (4,0) {};
    \node[label={[label distance=0cm,text depth=1em,rotate=-90] left: $\cdots0000$}] at (5,0) {};
    \node[label={[label distance=0cm,text depth=1em,rotate=-90] left: $110$}] at (6,0) {};
    \node[label={[label distance=0cm,text depth=1em,rotate=-90] left: $1$}] at (7,0) {};
    \node[label={[label distance=0cm,text depth=1em,rotate=-90] left: $01$}] at (8,0) {};
    \node[label={[label distance=0cm,text depth=1em,rotate=-90] left: $1101$}] at (9,0) {};
    \node[label={[label distance=0cm,text depth=1em,rotate=-90] left: $11$}] at (10,0) {};
    \node[label={[label distance=0cm,text depth=1em,rotate=-90] left: $111$}] at (11,0) {};
    
    \node[v,label=left:$v_1$] (v1l) at (1,1) {}; 
    \node[v] (v1r) at (1,1) {};

    \node[v,label=left:$v_2$] (v2l) at (2,1) {};
    \node[v] (v2r) at (10,1) {};

    \node[v,label=left:$v_3$] (v3l) at (8,2) {};
    \node[v] (v3r) at (11,2) {};

    \node[v,label=left:$v_4$] (v4l) at (4,2) {};
    \node[v] (v4r) at (7,2) {};
    \node[v,label=left:$v_5$] (v5l) at (3,3) {};
    \node[v] (v5r) at (6,3) {};
    \node[vv,label=left:$v_6$] (v6l) at (5,4) {};
    \node[v] (v6r) at (9,4) {};

    \path [thick]
        (hl) edge (hr)
        (v2l) edge (v2r)
        (v3l) edge (v3r)
        (v4l) edge (v4r)
        (v5l) edge (v5r)
        (v6l) edge (v6r)
        ;

    \node[w] (w0) at (5.5,1) {};
    \node[w] (w1) at (5.5,2) {};
    \node[w] (w2) at (5.5,3) {};
    \node[w,label={below,text=red}:width$(\mathcal{A})$\text{$=$}4] (w3) at (5.5,4) {};
    \draw [red]
        (5.4,0.6) to (5.6,0.6)
        (5.6,0.6) to (5.6,4.4)
        (5.6,4.4) to (5.4,4.4)
        (5.4,4.4) to (5.4,0.6)
    ;
  \end{tikzpicture}
}
    \caption{Interval representation of infima and suprema strings of a (minimum) DFA.}
    \label{fig:dfawidth}
\end{figure}
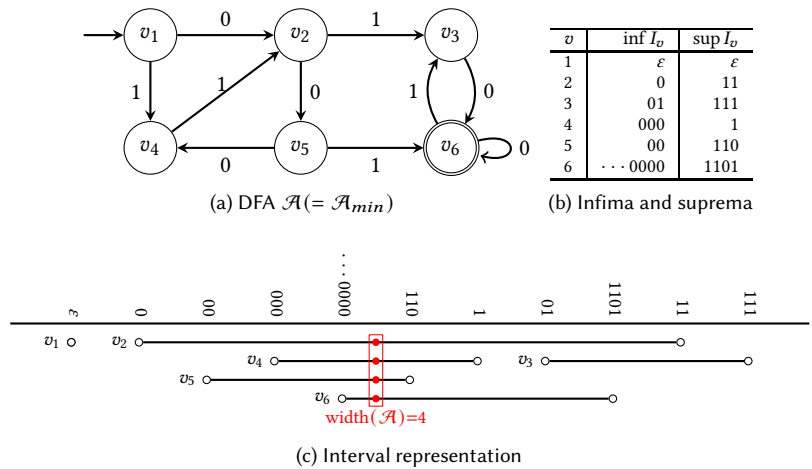

In view of such important properties, it is of interest to determine exactly the smallest colex width of automata accepting a given regular language $\mathcal L$.
If no restriction is imposed on the accepting automata (i.e.\ they can be arbitrary NFAs), then the corresponding quantity of interest is called the \emph{nondeterministic colex width} $\width^N(\mathcal L)$. If, on the other hand, one is interested in the smallest width of a DFA accepting $\mathcal L$, then the quantity is called the \emph{deterministic colex width} $\width^D(\mathcal L)$. As shown in \cite{CotumaccioJACM23}, these two quantities define two hierarchies of the regular languages, classifying them by their propensity to be compressed, sorted, and indexed. 
Interestingly, the two hierarchies do not coincide except for the lowest level, $\width^N(\mathcal L) = \width^D(\mathcal L) = 1$, capturing precisely the class of Wheeler languages. 

This paper is devoted to characterizing the fine-grained complexity of computing $\width^D(\mathcal L)$ given a DFA accepting the language $\mathcal L$.
More precisely, given a DFA $\mathcal A$ with $m$ transitions and an integer $p \ge 2$, in Theorem~\ref{thm:simple algorithm} we show that the problem of determining whether $\width^D(\mathcal L(\mathcal A)) < p$ can be solved in $O(m^p)$ time. We then refine this result in Theorem \ref{thm: parameterized running time} with a faster optimized algorithm; a C++ implementation of the algorithm behind Theorem \ref{thm: parameterized running time} is publicly available at \url{https://github.com/regindex/DeterministicWidth}.
Finally, in Theorem \ref{thm: hardness} we 
complement our upper bounds with matching lower bounds holding under the Strong Exponential Time Hypotheses (SETH).

While computing the width of a given DFA $\mathcal A$ is an easy problem (see also Table \ref{tab:width complexities}), we stress out that computing $\width^D(\mathcal L(\mathcal A))$ is a much harder task. For instance, the smallest DFA $\mathcal A'$ of minimum width equivalent to $\mathcal A$ could be exponentially larger than $\mathcal A$, even if $\width^D(\mathcal L(\mathcal A))  = 1$ \cite{manzini_et_al:LIPIcs.CPM.2024.23}. This means that, even in the case $\width^D(\mathcal L(\mathcal A))  = 1$ (Wheeler languages), in the worst case an algorithm building explicitly $\mathcal A'$ would be exponentially slower than our solution (running in quadratic time on Wheeler languages).

In the next section we introduce all necessary definitions, give a more formal definition of the above problem (as well as related ones), and discuss the state of the art in the field.

\section{Preliminaries, Problems, and State of the Art}\label{sec:preliminaries}

\subsection{Model of computation}\label{sec:RAM model}

Our algorithms work in the word-RAM model with memory word size of $w$ bits and a finite memory of size  $2^w$ words. Letting $N$ be the size (number of words) of the input, we assume that $w \geq \log_2 (N+1)$, i.e.\ that an address of one memory word is sufficient to access any portion of the input. We assume that standard arithmetic operations (including multiplication and division) between integers of $O(1)$ words, as well as bitwise operations between integers of $O(1)$ words and de-referencing a memory address, take constant time. 

We notice that, while sometimes in the literature the word-RAM model is assumed to have infinite memory (see, e.g., Hagerup \cite{HagerupRAM}), this creates issues with algorithms --- like the ones we present in this paper --- working in super-polynomial space (see \cite{bille2015regular}). In particular, in this setting one cannot assume anymore that an address to the working space uses just one machine word. This, in turn, requires to resort to mechanisms such as multiple addressing (for instance, Hagerup assumes double addressing~\cite{HagerupRAM}), which notably complicate algorithms' descriptions as, for instance, de-referencing an address does not take constant time anymore. Since our algorithms may use exponential working space in the worst case, for simplicity we decided to stick to the requirement (also common in the literature) that the working space does not exceed $2^w$ words, which is also a very reasonable assumption in practice. In any case, later (Remark \ref{remark: exp-space in word RAM}) we briefly argue that our results also hold in extended word-RAM models with multiple-addressing, where no limitation is imposed on the working space).

Letting our working space be expressed in big-O notation as $O(W)$ (words), we will sometimes simplify our analyses by saying that the above requirement (i.e.\ \emph{the working space cannot exceed $2^w$ words}) is equivalent to requiring $\log_2(W+1) \leq w$. While this is not strictly true (due to the constant hidden in the big-O notation), it becomes true by adjusting $w$ by an additive constant, an operation which can be simulated, for instance, by double addressing \cite{HagerupRAM} and which does not increase the asymptotic complexity of atomic operations in the model.

\subsection{Intervals and Strings}

With $[N]$ we denote the set of integers $\{1,\ldots, N\}$ and with $[M, N]$ the set of integers $\{M,\ldots, N\}$.
Let $\Sigma$ be a finite alphabet. Without loss of generality, in this paper we take $\Sigma = [\sigma]$ for some integer $\sigma \geq 2$.
A \emph{finite string} $\alpha\in \Sigma^*$  
is a finite concatenation of characters from $\Sigma$. 
The notation $|\alpha|$ indicates the length of the string $\alpha$.
The symbol $\epsilon$ denotes the empty string. 
The notation $\alpha[i]$ denotes the $i$-th character from the beginning of $\alpha$, with indices starting from 1.
Letting $\alpha,\beta\in\Sigma^*$, $\alpha\cdot \beta$ (or simply $\alpha\beta$) denotes the concatenation of strings $\alpha$ and $\beta$.
An \emph{$\omega$-string} $\beta \in \Sigma^\omega$ (or \emph{infinite string} / \emph{string of infinite length}) is an infinite numerable concatenation of characters from $\Sigma$. 
In this paper, we work with \emph{left-infinite} $\omega$-strings, meaning that $\beta \in \Sigma^\omega$ is constructed from the empty string $\epsilon$ by prepending an infinite number of characters to it. In particular, the operation of appending a character $a\in\Sigma$ at the end of a $\omega$-string $\alpha \in \Sigma^\omega$ is well-defined and yields the $\omega$-string $\alpha a$. 
The notation $\alpha^\omega$, where $\alpha \in \Sigma^*$, denotes the concatenation of an infinite (numerable) number of copies of string $\dots \alpha \cdot \alpha$ (prepended to the left). Note that this notation is the reverse of what is usually done in the theory of \(\omega\)-languages, where \(\omega\)-strings grow to the right and \(\alpha^\omega\) denotes a right-infinite string.
The co-lexicographic (or co-lex) order $\prec$ of two strings $\alpha,\beta\in \Sigma^* \cup \Sigma^\omega$ is defined as follows. (i) $\epsilon \prec \alpha$ for every $\alpha \in \Sigma^+ \cup \Sigma^\omega$, and (ii) if $\alpha = \alpha'a$ and $\beta=\beta'b$ (with $a,b\in \Sigma$ and $\alpha',\beta'\in \Sigma^* \cup \Sigma^\omega$), $\alpha \prec \beta$ holds if and only if $(a\prec b)  \vee (a=b \wedge \alpha' \prec \beta')$.
In this paper, the symbols $\prec$ and $\preceq$ will be used both to denote the total order between the alphabet's characters and the co-lexicographic order between strings/$\omega$-strings.

\subsection{Randomization Techniques}\label{subsec: randomization}

We will use Karp-Rabin hashing \cite{KR} (also known as polynomial hashing \cite{Dietzfelbinger1992}) of sets:

\begin{definition}[Karp-Rabin hashing~\cite{KR}] 
Let $q$ be a prime number and let $x \in [q-1]$. 
The Karp-Rabin fingerprint (or hash value) of a set $S \subseteq [n]$ is defined as\footnote{Technically, with this definition the term $x^0$ never appears since our sets are subsets of $[n] = \{1, \dots, n\}$ (we chose to do this for readability). This does not affect the hashing scheme's properties.} $\kappa(S) = \sum_{i\in S} x^i \mod q$.
\end{definition}

In other words, $\kappa(S)$ is the hash of the characteristic bitvector of $S$ (i.e., the bitvector $B \in \{0,1\}^n$ such that $B[i]=1$ if and only if $i\in S$).
If $x$ is chosen uniformly from $[q-1]$ and $S\neq S'$, then $\kappa(S) = \kappa(S')$ (i.e.\ the fingerprints collide) with probability bounded by $n/q$ \cite{Dietzfelbinger1992}.
By choosing $q \geq 2^{cw+1}$ for any constant $c$, the prime $q$ and fingerprints $\kappa(S)$ fit in $O(1)$  words and the collision probability is at most $2^{-cw}$. 
If $i\notin S$, then it is easy to see that $\kappa(S\cup \{i\}) = \kappa(S) + x^i \mod q$. Similarly, if $i\in S$, then $\kappa(S \setminus \{i\}) = \kappa(S) - x^i \mod q$.
Therefore, if $(x^i \mod q)$ is precomputed, these calculations can be performed in constant time in the word RAM model.

In this work, if $M$ is a multiset then $\kappa(M)$ is defined to be the Karp-Rabin fingerprint of the set associated with $M$, i.e.\ $\kappa(M) = \kappa(\{x\ :\ x \in M\})$ (in other words, multiplicities are ignored).

In this paper, \emph{with high probability} (abbreviated \emph{w.h.p.}) means with probability at least $1 - N^{-c}$ for an arbitrarily large constant $c$ fixed at the beginning, where $N$ is the size (number of memory words) of the input.

\subsection{DFAs, Wheeler DFAs, and Co-Lex Width}\label{sec:DFAs}

In this paper, we work with deterministic finite state automata (DFAs): 

\begin{definition}[DFA]
    A DFA $\mathcal A$ is a quintuple $(Q,\Sigma,\delta,s,F)$ where $Q$ is a finite set of states, $\Sigma$ is an alphabet set, $\delta:Q\times\Sigma\rightarrow Q$ is a transition function, $s\in Q$ is a source state, and $F\subseteq Q$ is a set of final states.
\end{definition}

The \emph{regular language} accepted by a DFA $\mathcal A = (Q,\Sigma,\delta,s,F)$ is $\mathcal L(\mathcal A) = \{\alpha\in \Sigma^*\ :\ \delta(s,\alpha)\in F\}$. 

For a state $u\in Q\setminus \{s\}$, we write $\lambda(u)$ for the set of all characters $a\in\Sigma$ such that $u$ has an in-going transition labeled $a$. For the source state $s$, in addition to the characters of all its in-going transitions, we furthermore add an artificial character $\#$ to $\lambda(s)$ such that $\#\prec a$ for all $a\in \Sigma$. For $u,v\in Q$ and $a\in\Sigma$ with $\delta(u,a)=v$, we sometimes write $u\xrightarrow{a} v$ and define $\lambda(u,v)=a$.

We extend the domain of the transition function $\delta$ in two ways: (1) to words $\alpha\in \Sigma^*$ as customary in the literature, i.e., for $a\in \Sigma$, $\alpha\in \Sigma^*$, and $q\in Q$, we let $\delta(q,a\cdot \alpha) = \delta(\delta(q,a),\alpha)$ and $\delta(q,\epsilon) = q$. (2) To sets of states: for any $S\subseteq Q$, we define $\delta(S,a) = \{\delta(u,a)\ :\ u\in S\}$.

Since in this article we deal with properties of languages, without loss of generality we always assume that DFAs are \emph{accessible} --- that is, for every $u\in Q$ there exists $\alpha\in \Sigma^*$ such that $\delta(s,\alpha) = u$ (i.e.\ any state in $Q\setminus \{s\}$ is reachable from the source) --- and \emph{co-accessible} --- that is, for every $u\in Q$, there exists $\alpha\in \Sigma^*$ such that $\delta(u,\alpha) \in F$ (that is, every state can reach a final state). Any DFA can be pruned so that it becomes accessible and co-accessible, without affecting the accepted language.

In this work, $n= |Q|$  denotes the number of states and $m = |\delta| = |\{(u,v,a)\in Q\times Q\times \Sigma :\delta(u,a)=v\}|$ the number of transitions of the input DFA.
Without loss of generality, we assume that $Q = [n]$, $\Sigma = [\sigma]$ for some $\sigma \geq 1$, and that the alphabet is \emph{effective}, i.e.\ that every character in $\Sigma$ appears on some transition (in particular, $\sigma \le m$).
We can also assume, without loss of generality,  $m\geq n-1$: otherwise, there exist states (different than the source $s$) with no incoming transitions. Those states can be safely removed from the automaton, without affecting the recognized language (a property that suffices, since in this paper we discuss algorithms computing properties of languages).

While automata specify initial and final states, in some cases we will only be interested in their topology. A \emph{semiautomaton (semi-DFA if deterministic)} is an automaton that does not specify initial and final states:

\begin{definition}[semi-DFA]
    A semi-DFA $\mathcal A$ is a triple $(Q,\Sigma,\delta)$ where $Q$ is a finite set of states, $\Sigma$ is an alphabet set, and $\delta:Q\times\Sigma\rightarrow Q$ is a transition function.
\end{definition}

Given a DFA $\mathcal A$ and an integer $p\geq 2$, we define its \emph{power semi-DFA} $\mathcal A^p$ as follows:

\begin{definition}[Power semi-DFA]\label{def:power semiDFA}
   Let $\mathcal{A}=(Q, \Sigma, \delta, s, F)$ be a DFA, and let $p\geq 2$. The power semi-DFA $\mathcal A^p$ is the triple $\mathcal A^p = (Q^p, \Sigma, \delta')$, where $\delta' : Q^p \times \Sigma \rightarrow Q^p$ is the transition function such that, for every $U=(u_1,\cdots,u_p),V=(v_1,\cdots,v_p)\in Q^p$ and $a\in\Sigma$, $\delta'(U,a)=V$ iff $\delta(u_i,a)=v_i$ for every $i\in[p]$.
\end{definition}

Given a DFA $\mathcal A$, the set $I_q$ consists of all words \emph{reaching} $q$ from the initial state, formally:

\begin{definition}\label{def:I_q}
    Let $\mathcal A = (Q, \Sigma, \delta, s, F) $ be a DFA. For $u\in Q$, the set $I_u$ is defined as
            $
            		I_u =\{\alpha \in \Sigma^* : u = \delta(s, \alpha)\}
            $
\end{definition}

An important role in our work is played by the \emph{infimum} and \emph{supremum} strings associated with every state:

\begin{definition}[Infimum and supremum strings \cite{alanko_et_al:LIPIcs.CPM.2024.1}]\label{def:inf sup}
Let $\mathcal A = (Q, \Sigma, \delta, s, F) $ be a DFA.
The infimum string $\inf I_u$ and the supremum string $\sup I_u$ of a state $u\in Q$ are defined as:
\begin{align*}
\inf I_u &= \gamma\in\Sigma^*\cup\Sigma^\omega \mbox{ s.t. } (\forall \beta\in\Sigma^*\cup\Sigma^\omega)(\beta 
\preceq I_u \rightarrow \beta \preceq \gamma \preceq I_u) \\
\sup I_u &= \gamma\in\Sigma^*\cup\Sigma^\omega \mbox{ s.t. } (\forall \beta\in\Sigma^*\cup\Sigma^\omega)(I_u \preceq \beta \rightarrow 
I_u \preceq \gamma \preceq \beta) 
\end{align*}
where the notation $\gamma \preceq I_u$ (similarly for $I_u \preceq \gamma $) stands for $(\forall \alpha \in I_u)(\gamma \preceq \alpha)$.
\end{definition}

We furthermore define:

\begin{definition}\label{def:interval of u}
    Let $\mathcal A = (Q, \Sigma, \delta, s, F) $ be a DFA. For $u\in Q$, we define  $\mathcal I(u):=(\inf I_u, \sup I_u) = \{\alpha \in \Sigma^*\ :\ \inf I_u \prec \alpha \prec \sup I_u\} \subseteq \Sigma^*$.
\end{definition}

Note that $\mathcal I(u)$ is an open interval and that $\inf I_u$ and $\sup I_u$ could be left-infinite strings (see  Figure~\ref{fig:dfawidth}).

A classic result from language theory \cite{nerode1958linear} states that the minimum DFA --- denoted with $\mathcal A_{\min}$ --- recognizing the language $\mathcal{L}(\mathcal A)$ 
of any DFA $\mathcal A$
is unique. 
The DFA $\mathcal A_{\min}$  can be computed from $\mathcal A$ in $O(m\log n)$ time with a classic partition-refinement algorithm due to Hopcroft
\cite{hopcroft1971n}.

We will now formally define the co-lex width of a DFA. For this purpose, we first introduce co-lex orders.
\begin{definition}[Co-lex Order \cite{CotumaccioJACM23}] \label{def: colex order}
    Let $\mathcal A = (Q, \Sigma, \delta, s, F)$ be a DFA.
	A \emph{co-lex order} for $\mathcal A$ is a strict partial order $<$ of $Q$ such that the following two conditions hold:
    \begin{enumerate}[(1)]
        \item\label{first colex axiom} For every $u, v\in Q$, if $u < v$, then $\max \lambda(u) \le \min \lambda(v)$. 
        \item\label{second colex axiom} For every $u, v, u',v'\in Q$ and $a\in\Sigma$, if $u=\delta(u', a)$, $v=\delta(v', a)$ and $u < v$, then $u' < v'$.  
    \end{enumerate}
\end{definition}

We remark that condition (1) is more general than the analogous condition required in the definition of Wheeler graphs \cite{gagie:tcs17:wheeler}; in their paper, the authors require \emph{input-consistency}: each incoming edge of a given state is required to have the same label, so that $\lambda(u)$ is a singleton for every $u\in Q$. The generalization of Definition \ref{def: colex order} allows us to work with arbitrary automata and was first proposed in \cite{CotumaccioJACM23}.

As mentioned before, the same notion can be naturally extended to arbitrary NFAs (we omit the details since in this work we focus on DFAs only). We remark that $\#\in \lambda(s)$ implies that for no $u\in Q$ does it hold that $u < s$.
The \emph{width} of a strict partial order $<$ on $Q$ is the size of its largest antichain, i.e., the largest set of pairwise \emph{incomparable} states, where two distinct states $u,v\in Q$ are said to be incomparable if neither $u < v$ nor $v < u$ holds. 

Following Cotumaccio et al.~\cite{CotumaccioJACM23}, we define the co-lex width of a given automaton and the deterministic and nondeterministic co-lex widths of a regular language: 

\begin{definition}[Co-lex Width] \label{def: co-lex width}
    Let $\mathcal A$ be a finite state automaton. 
    \begin{itemize}
        \item The \emph{co-lex width} of $\mathcal A$, $\width(\mathcal A)$, is defined as the minimum width of a co-lex order for $\mathcal A$, i.e., $\width(\mathcal A) = \min\{\width(<)\ |\ <\mathrm{\ is\ a\ colex\ order\ for\ }\mathcal A\}$.
        \item The \emph{deterministic co-lex width} $\width^D(\mathcal L)$ of a regular language $\mathcal L$, is defined as the minimum co-lex width of a DFA $\mathcal A$ accepting $\mathcal L$: $\width^D(\mathcal L) = \min\{\width(\mathcal A)\ |\ \mathcal A\ \mathrm{is\ a\ DFA\ and\ } \mathcal L(\mathcal A) = \mathcal L\}$. 
        \item The \emph{nondeterministic co-lex width} $\width^N(\mathcal L)$ of a regular language $\mathcal L$, is defined as the minimum co-lex width of an NFA $\mathcal A$ accepting $\mathcal L$: $\width^N(\mathcal L) = \min\{\width(\mathcal A)\ |\ \mathcal A\ \mathrm{is\ an\ NFA\ and\ } \mathcal L(\mathcal A) = \mathcal L\}$.
    \end{itemize}
\end{definition}

As shown by Kim et al.~\cite{KimOP23}, 
in the case of DFAs the co-lex width has a very intuitive interpretation: 

\begin{lemma}[Thm. 10 of \cite{KimOP23}]\label{def:width characterization intervals}
    Given any DFA $\mathcal A = (Q, \Sigma, \delta, s, F)$, the co-lex order $<$ such that $\width(\mathcal A) = \width(<)$ is such that, for any two states $u,v\in Q$:
    $$
    u < v \Leftrightarrow \sup I_u  \preceq \inf I_v 
    $$
\end{lemma}

In particular, the above theorem implies that we can adopt an intuitive interval representation for the colex order of a DFA. Later, we will use the following property: 

\begin{corollary}\label{cor: intervals vs colex width}
    Given any DFA $\mathcal A = (Q, \Sigma, \delta, s, F)$, the set of open-ended intervals $\{\mathcal I(u) = (\inf I_u, \sup I_u)\ :\ u\in Q\}$ has width (i.e.\ largest subset of mutually-intersecting intervals) at most $\width(\mathcal A)$.
\end{corollary}
\begin{proof}
Let $<$ be the co-lex order $<$ such that $\width(\mathcal A) = \width(<)$. 
For any two states $u,v\in Q$, if $(\inf I_u, \sup I_u) \cap (\inf I_v, \sup I_v) \ne \emptyset$, then by Lemma \ref{def:width characterization intervals}, $u$ and $v$ are not comparable by $<$, i.e.\ neither $u<v$ nor $v<u$ holds. 
This concludes the proof. Note that the opposite implication does not necessarily hold (in particular, $\width(\mathcal A)$ may be larger than the width of the above interval set). In particular, if $\inf I_u = \sup I_u$ then $(\inf I_u = \sup I_u) \cap (\inf I_v = \sup I_v) = \emptyset$ for any $v\neq u$ (because we are working with open-ended intervals). However, $u$ and $v$ may not be comparable according to $<$, e.g., if $\inf I_v \prec  
 \inf I_u = \sup I_u \prec \sup I_v$.
\end{proof}

Wheeler DFAs as introduced by Gagie et al.~\cite{gagie:tcs17:wheeler} are exactly those DFAs that have co-lex width 1. Wheeler languages \cite{alanko:iac21:wheeler} are regular languages admitting a Wheeler NFA (equivalently, DFA, as shown in \cite{alanko:iac21:wheeler}). 
As an example of low-width language families, Cotumaccio et al. \cite{CotumaccioJACM23} showed that
any regular language that can be obtained by the boolean combination (negation, union, intersection) of a constant number of Wheeler languages (for example, finite languages) has constant deterministic and nondeterministic widths.

As noted above, Gagie et al.~require \emph{input-consistency} of the automaton, i.e., $|\lambda(v)|=1$ for all $v\in Q\setminus\{s\}$. 
Importantly we remark that, from a language-theoretic perspective, restricting to input-consistent automata yields the same notion of deterministic width as the one defined in Definition \ref{def: co-lex width}; this quantity is therefore robust in this sense. This is true because (i) 
our new characterization of the deterministic width (Theorem \ref{thm: main: width - cycle}) depends solely on the minimum DFA for the language, and (ii) any DFA $\mathcal A$ can be easily turned into an equivalent input-consistent DFA $\mathcal A'$ (see also \cite{CotumaccioJACM23}).

We proceed with an example to illustrate the notion of co-lex width.

\begin{figure}[ht]
    \centering
    \subfloat[DFA $\mathcal{A}'$]{
  \begin{tikzpicture}[
    x=2.0cm,y=1.5cm,
    v/.style={circle, draw=black, minimum size=8mm},
    ]
    
    \node  (v0) at (-0.5,1) {};
    \node[v]  (v1) at (0,1) {$v_1$};
    \node[v, label=center:$v'_2$]  (v2a) at (1.3,1.3) {};
    \node[v, label=center:$v''_2$]  (v2b) at (0.7,0.7) {};
    \node[v]  (v3) at (2,1) {$v_3$};
    \node[v]  (v4) at (0,0) {$v_4$};
    \node[v]  (v5) at (1,0) {$v_5$};
    \node[v,accepting]  (v6) at (2,0) {$v_6$};

    \path [-stealth, thick]
        (v0) edge (v1)
        (v1) edge [above] node {0} (v2a)
        (v1) edge [left] node {1} (v4)
        (v2a) edge [above] node {1} (v3)
        (v2a) edge [right, pos=0.75] node {0} (v5)
        (v2b) edge [below, pos=0.75] node {1} (v3)
        (v2b) edge [left, pos=0.6] node {0} (v5)
        (v3) edge [bend left] node[right] {0} (v6)
        (v4) edge [above,left,pos=0.6] node {1} (v2b)
        (v5) edge [below] node {0} (v4)
        (v5) edge [below] node {1} (v6)
        (v6) edge [bend left] node[left] {1} (v3)
        (v6) edge [loop right] node {0} (v6)
        ;
    
  \end{tikzpicture}
}
\subfloat[Infima and suprema]{
  \footnotesize
  \begin{tabular}[b]{l|r|r}
    \hline
    $v$ & $\inf I_v$ & $\sup I_v$ \\\hline
    1 & $\epsilon$ & $\epsilon$ \\
    2$'$ & $0$ & $0$ \\
    2$''$ & $0001$ & $11$ \\
    3 & $01$ & $111$ \\
    4 & $000$ & $1$ \\
    5 & $00$ & $110$ \\
    6 & $\cdots0000$ & $1101$ \\
    \hline
  \end{tabular}
}

\subfloat[Interval representation]{
  \footnotesize
  \begin{tikzpicture}[
    x=0.8cm,y=-0.25cm,
    v/.style={circle, draw={black}, fill={white}, inner sep=1pt},vv/.style={circle, draw={black}, fill={white}, inner sep=1pt},
    w/.style={circle, fill={red}, inner sep=1pt},
    ]
    
    \node    (hl) at (0,0) {}; 
    \node    (hr) at (13,0) {}; 
    \node[label={[label distance=0cm,text depth=1em,rotate=-90] left: $\epsilon$}] at (1,0) {};
    \node[label={[label distance=0cm,text depth=1em,rotate=-90] left: $0$}] at (2,0) {};
    \node[label={[label distance=0cm,text depth=1em,rotate=-90] left: $00$}] at (3,0) {};
    \node[label={[label distance=0cm,text depth=1em,rotate=-90] left: $000$}] at (4,0) {};
    \node[label={[label distance=0cm,text depth=1em,rotate=-90] left: $\cdots0000$}] at (5,0) {};
    \node[label={[label distance=0cm,text depth=1em,rotate=-90] left: $110$}] at (6,0) {};
    \node[label={[label distance=0cm,text depth=1em,rotate=-90] left: $1$}] at (7,0) {};
    \node[label={[label distance=0cm,text depth=1em,rotate=-90] left: $01$}] at (8,0) {};
    \node[label={[label distance=0cm,text depth=1em,rotate=-90] left: $0001$}] at (9,0) {};
    \node[label={[label distance=0cm,text depth=1em,rotate=-90] left: $1101$}] at (10,0) {};
    \node[label={[label distance=0cm,text depth=1em,rotate=-90] left: $11$}] at (11,0) {};
    \node[label={[label distance=0cm,text depth=1em,rotate=-90] left: $111$}] at (12,0) {};
    
    \node[v,label=left:$v_1$] (v1l) at (1,1) {}; 
    \node[v] (v1r) at (1,1) {};

    \node[v,label=left:$v'_2$] (v2al) at (2,1) {};
    \node[v] (v2ar) at (2,1) {};

    \node[v,label=left:$v''_2$] (v2bl) at (9,1) {};
    \node[v] (v2br) at (11,1) {};

    \node[v,label=left:$v_3$] (v3l) at (8,2) {};
    \node[v] (v3r) at (12,2) {};

    \node[v,label=left:$v_4$] (v4l) at (4,1) {};
    \node[v] (v4r) at (7,1) {};
    \node[v,label=left:$v_5$] (v5l) at (3,2) {};
    \node[v] (v5r) at (6,2) {};
    \node[vv,label=left:$v_6$] (v6l) at (5,3) {};
    \node[v] (v6r) at (10,3) {};

    \path [thick]
        (hl) edge (hr)
        (v2al) edge (v2ar)
        (v2bl) edge (v2br)
        (v3l) edge (v3r)
        (v4l) edge (v4r)
        (v5l) edge (v5r)
        (v6l) edge (v6r)
        ;
  \end{tikzpicture}
}
    \caption{DFA $\mathcal{A}'$ with $\width(\mathcal{A}')$=3 and $\mathcal{L}(\mathcal{A}')=\mathcal{L}(\mathcal{A})$ where $\mathcal{A}$ is the DFA in Figure \ref{fig:dfawidth}, which is a certificate of $\width(\mathcal{L}(\mathcal{A}))<4$.}
    \label{fig:dfawidth:p3}
\end{figure}
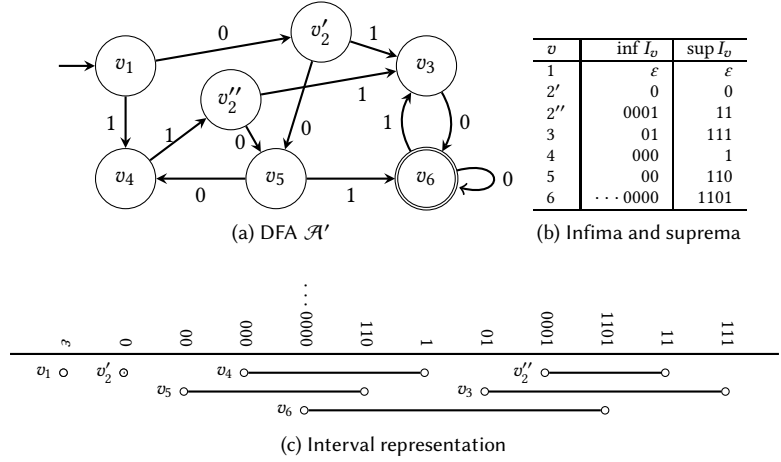

\begin{example}\label{example: language width}
    Recall the DFA $\mathcal{A}$ in Figure~\ref{fig:dfawidth}. From its interval representation, we can easily see that $\width(\mathcal{A})=4$ because the intervals $\mathcal I(v_2), \mathcal I(v_4), \mathcal I(v_5)$ and $\mathcal I(v_6)$ overlap. However, the deterministic width $\width^D(\mathcal{L}(\mathcal{A}))$ of its language $\mathcal{L}(\mathcal{A})$ can be smaller. In Figure~\ref{fig:dfawidth:p3}, we can see that there exists a DFA $\mathcal{A}'$ of smaller width (i.e., $\width(\mathcal{A}')=3$) that accepts the same language $\mathcal{L}(\mathcal{A})$; in other words, $\width^D(\mathcal{L}(\mathcal{A}))\le 3<4$. However, it is not obvious to determine if there exists any DFA of even smaller width that accepts $\mathcal{L}(\mathcal{A})$. In fact, it turns out that there does exist a DFA of width $2$ accepting $\mathcal{L}(\mathcal{A})$ as we show in Figure~\ref{fig:dfawidth:p2}. Is this the best we can do? It turns out that it is, i.e., one can show that no DFA of width 1 can accept $\mathcal{L}(\mathcal{A})$. While this is non-trivial to observe by just looking at the automaton $\mathcal A$ itself, in this paper we present a novel characterization of a language's width allowing us to decide whether $\width^D(\mathcal{L}(\mathcal{A}))< p$ 
    by only inspecting the smallest DFA $\mathcal A_{\min}$ equivalent to $\mathcal A$.
\end{example}

\subsection{Computational Problems and State of The Art}
The above definitions and the considerations made in Example \ref{example: language width} yield the following natural computational problem that is at the core of our paper:

\begin{cproblem}{DfaDetWidth}
    Input: DFA \(\mathcal A\) with $n$ states and $m$ transitions and integer $p \ge 2$.
    
    Output: decide if $\width^D(\mathcal L(\mathcal A)) < p$.
\end{cproblem}\medskip

\begin{figure}[ht!]
    \centering
\subfloat[DFA $\mathcal{A}''$]{
  \begin{tikzpicture}[
    x=2.0cm,y=3.5cm,
    v/.style={circle, draw=black, minimum size=8mm},
    ]
    
    \node  (v0) at (-0.5,1) {};
    \node[v]  (v1) at (0,1) {$v_1$};
    \node[v, label=center:$v'_2$]  (v2a) at (2.3,1.3) {};
    \node[v, label=center:$v'''_2$]  (v2c) at (2.3,1.0) {};
    \node[v, label=center:$v''_2$]  (v2b) at (2.3,0.7) {};
    
    \node[v, label=center:$v_3$]  (v3) at (4.0,1.0) {};
    
    \node[v, label=center:$v'''_4$]  (v4c) at (0.0,0.6) {};
    \node[v, label=center:$v'_4$]  (v4a) at (0.5,0.3) {};
    \node[v, label=center:$v''_4$]  (v4b) at (1.0,0.0) {};
    
    \node[v, label=center:$v''_5$]  (v5b) at (2.3,0.0) {};
    \node[v, label=center:$v'_5$]  (v5a) at (2.3,-0.3) {};
    
    \node[v,accepting, label=center:$v''_6$]  (v6c) at (3.4,-0.3) {};
    \node[v,accepting, label=center:$v'''_6$]  (v6d) at (4,0) {};
    \node[v,accepting, label=center:$v'_6$]  (v6a) at (4.6,0.3) {};

    \path [-stealth, thick]
        (v0) edge (v1)
        (v1) edge [left] node {1} (v4c)
        (v1) edge [bend left, above] node {0} (v2a)

        (v2a) edge [bend left=45, right, pos=0.4] node {0} (v5a)
        (v2b) edge [bend left=10, right] node {0} (v5b)
        (v2c) edge [bend left=45, right] node {0} (v5b)

        (v5a) edge [bend left=60, below] node {0} (v4a)
        (v5b) edge [bend left=10, below] node {0} (v4b)
        
        (v4a) edge [bend left=45, above] node {1} (v2b)
        (v4b) edge [bend left,above] node {1} (v2b)
        (v4c) edge [bend left, above] node {1} (v2c)

        (v2a) edge [bend left=45, above] node {1} (v3)
        (v2b) edge [bend left, above, pos=0.75] node {1} (v3)
        (v2c) edge [bend left=45, above, pos=0.75] node {1} (v3)

        (v5a) edge [bend right=45, above] node {1} (v6c)
        (v5b) edge [bend right=10, above, pos=0.85] node {1} (v6d)

        (v6a) edge [loop right] node {0} (v6a)
        (v6c) edge [bend right=40, below] node {0} (v6a)
        (v6d) edge [above] node {0} (v6a)

        (v6a) edge [bend left=10, left] node {1} (v3)
        (v6c) edge [bend left=40, left, pos=0.75] node {1} (v3)
        (v6d) edge [bend left=10, left, pos=0.75] node {1} (v3)

        (v3) edge [bend left=45, right, pos=0.5] node {0} (v6a)
        ;
    
  \end{tikzpicture}
}

\subfloat[Interval representation]{
  \footnotesize
  \begin{tikzpicture}[
    x=0.65cm,y=-0.25cm,
    v/.style={circle, draw={black}, fill={white}, inner sep=1pt},vv/.style={circle, draw={black}, fill={white}, inner sep=1pt},
    w/.style={circle, fill={red}, inner sep=1pt},
    ]
    
    \node    (hl) at (0.4,0) {}; 
    \node    (hr) at (19.6,0) {}; 
    \node[label={[label distance=0cm,text depth=1em,rotate=-90] left:     $\epsilon$}] at ( 1,0) {};
    \node[label={[label distance=0cm,text depth=1em,rotate=-90] left:            $0$}] at ( 2,0) {};
    \node[label={[label distance=0cm,text depth=1em,rotate=-90] left:           $00$}] at ( 3,0) {};
    \node[label={[label distance=0cm,text depth=1em,rotate=-90] left:          $000$}] at ( 4,0) {};
    \node[label={[label distance=0cm,text depth=1em,rotate=-90] left:   $\cdots0000$}] at ( 5,0) {};
    \node[label={[label distance=0cm,text depth=1em,rotate=-90] left:       $000100$}] at ( 6,0) {};
    \node[label={[label distance=0cm,text depth=1em,rotate=-90] left:         $1100$}] at ( 7,0) {};
    \node[label={[label distance=0cm,text depth=1em,rotate=-90] left:        $00010$}] at ( 8,0) {};
    \node[label={[label distance=0cm,text depth=1em,rotate=-90] left:          $110$}] at ( 9,0) {};
    \node[label={[label distance=0cm,text depth=1em,rotate=-90] left:         $1110$}] at (10,0) {};
    \node[label={[label distance=0cm,text depth=1em,rotate=-90] left:            $1$}] at (11,0) {};
    \node[label={[label distance=0cm,text depth=1em,rotate=-90] left:           $01$}] at (12,0) {};
    \node[label={[label distance=0cm,text depth=1em,rotate=-90] left:          $001$}] at (13,0) {};
    \node[label={[label distance=0cm,text depth=1em,rotate=-90] left:         $0001$}] at (14,0) {};
    \node[label={[label distance=0cm,text depth=1em,rotate=-90] left:        $11001$}] at (15,0) {};
    \node[label={[label distance=0cm,text depth=1em,rotate=-90] left:       $000101$}] at (16,0) {};
    \node[label={[label distance=0cm,text depth=1em,rotate=-90] left:         $1101$}] at (17,0) {};
    \node[label={[label distance=0cm,text depth=1em,rotate=-90] left:           $11$}] at (18,0) {};
    \node[label={[label distance=0cm,text depth=1em,rotate=-90] left:          $111$}] at (19,0) {};

    \node[v,label=left:$v_1$] (v1l) at (1,1) {}; 
    \node[v] (v1r) at (1,1) {};

    \node[v,label=left:$v'_2$] (v2al) at (2,1) {};
    \node[v] (v2ar) at (2,1) {};-
    \node[v,label=left:$v''_2$] (v2bl) at (14,1) {};
    \node[v] (v2br) at (15,1) {};-
    \node[v,label=left:$v'''_2$] (v2cl) at (18,1) {};
    \node[v] (v2cr) at (18,1) {};

    \node[v,label=left:$v_3$] (v3l) at (12,2) {};
    \node[v] (v3r) at (19,2) {};

    \node[v,label=left:$v'_4$] (v4al) at (4,1) {};
    \node[v] (v4ar) at (4,1) {};
    \node[v,label=left:$v''_4$] (v4bl) at (6,1) {};
    \node[v] (v4br) at (7,1) {};
    \node[v,label=left:$v'''_4$] (v4cl) at (11,1) {};
    \node[v] (v4cr) at (11,1) {};
    
    \node[v,label=left:$v'_5$] (v5al) at (3,1) {};
    \node[v] (v5ar) at (3,1) {};
    \node[v,label=left:$v''_5$] (v5bl) at (8,1) {};
    \node[v] (v5br) at (9,1) {};
    
    \node[vv,label=left:$v'_6$] (v6al) at (5,2) {};
    \node[v] (v6ar) at (10,2) {};
    \node[v,label=left:$v''_6$] (v6bl) at (13,1) {};
    \node[v] (v6br) at (13,1) {};
    \node[v,label=left:$v'''_6$] (v6cl) at (16,1) {};
    \node[v] (v6cr) at (17,1) {};

    \path [thick]
        (hl) edge (hr)
        (v2al) edge (v2ar)
        (v2bl) edge (v2br)
        (v2cl) edge (v2cr)
        (v3l) edge (v3r)
        (v4al) edge (v4ar)
        (v4bl) edge (v4br)
        (v4cl) edge (v4cr)
        (v5al) edge (v5ar)
        (v5bl) edge (v5br)
        (v6al) edge (v6ar)
        (v6bl) edge (v6br)
        (v6cl) edge (v6cr)
        ;
  \end{tikzpicture}
}
    \caption{DFA $\mathcal{A}''$ with $\width(\mathcal{A}'')$=2 and $\mathcal{L}(\mathcal{A}'')=\mathcal{L}(\mathcal{A})$ where $\mathcal{A}$ is as in Figure \ref{fig:dfawidth}.}
    \label{fig:dfawidth:p2}
\end{figure}
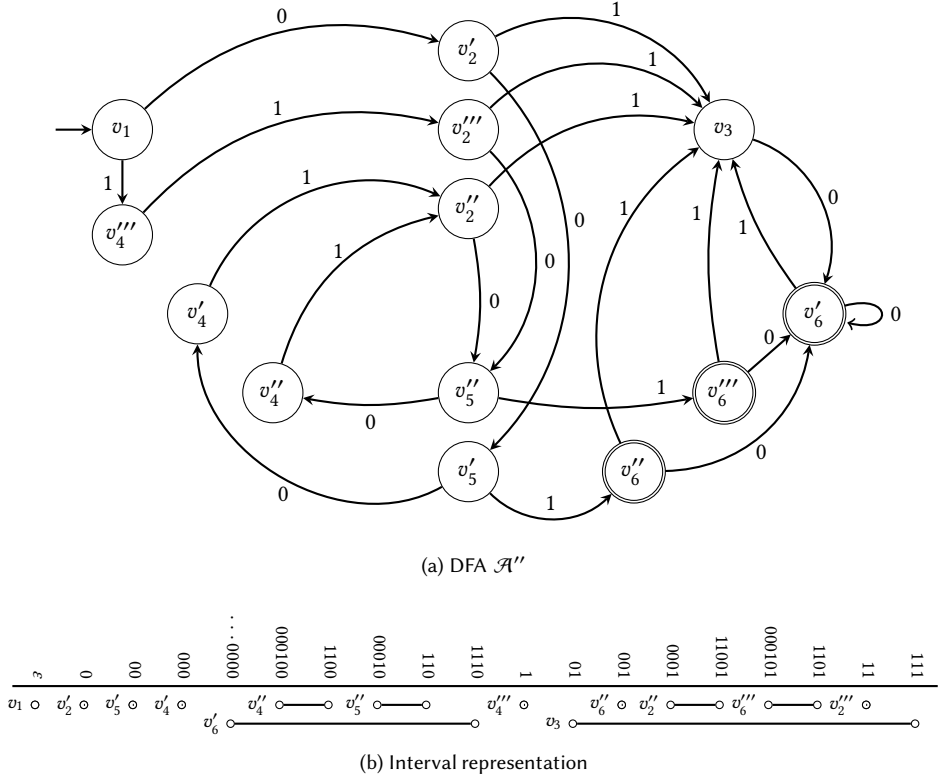

Five other related computational problems have been considered in the literature, depending on whether
(1) the input is a DFA or an NFA, and (2) the quantity to be computed is the co-lex width of the automaton, the deterministic co-lex width, or the non-deterministic co-lex width of the language recognized by the automaton.
In the vein of the terminology introduced for the \textsc{DfaDetWidth} problem we can refer to the additional five problems as: 

\begin{itemize}
    \item \textsc{DfaWidth}: given a DFA $\mathcal A$ and an integer $p\ge 2$, decide if $\width(\mathcal A) < p$.     
    \item \textsc{DfaNDetWidth}: given a DFA $\mathcal A$ and an integer $p\ge 2$, decide if \\$\width^N(\mathcal L(\mathcal A)) < p$.
    \item  \textsc{NfaWidth}: given an NFA $\mathcal A$ and an integer $p\ge 2$, decide if $\width(\mathcal A) < p$.
    \item \textsc{NfaDetWidth}: given an NFA $\mathcal A$ and an integer $p\ge 2$, decide if \\$\width^D(\mathcal L(\mathcal A)) < p$. 
    \item \textsc{NfaNDetWidth}: given an NFA $\mathcal A$ and an integer $p\ge 2$, decide if\\ $\width^N(\mathcal L(\mathcal A)) < p$.
\end{itemize}

As a matter of fact, a large discrepancy in the computational complexity of those problems occurs depending on whether the input is a DFA or an NFA, see Table \ref{tab:width complexities} (updated with our new results with respect to Table 1 of \cite{CotumaccioJACM23}).

\begin{table}[ht!]
\footnotesize
\renewcommand{\arraystretch}{2}
\centering
\begin{tabular}{|c|c|c|}\hline
\backslashbox{{output}\hspace{2pt}}{{\hspace{-5pt}input}}
&\makebox[12em]{{ $\mathcal A$ : DFA}}&\makebox[12em]{{ $\mathcal A$ : NFA}}\\\hline\hline
{$\text{width}(\mathcal A) \stackrel{?}{<} p $} & $O(n^2)$ \cite{cotumaccio2023prefix}, $O(m\log n)$ \cite{BeckerCCKKOP23} &NP-COMPLETE \cite[Thm. 2]{gibney2022complexity} \\\hline
{\hspace{-3pt}${\text{width}^D}(\mathcal L(\mathcal A)) \stackrel{?}{<} p$} \hspace{-3pt} & $O(m^p)$ [Thm.~\ref{thm:simple algorithm}], $\omega(m^{p-\epsilon})$ [Thm.~\ref{thm: hardness}]
& PSPACE-HARD \cite[Thm. 10]{DAgostinoMP23}\\\hline
{\hspace{-3pt}${\text{width}^N}(\mathcal L(\mathcal A))\stackrel{?}{<} 2$ \hspace{-3pt}} &  $O(m^{2})$ [Thm.~\ref{thm:simple algorithm}], $\omega(m^{2-\epsilon})$ [Thm.~\ref{thm: hardness}]   & PSPACE-HARD \cite[Thm. 10]{DAgostinoMP23}\\\hline
\end{tabular}
\vspace{1em}
\caption{Known lower and upper bounds for the six problems considered in the literature involving the computation of a colex width. In row-major order: \textsc{DfaWidth}, \textsc{NfaWidth}, \textsc{DfaDetWidth}, \textsc{NfaDetWidth}, \textsc{DfaNDetWidth}, \textsc{NfaNDetWidth}. $m$ and $n$ are the number of transitions and states of the input automaton,  $2 \le p \leq n$ is an integer (part of the problem instance), and $\epsilon>0$ is any constant. The lower bounds of Theorem~\ref{thm: hardness} are conditional on SETH.
 }\label{tab:width complexities}
\end{table}

The problem \textsc{NfaWidth} is already NP-complete as it includes the NP-hard problem of recognizing Wheeler automata~\cite{gibney2022complexity} (i.e.\ $\text{width}(\mathcal A) < 2$), and a colex partial order of width $<p$ is a polynomial certificate for the problem. 
This complexity increases for the problems \textsc{NfaDetWidth} and \textsc{NfaNDetWidth}, which are PSPACE-hard~\cite{DAgostinoMP23}. 

\textsc{DfaWidth}, instead, turns out to be polynomial-time solvable. Cotumaccio and Prezza~\cite{cotumaccio:soda21:psortable} and Cotumaccio et al.~\cite{CotumaccioJACM23} were the first to give polynomial algorithms solving this problem in $O(m^2 + n^{5/2})$ and $O(m^2)$ time with high probability, respectively. 
This was later improved by
Kim et al.~\cite{KimOP23}, who gave two algorithms running in $O(m n)$ and $O(n^2 \log n)$ time, by Becker et al.~\cite{BeckerCCKKOP23}, who showed how to solve the problem in near-linear time $O(m\log n)$ via partition refinement, and by Cotumaccio \cite{cotumaccio2023prefix}, who gave a recursive algorithm running in $O(n^2)$ time\footnote{In those works, it is assumed without loss of generality that all the edges entering in the same state bear the same label; this implies $m \leq n^2$, so the additive term $m$ gets absorbed by $n^2$.} (thus improving over the previous results in the dense case).

\textsc{DfaDetWidth}, the problem at the core of our paper,  has been studied for the first time by 
Alanko et al.~\cite{alanko:iac21:wheeler} in the special case $\width^D(\mathcal L) \stackrel{?}{<} 2$ (that is, recognizing Wheeler languages from an accepting DFA), for which a polynomial-time algorithm was provided. Later, Cotumaccio et al.~\cite{CotumaccioJACM23}, gave a dynamic programming algorithm solving the general problem in $m^{O(p)}$ time\footnote{While they only claim the bound $m^{O(p)}$, a more careful analysis shows that the running time of their algorithm is at least $\Omega(m^{5p})$ on sparse DFAs.}. No hardness results have been known for this problem prior to our work. 

For \textsc{DfaNDetWidth} little is known: In the special case where $\width^N(\mathcal L) = 1$ the two notions of deterministic and non-deterministic width coincide~\cite{CotumaccioJACM23}: $\width^N(\mathcal L) = \width^D(\mathcal L) = 1$. Hence, recognizing whether the non-determinisitic co-lex width of the language recognized by a DFA is strictly smaller than $2$ (i.e.\ equal to 1) can be done in polynomial time. 
Apart from this and the PSPACE-hardness of \textsc{NfaNDetWidth}, nothing else (not even computability) is known for the two problems \textsc{DfaNDetWidth} and \textsc{NfaNDetWidth}  for when the nondeterministic width is strictly larger than 1. The reason why a simple NFA-enumeration strategy does not work, is that no bound is known on the number of states of an NFA realizing the nondeterministic width of the language (as a function of the input DFA/NFA's size).

\subsection{Entanglement of a Regular Language}
Before describing our contribution, we need one additional important tool.
Cotumaccio et al.~\cite{CotumaccioJACM23} showed that the deterministic co-lex width of a language can be characterized using the notion of \emph{entanglement} of a language:

\begin{definition}[Entanglement, see Def.~4.7 in~\cite{CotumaccioJACM23}]
    Let $\mathcal A=(Q, \Sigma, \delta, s, F)$ be a DFA. A set of states $S\subseteq Q$ is \emph{entangled} if there exists a monotone (in colex order) sequence of strings $(\alpha_i)_{i\in \NN}$ such that for all $u\in S$, it holds that $u=\delta(s, \alpha_i)$ for infinitely many $i\in\NN$. The entanglement of $\mathcal A$, denoted with $\ent(\mathcal A)$, is defined as the maximum cardinality of a set of entangled states, i.e., 
    \[
        \ent(\mathcal A):=\max\{|S|: S\subseteq Q \text{ and $S$ is entangled}\}.
    \]
    \label{def: entangle}
\end{definition}

Cotumaccio et al.~\cite{CotumaccioJACM23} prove the following theorem.

\begin{theorem}[Thm.~4.21 in~\cite{CotumaccioJACM23}]
    Let $\mathcal A = (Q, \Sigma, \delta, s, F)$ be a DFA recognizing a language $\mathcal L$ and let $\mathcal A_{\min}$ be the minimum DFA recognizing $\mathcal L$. Then 
    \[
        \width^D(\mathcal L) = \ent(\mathcal A_{\min}).
    \]
    \label{thm: widthD L equals ent Amin}
\end{theorem}

\section{A new Characterization of the Deterministic Co-Lex Width} \label{sec: main characterization}

In this section, we present a new \emph{topological} characterization of the deterministic co-lex width of regular languages. Before stating the complete characterization for a regular language being of width at least $p$ in Theorem~\ref{thm: main: width - cycle} below, we start with a very simple result that gives a sufficient condition: the existence of $2p$ cycles labeled with the same string in the smallest DFA $\mathcal A_{\min}$ for the language.

\begin{lemma}
    Let $\mathcal A = (Q, \Sigma, \delta, s, F)$ be a DFA recognizing a language $\mathcal L$. If
    there exists a non-empty string $\gamma$ such that
    the minimum DFA $\mathcal A_{\min}$ accepting $\mathcal L$ contains $2p$ pairwise distinct nodes $u_i$, $i \in [2p]$, with $\delta(u_i, \gamma)=u_i$ for all $i\in [2p]$, then $\width^D(\mathcal L)\ge p$.
\end{lemma}
\begin{proof}
    First of all, notice that there exists $\alpha_i\in I_{u_i}$ for each $i\in [2p]$ as otherwise, if $u_i$ was not reached from the source, the state $u_i$ could be eliminated from $\mathcal A_{\min}$, contradicting minimality. Now, w.l.o.g., we can assume that, for some $0 \leq \ell \leq 2p$:
    \[
        \alpha_1 \preceq \ldots \preceq \alpha_\ell \preceq \gamma \preceq \alpha_{\ell + 1} \preceq \ldots \preceq \alpha_{2p}.
    \]
    Note that by repeating the cycle a sufficiently large number of times we can assume $\gamma$ to be sufficiently long such that it is not a suffix of any of the $\alpha_i$. 
    Then, it follows that, for every $j\in \NN_{\ge 0}$, 
    \[
        \alpha_{\ell + 1} \gamma^{j + 1} 
        \preceq \ldots \preceq
        \alpha_{2p} \gamma^{j + 1}
        \preceq 
        \alpha_{\ell + 1} \gamma^{j} 
        \preceq \ldots \preceq
        \alpha_{2p} \gamma^{j}
    \]
    and, on the other hand, 
    \[
        \alpha_{1} \gamma^{j} 
        \preceq \ldots \preceq
        \alpha_{\ell} \gamma^{j}
        \preceq 
        \alpha_{1} \gamma^{j + 1} 
        \preceq \ldots \preceq
        \alpha_{\ell} \gamma^{j + 1}.
    \]
    It follows that both the set of states $\{u_1, \ldots, u_\ell\}$ as well as the set of states $\{u_{\ell + 1}, \ldots, u_{2p}\}$ is entangled. Hence, $\width^D(\mathcal L) = \ent(\mathcal A_{\min}) \ge \max\{\ell, 2p - \ell\} \ge p$, using Theorem~\ref{thm: widthD L equals ent Amin}.
\end{proof}
We stress that this lemma only gives a sufficient condition for the width of the language being at least $p$. The main contribution of this section instead is to give a complete characterization, i.e., a sufficient and necessary condition. It turns out that also this condition is based on cycles, however $p$ cycles are sufficient (rather than $2p$ as before) if we in addition require the intervals $\mathcal I(u_i)$ of the nodes $u_i$ to intersect pairwise. 
Our new complete characterization of the co-lex width of a regular language is given in the following theorem.  
\begin{theorem}\label{thm: main: width - cycle}
    Let $\mathcal A = (Q, \Sigma, \delta, s, F)$ be a DFA recognizing a language $\mathcal L$. Then, $\width^D(\mathcal L)\ge p$ if and only if 
    there exists a non-empty string $\gamma$ such that
    the minimum DFA $\mathcal A_{\min}$ accepting $\mathcal L$ contains $p$ pairwise distinct nodes $u_i$, $i \in [p]$, with $\delta(u_i, \gamma)=u_i$ for all $i\in [p]$ and $\mathcal I(u_i) \cap \mathcal I(u_j) \neq \emptyset$ for all $i,j\in [p]$.
\end{theorem}

We will prove this theorem by showing that the $p$ states $\{u_1,\cdots,u_p\}$ in the theorem are entangled if and only if such cycles exist, and finally combining this result again with Theorem~\ref{thm: widthD L equals ent Amin}. The importance of our new characterization lies in the fact that it allows us to reduce the problem of computing the deterministic width of a language to the problem of detecting particular cycles in a graph. In the next sections, we show how to solve this problem by resorting to power semi-DFAs (see Definition \ref{def:power semiDFA}), and we complete this upper bound with a matching conditional lower bound and other hardness results.

Before proving Theorem~\ref{thm: main: width - cycle}, we  establish some important observations. We start by observing that every state belonging to a cycle that does not correspond to the state's infimum (supremum, resp.) string, is reached by a string being co-lex smaller (greater, resp.) than the infinite string labelling the cycle.

\begin{lemma}
    For a DFA $\mathcal{A}=(Q, \Sigma, \delta, s, F)$, let $u\in Q$ be a state such that $\delta(u,\gamma)=u$ for some non-empty string $\gamma\in\Sigma^*$. If $\gamma^\omega\ne\inf I_u$ ($\ne\sup I_u$, respectively) then, there must exist a string $\alpha\in I_u$ such that $\alpha\prec\gamma^\omega$ ($\gamma^\omega \prec \alpha$, respectively).
    \label{lemma: langwidth: string greater than loop exists}
\end{lemma}
\begin{proof}
    We consider the case $\gamma^\omega\ne\inf I_u$, the other case is analogous. Assume, for  a contradiction, that for every $\alpha\in I_u$ it holds that $\gamma^\omega\prec\alpha$. Then $\gamma^\omega$ is a lower bound of $I_u$. Since $\gamma^\omega\ne\inf I_u$, we have $\gamma^\omega\prec\inf I_u$. Hence, for every $\alpha\in I_u$, there must exist a sufficiently large $k\ge 0$ such that $\alpha\gamma^k\prec\inf I_u$. However, since $\delta(u,\gamma)=u$, then also $\alpha\gamma^k\in I_u$. Thus, we reached a contradiction, as $\alpha\gamma^k\prec\inf I_u$ is not possible.
\end{proof}

Next, we observe two useful properties of entangled states. For every set of entangled states, (1) their corresponding co-lex intervals always overlap, and (2) the entanglement is propagated to some of their predecessors.

\begin{lemma}
    For a DFA $\mathcal{A}=(Q, \Sigma, \delta, s, F)$, let $S=\{u_1,\cdots,u_p\} \subseteq Q$ be a set of entangled states. Then:
    \begin{enumerate}
        \item \label{lemma: langwidth: sub: intervals of entangled states are overlapping} For every distinct $i,j\in [p]$, $\mathcal{I}(u_i)\cap\mathcal{I}(u_{j})\ne\emptyset$. 
        \item \label{lemma: langwidth: sub: entangled predecessor exists} There exists a set $S'\subseteq Q$ of states such that (i) for some $c\in\Sigma$, $S=\{\delta(u,c):u\in S'\}$, (ii) $|S'|=|S|$ and (iii) $S'$ is entangled. We call such a set $S'$ an \emph{entangled predecessor} of $S$ with $c$.
    \end{enumerate}
    \label{lemma: langwidth: entangled properties}
\end{lemma}

\begin{proof}

    Since $S=\{u_1,\cdots,u_p\}$ is entangled, without loss of generality, we can assume that there exists an increasing monotone sequence $(\alpha_t)_{t\in \mathbb{N}}$ (the decreasing case is symmetric) such that $\alpha_{i+kp}\in I_{u_{i}}$ for every $i \in [p]$ and $k\ge 0$.
    
    \begin{enumerate}
        \item Let $i, j\in [p]$ be distinct and assume, w.l.o.g., that $i<j$. Then, from the given monotone sequence, we can obtain strings $\alpha_i,\alpha_{i+p}\in I_{u_i}$ and $\alpha_{j},\alpha_{j+p}\in I_{u_{j}}$, for which it holds that $\alpha_i\prec\alpha_{j}\prec\alpha_{i+p}\prec\alpha_{j+p}$. By the definition of the infimum and supremum, we have (i) $\inf I_{u_i} \preceq \alpha_i \prec \alpha_{j} \prec \alpha_{i+p}\preceq \sup I_{u_i}$ and (ii) $\inf I_{u_{j}}\preceq\alpha_{j}\prec\alpha_{i+p}\prec\alpha_{j+p}\preceq\sup I_{u_{j}}$. Therefore, $\max\{\inf I_{u_i},\inf I_{u_{j}}\}\preceq \alpha_{j}\prec\alpha_{i+p}\preceq\min\{\sup I_{u_i},\sup I_{u_{j}}\}$, which implies that the two intervals $\mathcal{I}(u_i)=(\inf I_{u_i},\sup I_{u_i})$ and $\mathcal{I}(u_{j})=(\inf I_{u_{j}},\sup I_{u_{j}})$ are overlapping.
        
        \item Since the alphabet is finite, the sequence $(\alpha_t)_{t\in \mathbb{N}}$ is infinite, increasing and monotone, from a sufficiently large $k\ge 0$, all strings $\alpha_{i+k'p}$ must share the last character. Therefore, w.l.o.g., we can assume that every $\alpha_t$ ends with the same character $c\in\Sigma$.
        For $i\in [p]$ and $v\in Q$, let 
        \[
            P_{i,v}
            :=\{
                t\in \mathbb{N}: t = i + kp \text{ for some } k\ge 0 \text{ and }\alpha_{t}=\alpha'c \text{ for some }\alpha'\in I_v 
            \} 
        \]
        be the set of string indices $t = i + kp$ in the given infinite sequence such that the string $\alpha_t$ can be read by reaching $u_i$ through a transition from $v$. Note that, for every $i \in [p]$, it holds that $P_i:=\bigcup_{v\in Q} P_{i,v} = \{t\in \mathbb{N} : t = i+kp \text{ for some }k\ge 0\}$, i.e., $P_i$ is an infinite set. Hence, there exists $u'_i$ such that $P_{i,u'_i}$ is an infinite set. Let $S':=\{u'_i: i \in [p] \}$. The set $S'$ obviously satisfies (i). Also (ii) must be satisfied due to determinism. We will now argue that (iii) holds, i.e., $S'$ is entangled. Consider the infinite sequence $(\alpha'_t)_{t\in\mathbb{N}}$ such that $\alpha'_t c=\alpha_t$ for every $t\in \mathbb{N}$. Observe that, for every $i \in [p]$, there are infinitely many $k\ge 0$ such that $\alpha'_{i+kp}\in I_{u'_i}$. 
        Hence, for every $i\in [p]$ and $t\in \mathbb{N}$, there exists $t'\ge t$ such that $\alpha_{t'}\in I_{u'_i}$ and hence there exists a monotone sequence as in Definition~\ref{def: entangle} and thus $S'$ is entangled.\qedhere
    \end{enumerate}
    
\end{proof}

We are now ready to prove Theorem~\ref{thm: main: width - cycle} using the above two results.

\begin{proof}[Proof of Theorem \ref{thm: main: width - cycle}]

$(\Leftarrow)$
Assume that there exists a non-empty string $\gamma$ such that $\mathcal A_{\min}$ contains $p$ pairwise distinct nodes $u_i$, $i \in [p]$, with $\delta(u_i, \gamma)=u_i$ for all $i\in [p]$ and $\mathcal I(u_i) \cap \mathcal I(u_j) \neq \emptyset$ for all $i,j\in [p]$. 
Note that $\inf I_{u_i}\preceq \gamma^\omega \preceq \sup I_{u_i}$ so it must either hold that (i) $\gamma^\omega\ne \inf I_{u_i}$ for all $i \in [p]$ or (ii) $\gamma^\omega\ne \sup I_{u_i}$ for all $i\in [p]$. Otherwise $\inf I_{u_i}\prec\sup I_{u_i}=\inf I_{u_j}\prec \sup I_{u_j}$ for some $i,j\in[p]$ implying $\mathcal{I}(u_i)\cap\mathcal{I}(u_j)=\emptyset$ by definition. 
Without loss of generality, let us assume that (i) holds, the case where (ii) holds is symmetric.
For $i\in [p]$, let $\alpha_i\in I_{u_i}$ be a string obtained by applying Lemma~\ref{lemma: langwidth: string greater than loop exists} for state $u_i$, respectively. 
Without loss of generality, we can assume that $\alpha_1\prec\alpha_2\prec\cdots\prec\alpha_p\prec\gamma^\omega$.
Note that, since $\mathcal{A}$ is a DFA, for every $i,j\in[p]$, it holds that $\alpha_i\neq\alpha_{j}$ for $i\neq j$.
Moreover, since $\gamma$ labels a cycle, we can assume $\gamma$ is sufficiently long such that, for every $i\in[p]$, (a) $\alpha_i\prec\gamma$ and (b) $\gamma$ is not a suffix of $\alpha_i$; it is worth noting that the condition (b) is important for the symmetric case where (ii) holds. This implies that $\alpha_p\prec\alpha_1\gamma$, and consequently, we obtain $\alpha_1\prec\alpha_2\prec\cdots\prec\alpha_p\prec\alpha_1\gamma$.

Observe that appending the same string $\gamma$ at the end of each of these strings does not affect their relative co-lex order. Therefore, we have that $\alpha_1\gamma^{k}\prec\cdots\prec\alpha_p\gamma^{k}
\prec \alpha_1\gamma^{(k+1)}\prec\cdots\prec\alpha_p\gamma^{(k+1)}$ for every integer $k\ge 0$, which yields an infinite sequence. Since $\alpha_i\gamma^{k}\in I_{u_i}$ for every $i\in[p]$ and $k\ge 0$, this infinite sequence shows that $\{u_1,\cdots,u_p\}$ is entangled. Now, from Theorem~\ref{thm: widthD L equals ent Amin} it follows that $\width^D(\mathcal L)\ge p$.

$(\Rightarrow)$ 
Since $\width^D(\mathcal L)\ge p$, from Theorem~\ref{thm: widthD L equals ent Amin} we know that there exists a set of entangled states of size at least $p$. Let $S_1=\{v_1,v_2,\cdots,v_p\}$ be such a set. For $k\in[n^p+1]$, let $S_{k+1}$ be an entangled predecessor of $S_k$ with $c_k\in\Sigma$ that can be obtained by Lemma~\ref{lemma: langwidth: entangled properties}~(\ref{lemma: langwidth: sub: entangled predecessor exists}). Note that for $k\ge 1$, it holds that (i) $|S_k|=|S_1|$ and (ii) $S_k$ is entangled by definition. Let $f_k(v_i)$ be the state $w\in S_{k}$ such that $v_i=\delta(w,\gamma_k)$ where $\gamma_k=c_{k-1}\cdots c_1$. Observing that $S_k=\{f_k(v_i):i\in[p]\}$ and $|S_k|=|S_1|$, for every $k\ge 1$, we have that $f_k$ is a one-to-one function. Now, for a fixed $k\in[n^p+1]$, consider the set $\{(f_k(v_1),\cdots,f_k(v_p)): k\in [n^p + 1]\}$. As the image of $f_k$ is a subset of the states in $\mathcal{A}_{\min}$, the number of such $p$-tuples is at most $n^p$. Therefore by the pigeonhole principle, there must exist two integers $1\le k'<k''\le n^p+1$ such that $(f_{k'}(v_1),\cdots,f_{k'}(v_p)) = (f_{k''}(v_1),\cdots,f_{k''}(v_p))$. Let $u_i=f_{k'}(v_i)=f_{k''}(v_i)$ and $\gamma=c_{k''-1}c_{k''-2}\cdots c_{k'}$. By definition, it holds that $u_i=\delta(u_i,\gamma)$ for all $i\in [p]$ and recall that $S_{k'}=\{u_i:i\in[p]\}$ is entangled. Therefore, by Lemma~\ref{lemma: langwidth: entangled properties}~(\ref{lemma: langwidth: sub: intervals of entangled states are overlapping}), it holds that $\mathcal I(u_i) \cap \mathcal I(u_{j}) \neq \emptyset$ for every distinct $i,j\in[p]$. 
\end{proof}

\section{Algorithms for the \textsc{DFAD\MakeLowercase{et}W\MakeLowercase{idth}} Problem}\label{sec:algorithm}
In Subsection \ref{sec: optimal algo} we present a simple algorithm for \textsc{DFADetWidth} running in $O(m^p)$ time and $O(n^p)$ space in the word RAM model with word size $w$, under the assumption that the working space does not exceed the model's space budget of $2^w$ words\footnote{Note that this restriction is due to the exponential working space of the algorithm. We leave it as an open problem to determine whether \textsc{DFADetWidth} can be solved in polynomial space and $O(m^p)$ time, for \emph{all} values of $p$.} (this is true, for instance, when $p$ is a constant). Under the Strong Exponential Time Hypothesis (SETH), this running time is optimal (even for constant $p$) by the results that we present in Section \ref{sec: hardness}.  
Then, in Subsection~\ref{sec: algo optimization} we optimize this algorithm by using randomization and further combinatorial observations.

\subsection{A Simple Optimal-Time Algorithm}\label{sec: optimal algo}

Our algorithm is based on Theorem~\ref{thm: main: width - cycle}:
by that theorem, we have to look for $p$ cycles $\delta(u_i,\gamma)=u_i$ in $\mathcal A_{\min}$ (for pairwise distinct $u_1, \dots, u_p$) such that the co-lex intervals $\mathcal I(u_i)$ of their starting states pairwise intersect. 
We can easily achieve this as follows. Given an input DFA $\mathcal A$ with $n$ states and $m$ transitions:
\begin{itemize}
    \item[(1)] We build $\mathcal A_{\min}$ ($O(m\log n)$ time by Hopcroft's algorithm~\cite{hopcroft1971n}), 
    \item[(2)] we compute $\mathcal I(u)$ for each state $u$ of $\mathcal A_{\min}$ ($O(m\log n)$ time~\cite{BeckerCCKKOP23}; read below for more details), and 
    \item[(3)] we test acyclicity of the semi-DFA $\mathcal B$ described below, via a DFS visit.
\end{itemize}

The semi-DFA $\mathcal B$ is a pruned version of the power semi-DFA $\mathcal A_{\min}^p$ (see Definition~\ref{def:power semiDFA}), such that $\mathcal B$ has a cycle if and only if the conditions of Theorem~\ref{thm: main: width - cycle} are satisfied:

\begin{definition}\label{def:semi-DFA B}
    Let $\mathcal A_{\min} = (Q, \Sigma, \delta, s, F)$ be a minimum DFA. We define the following semi-DFA $\mathcal B = (Q',\Sigma, \delta)$:
    \begin{itemize}
        \item $Q' = \{\bar Q \subseteq Q\ :\ |\bar Q| = p \wedge (\forall u,v\in \bar Q)(\mathcal I(u) \cap \mathcal I(v) \neq \emptyset) \}$
        \item We overload notation with respect to $\delta$ by extending it to sets (see also Section \ref{sec:DFAs}): for $\{u_1, \dots, u_p\}, \{v_1, \dots, v_p\} \in Q'$ and $a\in \Sigma$, we define $\delta(\{u_1, \dots, u_p\}, a) = \{\delta(u_1,a), \dots, \delta(u_p,a)\}$.
    \end{itemize}
\end{definition}

See Figure~\ref{fig:langwidth:unordered} for two examples of $\mathcal{B}$ constructed from the automaton $\mathcal{A}_{\min}$ of Figure~\ref{fig:dfawidth} for $p=2$ and $p=3$, respectively.
By the very definition of $\mathcal B$, this semi-DFA has a cycle labeled with a string $\gamma$ if and only if $\mathcal A_{\min}$ contains $p$ cycles $\delta(u_i,\gamma)=u_i$ 
(for pairwise distinct $u_1, \dots, u_p$) such that the co-lex intervals $\mathcal I(u_i)$ of their starting states pairwise intersect. 
The DFS visit of $\mathcal B$ allows us to detect cycles in this semi-DFA, thereby allowing us to determine whether the conditions of Theorem \ref{thm: main: width - cycle} are satisfied. Below, we discuss steps (2) and (3) in more detail and analyze the algorithm's complexity. 

\begin{figure}[ht!]
    \centering
    \subfloat[$p=2$]{
        \resizebox{0.3\textwidth}{!}{
            \begin{tikzpicture}[
        scale=0.8,
        x=2.0cm,y=2.0cm,
        v/.style={circle, draw=black, minimum size=10mm},
    ]
    
    \node[v]  (v23) at (0.0, 2.2) {$2,3$};
    \node[v]  (v24) at (0.0, 0.8) {$2,4$};
    \node[v]  (v46) at (0.8, 0.0) {$4,6$};
    \node[v]  (v36) at (2.2, 0.0) {$3,6$};
    \node[v]  (v25) at (3.0, 0.8) {$2,5$};
    \node[v]  (v45) at (3.0, 2.2) {$4,5$};
    \node[v]  (v26) at (2.2, 3.0) {$2,6$};
    \node[v]  (v56) at (0.8, 3.0) {$5,6$};
    
    \path [-stealth, thick]
        (v24) edge [bend left, left] node {1} (v23)
        (v25) edge [bend right, right] node {0} (v45)
        (v25) edge [bend left, below] node {1} (v36)
        (v26) edge [bend right, above] node {1} (v56)
        (v45) edge [bend right, right] node {1} (v26)
        (v46) edge [bend right, right] node {1} (v23)
        (v56) edge [bend left, right] node {0} (v46)
        (v23) edge [bend left, above] node {0} (v56)
        ;
\end{tikzpicture}
        }
    }
    \hspace{15mm}
    \subfloat[$p=3$]{
        \resizebox{0.23\textwidth}{!}{
            \begin{tikzpicture}[
        x=2.0cm,y=2.0cm,
        v/.style={circle, draw=black, minimum size=10mm},
    ]
    
    \node[v]  (v245) at (0,1.2) {$2,4,5$};
    \node[v]  (v236) at (1,2.0) {$2,3,6$};
    \node[v]  (v456) at (2,1.2) {$4,5,6$};
    \node[v]  (v256) at (1.7,0) {$2,5,6$};
    \node[v]  (v246) at (0.3,0) {$2,4,6$};
    
    \path [-stealth, thick]
        (v245) edge [above] node {1} (v236)
        (v456) edge [above] node {1} (v236)
        (v256) edge [right] node {0} (v456)
        ;
\end{tikzpicture}
        }
    }
    
    \caption{
    Pruned power semi-DFA $\mathcal B$ of Definition \ref{def:semi-DFA B} constructed from the automaton $\mathcal{A}_{\min}$ of Figure~\ref{fig:dfawidth} for $p=2$ (a) and $p=3$ (b). By  Theorem~\ref{thm: main: width - cycle} 
    and by Definition \ref{def:semi-DFA B},
    we obtain that $\width^D(\mathcal{L}(\mathcal{A}_{\min}))=2$ since (a) contains a cycle (i.e.,  $\width^D(\mathcal{L}(\mathcal{A}))\ge 2$) while (b) is acyclic (i.e., $\width^D(\mathcal{L}(\mathcal{A}_{\min}))<3$).}
\label{fig:langwidth:unordered}
\end{figure}
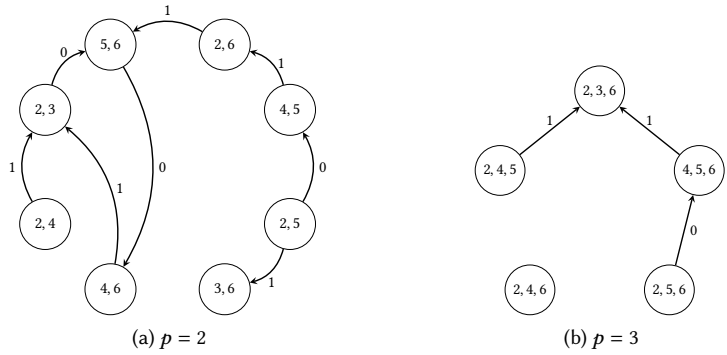

\paragraph{Computing the intervals $\mathcal I(u)$}

In order to compute the intervals $\mathcal I(u)$ for all states $u$ in $\mathcal A_{\min}$, we use the partition refinement algorithm described by Becker et al.~\cite{BeckerCCKKOP23}, running in time $O(m \log n)$. 
Observe that the set $\{\inf I_u, \sup I_u\ :\ u\in Q\}$ contains at most $2n$ elements; Becker et al.'s algorithm ~\cite{BeckerCCKKOP23} implicitly sorts these strings and represents each of them with their rank --- an integer in the range $[2n]$ --- in this sorted list.
For simplicity, we do not introduce new notation 
for the intervals $\mathcal I(u)$: in what follows, they have to be interpreted as pairs over $[2n]$ as above described. We remark that Becker et al.'s algorithm ~\cite{BeckerCCKKOP23} assumes two additional properties for $\mathcal A_{\min}$: (i) input-consistency (i.e.\ each incoming edge of any given state is required to have the same label), and (ii) the source state does not have incoming transitions. While $\mathcal A_{\min}$ may not necessarily satisfy these conditions, the properties can be restored easily by adding a new dummy source state and removing incoming transitions with non-minimum (or non-maximum) labels. We give a detailed account of this transformation of $\mathcal A_{\min}$ in Appendix~\ref{sec: appendix: infsup}.

\paragraph{Testing acyclicity of $\mathcal B$}

We test if $\mathcal B$ is acyclic by DFS-visiting it \emph{on-the-fly}, i.e.\ without explicitly building it. 
Let us denote $\mathcal A_{\min} = (Q, \Sigma, \delta, s, F)$.
By the folklore DFS-based algorithm for testing acyclicity \cite{leiserson1994introduction}, this step requires us to keep a color vector $C : Q' \rightarrow \{black, white, gray\}$, for each state $q \in Q'$ of $\mathcal B$ (white nodes have not been visited, gray nodes have been entered but not yet exited, and black nodes have been entered and exited). 
We implement $C$ with direct addressing. First of all, without loss of generality we may assume that $Q = [n']$, with $n' = |Q| \le n$; this can be achieved with a simple pre-processing of $\mathcal A_{\min}$.
Given any $q \in Q'$, where $q = \{u_1, u_2, \dots, u_p\} \subseteq Q$ ($q$ is not necessarily given in sorted order), 
the entry of $C$ associated with $q$ is obtained by sorting $q$ in $O(p\log p)$ time (without loss of generality, $u_1 < u_2 < \dots <  u_p$) and using the tuple $(u_1, u_2, \dots, u_p)$ as an address in memory\footnote{Notice that this is possible since in our word-RAM model we assume (see Section \ref{sec:RAM model}) that the total working space $O(n^p)$ of our algorithm must not exceed the model's space budget of $2^w$. In particular, this implies that the tuple $(u_1, u_2, \dots, u_p)$ fits in one memory word and can be used as an address.}.
Array $C$ uses in total $O(n^p)$ words of space. We initialize with color \emph{white} all the cells of $C$ in $O(n^p)$ time.

At each step of the DFS visit, we are visiting some node $q = \{u_1, u_2, \dots, u_p\}$ of $\mathcal B$. By keeping $\mathcal A_{\min}$ in adjacency-list format, 
we can iterate over all the candidate successors of $q$ (each identified by a combination of $p$ outgoing transitions from $u_1, u_2, \dots, u_p$) by keeping $p$ pointers inside the $p$ adjacency lists of $u_1, u_2, \dots, u_p$ (i.e.\ the $i$-th pointer points to a successor $\delta_{c_i}(u_i)$ of $u_i$, for some $c_i \in \Sigma$).
Let $u_1', u_2', \dots, u_p'$ be the successors of $u_1, u_2, \dots, u_p$, respectively, identified by those pointers. We push state $\{u_1', u_2', \dots, u_p'\}$ on the DFS stack if and only if $u_1', u_2', \dots, u_p'$ (1) are pairwise distinct, (2) are reached by the same label from $u_1, u_2, \dots, u_p$, and (3) have pairwise-intersecting co-lex intervals $\mathcal I(u_1'), \mathcal I(u_2'), \dots, \mathcal I(u_p')$. These three conditions can be easily checked in $O(p\log p)$ time by sorting $u_1', u_2', \dots, u_p'$ and their co-lex intervals by their first coordinate. 

Note that, while this procedure spends $O(p\log p)$ time also in cases where $\{u_1', u_2', \dots, u_p'\}$ is not a state of $\mathcal B$, the total number of tuples $(u_1', u_2', \dots, u_p')$ 
that we test is upper-bounded by $\binom{m}{p}$. It follows that, overall, testing all these tuples takes time $O\big(p\log p\cdot \binom{m}{p}\big) \subseteq O(m^p)$. 
Note also that, along with each node $q = \{u_1, u_2, \dots, u_p\}$ of $\mathcal B$ pushed on the DFS stack ($O(p)$ space), we need to store its associated $p$ pointers inside the adjacency lists of $u_1, u_2, \dots, u_p$ (in order to continue the visit with the next candidate successor of $q$). Since $\mathcal B$ has at most $\binom{n}{p}$ states, the total stack size never exceeds $O\big(p\cdot \binom{n}{p}\big) \subseteq O(n^p)$ words. 
In total, we spend time $O\big(p\log p \cdot \binom{n}{p}\big) \subseteq O(n^p) \subseteq O(m^p)$ operating on the stack.

We obtain: 

\begin{theorem}\label{thm:simple algorithm}
    In the word-RAM model with word size $w$, we can solve \textsc{DfaDetWidth} with parameter $p \geq 2$ on any DFA $\mathcal A$ with $n$ states and $m$ transitions 
    in $O(m^p)$ time and $O(n^p)$ words of working space, provided that the working space does not exceed the space budget $2^w$ of the model.
\end{theorem}


\subsection{Optimized Algorithm}\label{sec: algo optimization}

While in Theorem \ref{thm:simple algorithm} we used the simple upper-bound $\binom{n}{p}$ for the number of states of the semi-DFA $\mathcal B$ of Definition \ref{def:semi-DFA B}, we now observe that the size of $\mathcal B$ could be much smaller than that. This is due to the fact that every state $\{u_1, \dots, u_p\}$ of $\mathcal B$ must be such that the intervals $\mathcal I(u_1), \dots, \mathcal I(u_p)$ pairwise overlap. Now, recall (Corollary \ref{cor: intervals vs colex width}) that 
the width of $\mathcal A_{\min}$ upper-bounds the size of the largest subset of $\{\mathcal I(u)\ :\ u\in Q\}$ whose members pairwise overlap; as a result, the size of $\mathcal B$ depends on $\width(\mathcal A_{\min})$. 
By using Karp-Rabin hashing \cite{KR} (see also Section \ref{subsec: randomization}), we can furthermore represent any subset of $p$ states of $\mathcal A_{\min}$ (forming one state of $\mathcal B$) in $O(1)$ memory words, with only a low probability $2^{-\Omega(w)}$ of collisions. 
Then, $\mathcal B = (Q',\Sigma, \delta)$ can be represented in adjacency-list format as a map (a dictionary associating every state to its adjacency list) supporting constant-time lookups and insertions w.h.p. \cite{Dietzfelbinger90}.
Ultimately, these observations lead to an optimized algorithm whose running time is bounded by the size of $\mathcal B$ --- which we show to be $O\big(m\cdot \binom{\width(\mathcal A_{\min})}{p}\big)$ --- plus the time for minimizing the input DFA.
In the most favorable scenario --- i.e., when $\mathcal A_{\min}$ is a Wheeler DFA or, more in general, if $\width(\mathcal A_{\min}) \in O(\log\log n)$ --- our algorithm terminates in time $O(m\log n)$ with high probability. 

Let us formalize these intuitions. 
Let $\mathcal A_{\min} = (Q, \Sigma, \delta, s, F)$, and let $G(T) = (Q,E)$ be the intersection graph of the intervals $T=\{\mathcal I(u)=(\inf I_u,\sup I_u)\ :\ u\in Q\}$, i.e.\ the graph with edge set $E = \{(u,v)\in Q^2\ :\ u\neq v \wedge \mathcal I(u) \cap \mathcal I(v) \neq \emptyset\}$.
By the very definition of $\mathcal B$, the number of nodes of $\mathcal B$ is equal to the number of pairwise distinct $p$-cliques in $G(T)$.

A similar reasoning holds for the number of transitions of $\mathcal B$.
For each $a\in\Sigma$ let $T_a = \{\mathcal I(u)\ :\ u\in Q \wedge \delta(u,a)\in Q\}$ be the set of intervals of all the states
of $\mathcal A_{\min}$
having an outgoing transition labeled with character $a$. 
Consider the intersection graph $G(T_a)$ defined by $T_a$, and let us fix some notation: 

\begin{itemize}
    \item let $n_a = |\{u\in Q\ :\ \delta(u,a)\in Q\}|$ be the number of states of $\mathcal A_{\min}$ having an out-going transition labeled with $a\in \Sigma$,
    \item $N_a$ be the number of subsets $S\subseteq Q$ of cardinality $|S|=p$  such that (i) the intervals $\mathcal I(u)$ of states $u\in S$ mutually overlap (in particular, $S$ is a state of $\mathcal B$), and (ii) $\delta(u,a)\in Q$ for each $u\in S$. Observe that, for such $S$, the triple  $(S,\delta(S,a),a)$ is a transition of $\mathcal B$ if and only if $|\delta(S,a)|=p$.
    \item $M = \sum_{a\in \Sigma} N_a$.
\end{itemize}

Since $\mathcal A_{\min}$ is a DFA, observe that $\sum_{a\in \Sigma}n_a = m$. 
By definition, observe also that $N_a$ upper-bounds the number of transitions of $\mathcal B$ labeled with character $a$, and $M$ upper-bounds the total number of transitions of $\mathcal B$.

\begin{remark}
    By the very definition of the interval graph $G(T_a)$, observe that $N_a$ is equal to the number of $p$-cliques in $G(T_a)$.  On the other hand, by the definition of $N_a$ it follows that the number of transitions labeled by $a$ in $\mathcal B$ is upper-bounded by $N_a$, and that the total number of transitions in $\mathcal B$ is upper-bounded by $M = \sum_{a\in \Sigma} N_a$.
\end{remark}

We now upper-bound the number of $p$-cliques in each graph $G(T_a)$, as a function of its size and the size of its largest clique.

\begin{lemma}\label{lem: number of cliques in interval graph}
    Let $G = (V,E)$ be an interval graph with $t = |V|$ vertices, each vertex $u\in V$ corresponding to an open interval $(\ell_u, r_u)$ with $\ell_u<r_u$. Moreover, assume that the maximum clique of $G$ is of size $p'$. Then, the number of cliques of size $p\le p'$ in $G$ is at most $t \cdot \binom{p'-1}{p-1}$.
\end{lemma}
\begin{proof}
    We first note that, for every intersection graph over open intervals $\{(\ell_u, r_u):u\in V\}$ can be represented as the intersection graph over left-closed intervals $\{[\ell_u, r_u):u\in V\}$. To see this, observe that it holds for every $u,v(\ne u)\in V$ $(\ell_u, r_u)\cap(\ell_v,r_v)\ne\emptyset\Leftrightarrow r_u> \ell_v\Leftrightarrow [\ell_u, r_u)\cap[\ell_v,r_v)\ne\emptyset$. Hence, in the following, we assume to be given an interval graph over left-closed intervals of integers.
    
    We partition the vertices in $G$ into $p'$ \emph{chains} $P_1, \ldots, P_{p'}(\subseteq V)$ such that for every $i\in[p']$, all distinct vertices $u,v(\ne u)\in P_i$ in the same chain have disjoint intervals, i.e., $[\ell_u, r_u)\cap[\ell_v,r_v)=\emptyset$. It is clear that such a partition exists, take for example the $p'$ chains of a smallest chain decomposition of a maximal co-lex order on the states of $\mathcal A_{\min}$. Since we want to obtain an upper bound of the number of cliques of size $p$ in $G$, we can assume the following. 
    \begin{enumerate}
        \item We assume each chain $P_i$ completely covers the interval $[\min_{u\in V} \ell_u, \max_{u\in V} r_u)$. Otherwise, we can enlarge all intervals arbitrarily while keeping the intervals in the same chain not overlapping. Note that this enlarging transform can only add edges to $G$ and thus the number of cliques of size $p$ can only be increased. 
        \item We also assume that all the left endpoints of the intervals are distinct except the leftmost points at $\min_{u\in V} \ell_u$, i.e., the cardinality of the set $\{\ell_u:u\in V\}$ of all left endpoints is exactly $t-p'+1$. Otherwise, we can subtract each left endpoint $\ell_u$ by some value $\epsilon_u\ge 0$ (subtract also from the right endpoint of the adjacent interval on its left in the same chain to keep the previous assumption hold) in such a way that the edges in $G$ are preserved; i.e., we apply a perturbation that can only adds edges to $G$. Note that such values $\epsilon_u$'s always exist since we can choose arbitrarily small values.
    \end{enumerate}
    Now there is a one-to-one correspondence between cliques of size $p$ in $G$ and the sets of all $p$ intervals from $p$ out of the $p'$ chains that have a pairwise non-empty intersection. Let $x_0<x_1<\cdots<x_{t-p'+1}$ be the $t-p'+1$ pairwise distinct left endpoints. We count the maximal number of such $p$ intervals by counting them for each of the range $[x_i, x_{i + 1})$, with $i=0,\ldots, t - p'$, separately (here $x_{t - p' + 1}:=\max_{u\in V} r_u$). In the range $[x_0, x_1)$, we can find $\binom{p'}{p}$ many sets of overlapping sets of intervals of size $p$. In the range $[x_i, x_{i+1})$ for $i \in [t - p']$ instead,
    we can find $\binom{p'- 1}{p - 1}$ new cliques of size $p$ in $G$ that we have not seen before, namely exactly those cliques that result from combining the \emph{new vertex} (the vertex whose left endpoint is $x_i$) with any subset of size $p' - 1$ of the other $p - 1$  vertices whose intervals contain the point $x_i$. Hence, there are at most
    $
        \binom{p'}{p} + (t - p') \cdot \binom{p'- 1}{p - 1}
        \le t \cdot \binom{p'- 1}{p - 1}
    $
    cliques of size $p$ in $G$.
\end{proof}

By definition, each graph $G(T_a)$ has $n_a$ nodes. 
Moreover, let $p' = \width(\mathcal A_{\min})$. Then, by Corollary \ref{cor: intervals vs colex width}, the largest clique of $G(T_a)$ has size at most $p'$. 
If $p\leq p'$, we apply Lemma \ref{lem: number of cliques in interval graph} with $t = n_a$ and obtain that $G(T_a)$ has at most $n_a\binom{p'-1}{p-1}$ cliques. We conclude that $M=\sum_{a\in\Sigma}N_a\le \sum_{a\in\Sigma}n_a\binom{p'-1}{p-1}= m\binom{p'-1}{p-1}$. 
If, on the other hand,  $p > p'$, then there is no subset of $p$ states $u_1\, \dots, u_p$ of $\mathcal A_{\min}$ whose intervals $\mathcal I(u_1), \dots, \mathcal I(u_p)$ pairwise overlap, thus $\mathcal B$ is empty (i.e.\ it has zero states and zero transitions). 
To sum up: 

\begin{corollary}\label{cor: size of B}
    If $p \leq \width(\mathcal A_{\min})$, then $\mathcal B$ has at most $M \leq m\binom{\width(\mathcal A_{\min})-1}{p-1}$ transitions, where $m$ is the number of transitions in $\mathcal A_{\min}$. Otherwise ($p > \width(\mathcal A_{\min})$), $\mathcal B$ is empty. 
\end{corollary}

\subsubsection{Algorithm description}

We proceed by describing our final procedure: Algorithm \ref{algo: optimized algo}. The algorithm starts in Lines \ref{line: hopcroft} and \ref{line: intervals} by minimizing the input DFA $\mathcal A$ (thereby obtaining $\mathcal A_{\min} = (Q,\delta, \Sigma, s, F)$) and computing the co-lex intervals $\mathcal I(u)$ for each state of $\mathcal A_{\min}$, using Becker et al.'s procedure~\cite{BeckerCCKKOP23} (described in Section \ref{sec: optimal algo} - paragraph \emph{Computing the intervals $\mathcal I(u)$}). These steps run in $O(m\log n)$ time --- where $m$ and $n$ are the number of transitions and states, respectively, of $\mathcal A_{\min}$.

In Line \ref{line: init B} we initialize the (initially empty) data structure $\mathcal{\bar B}$ implementing the semi-DFA $\mathcal B$ of Definition \ref{def:semi-DFA B}.
$\mathcal{\bar B}$ is a map (dictionary) of adjacency lists $\mathcal{\bar B}  : \kappa(Q') \rightarrow (\kappa(Q') \times \Sigma)^*$ supporting lookups and insertions in constant time w.h.p.  \cite{Dietzfelbinger90}. 
Here, $\kappa$ denotes the Karp-Rabin hash function (see Section \ref{subsec: randomization}), $Q' \subseteq 2^Q$ is the set of states of $\mathcal B$ (see  Definition \ref{def:semi-DFA B}), and $\kappa(Q') = \{\kappa(T)\ :\ T\in Q'\}$ denotes the set of Karp-Rabin fingerprints of $\mathcal B$'s states. 
In other words, $\mathcal{\bar B}$ represents states of $\mathcal B$ compactly as their Karp-Rabin fingerprints. Ultimately, this allows us to speed up by a factor of $p$ (i.e.\ the cardinality of any state of $\mathcal B$) any operation involving manipulating a state/transition of $\mathcal{\bar B}$, at the cost of introducing randomization. 
Given state $T \in Q'$ of $\mathcal B$, after construction (Lines \ref{line:for each char}-\ref{line:power remove u2}) $\mathcal{\bar B}[\kappa(T)]$ will store the adjacency list of state $T$:  with high probability, $(\kappa(T'),a)\in \mathcal{\bar B}[\kappa(T)]$ if and only if $\mathcal B$ contains a transition $\delta(T,a)=T'$.
Importantly, as discussed in detail below, we are also able to build 
$\mathcal{\bar B}$ in $O(1)$ time per edge (with high probability of success).

After initializing $\mathcal{\bar B}$, in Line \ref{line: init KR exponents} we pre-compute $(x^i \mod q)$ for each $i\in [|Q|]$ in $O(|Q|) = O(n)$ time; here, $q\in \Theta(2^w)$ and $x\in [q]$ are the global parameters chosen for Karp-Rabin hashing, see also Section \ref{subsec: randomization}.

Lines \ref{line:for each char}-\ref{line:power remove u2} describe how to build $\mathcal{\bar B}$. For each $a\in\Sigma$, we build all transitions of $\mathcal B$ labeled with character $a$, as follows. 
First of all, in Line \ref{line: radixsort} we collect the co-lex intervals $\mathcal I(u) = (\ell_u,r_u)$ of all states $u$ with $\ell_u\ne r_u$ and having an out-going transition labeled with $a$, i.e.\ such that $\delta(u,a)\in Q$, and then re-arrange them into  triples of the form 
$(\ell_u, 1, u)$, $(r_u, 0, u)$, which are then lexicographically sorted in a list $L_a$. While we did not write this explicitly in the pseudocode in order to keep it simple, the set of all sorted lists 
$\{L_a\ :\ a\in \Sigma\}$ 
can be pre-computed in $O(m)$ time  immediately after the computation of $\mathcal A_{\min}$ in Line~\ref{line: hopcroft} and of the set $\{\mathcal{I}(u)\}$ in Line~\ref{line: intervals}, by radix-sorting the set of quadruples 
$\{(a,\ell_u, 1, u)$, $(a,r_u, 0, u)\ :\ a\in \Sigma, u\in Q, (\ell_u,r_u) = \mathcal I(u), \ell_u\ne r_u, \delta(u,a)\in Q \}$ 
and then, while visiting this sorted list of quadruples from the first to last, appending  $(x,b,u)$ to $L_a$ for each quadruple $(a,x,b,u)$.
We briefly discuss the purpose and meaning of the three components of the triples $(\ell_u, 1, u)$, $(r_u, 0, u)$.
As far as the first component is concerned, recall that $\ell_u$ and $r_u$ are  the ranks of $\inf I_u$ and $\sup I_u$ in the co-lexicographically sorted set $\{ \inf I_u, \sup I_u\ :\ u\in Q\}$, of size at most $2n$. These integers therefore satisfy $\ell_u, r_u \in [2n]$.
The second component (either the integer 0 or 1) of those triples has the role of handling cases where $r_u = \ell_v$ (i.e.\ $\sup I_u = \inf I_v$), for some pair of states $u\neq v$; in that case, these integers force the triple $(r_u, 0, u)$ to appear before $(\ell_v, 1, v)$ in the sorted list. This will, in turn, allow us to ``close'' the interval associated with $u$ before ``opening'' the interval associated with $v$ while scanning the sorted list of triples, thereby implementing the fact that intervals are open-ended.
Finally, the third component of those triples is the state $u \in [n]$ associated with the triple; this is needed later in order to reconstruct the sets of states of cardinality $p$ whose co-lex intervals mutually overlap.

The high-level intuition behind the following Lines \ref{line: new S}-\ref{line:power remove u2} is simple. Assume we have processed triples 
$L_a[1,i-1]$ and that we are about to process triple $L_a[i]$, for $i=1, \dots, n_a$. 
We maintain a set  $S$ containing all states whose intervals have been opened but not yet closed in 
$L_a[1,i-1]$ (in particular, whose intervals mutually overlap), i.e.: 
$S = \{u\ :\ (\ell_u,1,u)\in L_a[1,i-1] \wedge (r_u,0,u)\notin L_a[1,i-1]\}$. 
Set $S$ is implemented with a data structure whose details are described below; here, we first describe the way $S$ is used as a black-box in order to build the transitions of $\mathcal B$.
We show how to update and use $S$, according to the content of $L_a[i]$.
If $L_a[i] = (r_u,0,u)$ for some $u\in Q$ (Line \ref{line:power remove u2}), then we are about to close the interval associated with state $u$. We need to perform only one simple action: remove $u$ from $S$. As described in the next subsection, this operation takes constant time w.h.p.
If, on the other hand, $L_a[i] = (\ell_u,1,u)$ for some $u\in Q$ (Line \ref{line: if open interval}), then we are about to open the interval associated with state $u$. Before inserting $u$ into $S$ (Line \ref{line: insert}), we enumerate all subsets $S'$ of $S$, of cardinality $|S'|=p-1$ and update structure $\mathcal{\bar B}$ with the transition of  $\mathcal B$ starting from state  $S'\cup\{u\}$, labeled with $a$, and landing in state $\delta(S'\cup\{u\},a)$, if and only if $\delta(S'\cup\{u\},a)$ is indeed a state of  $\mathcal B$, i.e.\ if  $|\delta(S'\cup\{u\},a)|=p$. If this condition is satisfied, then this update of $\mathcal{\bar B}$ is performed by the insertion of Line \ref{line: insert}: $\mathcal{\bar B}[\kappa_1].$\textsc{append}$((\kappa_2, a))$, where $\kappa_1 = \kappa(S'\cup\{u\})$ and $\kappa_2 = \kappa(\delta(S'\cup\{u\},a))$. 
The (slightly) more challenging part is to show how to implement $S$ so that it allows enumerating all pairs of Karp-Rabin fingerprints $(\kappa_1,\kappa_2)$ satisfying the above condition (i.e.\ $|\delta(S'\cup\{u\},a)|=p$), in constant time w.h.p. each. 
The data structure for $S$ offering the above-described interface with the the claimed running times is described in Subsection \ref{sec:structure S}. To conclude, after $\mathcal{\bar B}$ has been built, a simple DFS visit can detect acyclicity in time proportional to $\mathcal B$'size (Line \ref{line:acyclic test}).

\begin{algorithm2e}[ht!]
	\caption{\textsc{DeterministicLanguageWidth}($\mathcal A, p$) --- optimized version}
	\label{algo: optimized algo}
	\Input{DFA $\mathcal A$ and integer $p\geq 2$.}
	\Output{``\emph{\textsc{Yes}}'' if $\width^D(\mathcal L(\mathcal A)) < p$, ``\emph{\textsc{No}}'' otherwise.}
    \medskip

    $\mathcal A_{\min} = (Q, \Sigma, \delta, s, F) \gets \textsc{Minimize}(\mathcal A)$\tcp*{Run Hopcroft's algorithm}\label{line: hopcroft}

    $\{\mathcal I(u)\ : u\in Q\} \gets \textsc{ComputeIntervals}(\mathcal A_{\min})$\tcp*{Compute co-lex intervals}\label{line: intervals}

    $\mathcal{\bar B} \gets \emptyset$ \tcp*{Init semi-DFA $\mathcal{\bar B}$: a map $\mathcal{\bar B}  : \kappa(Q') \rightarrow (\kappa(Q') \times \Sigma)^*$ of adjacency lists}\label{line: init B}

    Pre-Compute $(x^i \mod q)$ for each $i\in [|Q|]$\tcp*{$x, q$: the parameters used in Karp-Rabin hashing}\label{line: init KR exponents}

    \BlankLine

    \For{$a \in \Sigma$}{\label{line:for each char}

        $L_a\gets $ \textsc{Sort}( $\{(\ell_u,1,u), (r_u,0,u)\ :  u \in Q \wedge (\ell_u, r_u) = \mathcal I(u) \wedge \ell_u\ne r_u \wedge \delta(u,a)\in Q\}$ )\;\label{line: radixsort}
            
        $S\gets \emptyset$ \tcp*{current queue of co-lex overlapping states}\label{line: new S}

        \For{$(x, b, u)\in L_a$\label{line: for triples}}{
        \uIf{$b=1$}{\label{line: if open interval}
                \BlankLine
                \tcp*[l]{The following loop runs in $O(1)$ time w.h.p. per iteration (see description)}
                \For{\emph{\textbf{each}} 
                $\Big(\kappa(S'), \kappa(\delta(S',a))\Big)$ \emph{\textbf{such that}}
                $S'\subseteq S \wedge |S'|=p-1$\label{line:enumerate pairs}}{
                
	                \If{ $|\delta(S'\cup\{u\},a)|=p$\label{line: check cardinality p}}{
	                	   
		               $\kappa_1 \gets \kappa(S') + x^u \mod q$ \tcp*{Karp-Rabin hash of $S'\cup\{u\}$}\label{line: compute KR1}
		
		                $\kappa_2 \gets \kappa(\delta(S',a)) + x^{\delta(u,a)} \mod q$\tcp*{Karp-Rabin hash of $\delta(S'\cup \{u\},a)$ }\label{line: compute KR2}
		                $\mathcal{\bar B}[\kappa_1].$\textsc{append}$((\kappa_2, a))$\tcp*{Update adjacency list of $S'\cup\{u\}$}\label{line: append}
	                }
                }
                $S.$\textsc{insert}$(u)$\;\label{line: insert}
        }
        \Else{
        $S.$\textsc{delete}$(u)$\label{line:power remove u2}
        }
    }
        
    }

    \If{\textsc{Acyclic}$(\mathcal{\bar B})$}{\label{line:acyclic test}
        \Return{\textsc{Yes}}
    }

    \Return{\textsc{No}}

\end{algorithm2e}

\subsubsection{Running time}\label{sec: running time}

As mentioned above, Lines \ref{line: hopcroft}-\ref{line: init KR exponents} run in $O(m\log n)$ time, where $m$ and $n$ are the number of transitions and states of the input DFA $\mathcal A$, respectively. Above, we have also shown that all calls to the sorting step in Line \ref{line: radixsort} cost overall $O(m)$ time. For a given $a\in \Sigma$ chosen in Line \ref{line:for each char}, the total number of iterations of the \texttt{for each} loop in Line \ref{line:enumerate pairs} is $N_a$ since we are enumerating all sets $S'\cup \{u\}$ of cardinality $|S'\cup \{u\}|=p$ such that (i) all states $v\in S'\cup \{u\}$ have an out-going label $a$ (i.e. $\delta(v,a)\in Q$), and (ii) the co-lex intervals $\mathcal I(v)$ of all  $v\in S'\cup \{u\}$ overlap. It follows that, overall, the \texttt{for each} loop in Line \ref{line:enumerate pairs} runs $M = \sum_{a\in\Sigma} N_a$ times. Since, as we show below, set $S$ can be implemented with a data structure supporting all operations involved in Lines \ref{line:enumerate pairs}-\ref{line: append} in $O(1)$ time w.h.p. for each loop iteration, overall the \texttt{for each} iterations cost $O(M)$ time w.h.p. In the next subsection we also show that the data structure for $S$ supports the operations in Lines \ref{line: insert} and \ref{line:power remove u2} in $O(1)$ time w.h.p. Overall, these operations cost therefore $O(n)$ time w.h.p. Finally, testing acyclicity of $\mathcal B$ in Line \ref{line:acyclic test} costs linear time w.h.p. in its size (at most $M$). 
By Corollary \ref{cor: size of B}, we conclude that Algorithm \ref{algo: optimized algo} runs in $O(m\log n + M) \subseteq O\left(m\log n + m\binom{\width(\mathcal A_{\min})-1}{p-1}\right)$ time w.h.p. and $O(m + M) \subseteq O\left(m + m\binom{\width(\mathcal A_{\min})-1}{p-1}\right)$ words of space.

\subsubsection{Data structure for set $S$}\label{sec:structure S}

We now describe how to implement set $S$ used by Algorithm \ref{algo: optimized algo}.
In the rest of the section, character $a$, chosen in Line \ref{line:for each char}, is a global constant: note that we initialize a new set $S$ every time character $a$ changes.

Algorithm \ref{algo: optimized algo} requires the following operations to be supported by $S$:

\begin{enumerate}
	\item Initialization: $S\gets \emptyset$ (Line \ref{line: new S}).
	\item Insertion of an element not already belonging to $S$: $S.$\textsc{insert}$(u)$ (Line \ref{line: insert}).
	\item Deletion of an element belonging to $S$: $S.$\textsc{delete}$(u)$ (Line \ref{line:power remove u2}).
	\item Given a state $u\notin S$,  enumerate all pairs of fingerprints of the form $\Big(\kappa(S'), \kappa(\delta(S',a))\Big)$, for all subsets
	$S'\subseteq S$ of $S$ such that $|S'|=p-1$ (Line \ref{line:enumerate pairs}).
	\item For each pair $\Big(\kappa(S'), \kappa(\delta(S',a))\Big)$ enumerated in Operation (4), evaluate the condition $|\delta(S'\cup\{u\},a)|=p$ (Line \ref{line: check cardinality p}).
\end{enumerate}

In order for the running times claimed in Section \ref{sec: running time} to hold, 
operations (1-3) have to be supported in $O(1)$ time w.h.p. Operations (4-5) have to be supported in $O(1)$ time w.h.p. for each enumerated pair. We now show how to achieve this.

\paragraph{Intuition}
Clearly, the challenging operations are (4-5).
Intuitively (read below for a more formal description), we manage to support Operations (4-5) in constant time per enumerated pair by maintaining the elements of $S$ in a priority queue, sorted by time of insertion in $S$. 
Let $RECENT \subseteq S$ denote the set of the $p-1$ most recent elements in $S$.
We maintain the 
multi-set $\Delta(RECENT,a) = [\delta(u,a)\ :\ u\in RECENT]$ containing the images through $\delta(\cdot,a)$ of the states in $RECENT$, the Karp-Rabin fingerprint $\kappa_1$ of $RECENT$, and the Karp-Rabin fingerprint $\kappa_2$ of the set $\delta(RECENT,a)$. 
To enumerate the fingerprint pairs in Operation (4), 
corresponding to subsets $S'\subseteq S$ of cardinality $p-1$, 
we first output $\kappa_1$ --- the fingerprint of $S'=RECENT$ --- and then incrementally change $S'$, $\kappa_1$, and $\kappa_2$ until all such subsets of $S$ have been generated. Each generated subset $S'$ differs from the previously-generated one by at most two states, so this enumeration can be performed in constant time per element. Operation (5) is supported in constant time because, while enumerating subsets $S'\subseteq S$, we also incrementally maintain their images (multiset) $\Delta(RECENT,a)$. This allows us computing $|\delta(S'\cup\{u\},a)|=p$ in the claimed running time.

\paragraph{Data Structure}
Let us make more formal the above intuition. 
The data structure representing set $S$ in Algorithm \ref{algo: optimized algo} is formed by the following components:

\begin{enumerate}
    \item A doubly-linked list $L$ storing the elements of $S$ in reversed order of insertion, from the most recent to the oldest. We require for $L$ to be \emph{stable}, i.e.\ once a cell containing element $x\in S$ is created, its address in memory does not change until the cell is deleted.
    Notation $L[i]$ indicates the element of $S$ contained in the $i$-th cell of $L$ (i.e.\ $L[1]$ is the element of $S$ that was inserted most recently). Together with $L$, we always keep the address of its head (storing $L[1]$) and tail (storing $L[|L|]$).
    \item A dictionary denoted as ``$L^{-1}$'' associating to every $x\in S$, the address $L^{-1}[x]$ of the cell  of $L$ containing $x$.
    \item A pointer $PTR_R$ ($R$ stands for \emph{recent}) to the cell of $L$ containing the $(p-1)$-th most recent element of $S$, i.e. $L[p-1]$. If $L$ has less than $p-1$ elements, $PTR_R$ points to the last element, $L[|L|]$.
    \item A dictionary $RECENT$ storing the elements $L[1,\min\{p-1,|L|\}]$. In the description below, we will use the same name $RECENT$ both for this dictionary and for the set it represents. 
    \item A dictionary $\Delta(RECENT,a)$ storing the \emph{multiset} $[\delta(u,a)\ :\ u \in RECENT]$. Again, we will use the notation $\Delta(RECENT,a)$ both for the dictionary and the multiset it represents. 
    \item The Karp-Rabin fingerprint $\bar\kappa_1 = \kappa(RECENT)$.
    \item The Karp-Rabin fingerprint $\bar\kappa_2 = \kappa(\delta(RECENT,a))$. 
\end{enumerate}

We use the dictionary of \cite{Dietzfelbinger90} for $L^{-1}$, $RECENT$, and $\Delta(RECENT,a)$. This data structure supports updates, lookups, and cardinality\footnote{i.e.\ $|D|$ returns the number of \emph{distinct} elements stored in dictionary $D$.} operations in constant-time w.h.p. and uses space linear in the number of elements it contains. 
Importantly, note that dictionary $\Delta(RECENT,a)$ stores a \emph{multiset}; in particular, inserting/deleting an element from the dictionary increases/decreases its multiplicity by 1.

\paragraph{Supporting Operations on $S$}

\paragraph{(1) Initialization}

Operation (1), $S \gets \emptyset$, is implemented by:

\begin{itemize}
	\item Initializing empty $L$, $L^{-1}$, $RECENT$, $\Delta(RECENT,a)$.
	\item Initializing $PTR_R$ as a null pointer,
	\item Initializing $\bar \kappa_1 \gets \bar \kappa_2 \gets 0$. 
\end{itemize}

All these operations take constant time. Note that $\bar \kappa_1 = \bar \kappa_2 = 0$ are indeed the Karp-Rabin fingerprints of $RECENT = \delta(RECENT,a) = \emptyset$.

\paragraph{(2) Insertions} Operation $S.$\textsc{insert}$(u)$ is supported as follows. First of all note that, before calling $S.$\textsc{insert}$(u)$, state $u$ does not belong to $S$ (because the sorted list $L_a$ contains only one triple $(\ell_u,1,u)$, triggering a call to this insertion). 

We prepend $u$ at the beginning (head) of $L$ and insert the address of the new head of $L$ in $L^{-1}[u]$. We distinguish two cases:

\begin{itemize}
	\item[(A)] If $|RECENT|<p-1$: we insert $u$ in dictionary $RECENT$, and $\delta(u,a)$ in dictionary $\Delta(RECENT,a)$.  We update $\kappa_1 \gets \kappa_1 + x^u \mod q$. If the multiplicity of $\delta(u,a)$ in $\Delta(RECENT,a)$ is equal to 1, then we update $\kappa_2 \gets \kappa_2 + x^{\delta(u,a)} \mod q$.
    If $|RECENT|=1$, $PTR_R$ is initialized at the head of $L$. 
	\item[(B)] If $|RECENT|=p-1$. In this case, $RECENT$ is full: 
    after the insertion of $u$, the element pointed by $PTR_R$ will no longer be among the $p-1$ most recent elements in $S$. Let $v$ be the state contained in the cell of $L$ pointed by $PTR_R$. Then, we remove $v$ from $RECENT$, remove (one occurrence of) $\delta(v,a)$ from $\Delta(RECENT,a)$, and update $\kappa_1 \gets \kappa_1 - x^v \mod q$. If the multiplicity of $\delta(v,a)$ in $\delta(RECENT,a)$ is equal to 0, then we update $\kappa_2 \gets \kappa_2 - x^{\delta(v,a)} \mod q$.
    At this point, we run the operations described above in step (A) (that is, we insert $u$ and $\delta(u,a)$ in $RECENT$ and $\Delta(RECENT,a)$, respectively, updating $\kappa_1$ and $\kappa_2$ accordingly).
\end{itemize} 

\paragraph{(3) Deletions} Operation $S.$\textsc{delete}$(u)$ is supported as follows. First of all note that, before calling $S.$\textsc{delete}$(u)$, state $u$ does belong to $S$: because of the way Algorithm \ref{algo: optimized algo} operates,  and the fact that we excluded states $u\in Q$ with $\ell_u=r_u$, the triple $(r_u,0,u)$ --- triggering a call to $S.$\textsc{delete}$(u)$ --- always comes after triple $(\ell_u,1,u)$ --- triggering a call to $S.$\textsc{insert}$(u)$ --- in $L_a$. Moreover, these are the only two triples associated with state $u$.

If $u \notin RECENT$, then deleting $u$ is easy: we simply delete the cell of $L$ with address $L^{-1}[u]$, as well as the corresponding record in $L^{-1}$. On the other hand, if $u \in RECENT$, then $u$ is one of the $p-1$ most recent elements in $S$ so we also need to update $PTR_R$, $RECENT$, $\Delta(RECENT,a)$, $\kappa_1$, and $\kappa_2$.
First of all, we delete $u$ from $RECENT$ and decrease the multiplicity of $\delta(u,a)$ by one unit in dictionary $\Delta(RECENT,a)$.
We apply the corresponding modifications to $\kappa_1$ and $\kappa_2$: the former fingerprint is updated as $\kappa_1 \gets \kappa_1 - x^u \mod q$. The latter fingerprint is modified only if the multiplicity of $\delta(u,a)$ in $\Delta(RECENT,a)$ is equal to zero: in that case, we update $\kappa_2 \gets \kappa_2 - x^{\delta(u,a)} \mod q$.
At this point, 
we delete the cell of $L$ with address $L^{-1}[u]$, as well as the corresponding record in $L^{-1}$, and restore $PTR_R$ appropriately so that it points to $L[\min\{p-1,|L|\}]$. Let $v$ be the node pointed now by $PTR_R$. If $v \notin RECENT$, then we insert $v$ in $RECENT$ and insert $\delta(v,a)$ in $\Delta(RECENT,a)$. We finally update correspondingly the fingerprints $\kappa_1$ and $\kappa_2$: the former fingerprint is updated as $\kappa_1 \gets \kappa_1 + x^v \mod q$. The latter fingerprint is modified only if the multiplicity of $\delta(v,a)$ in $\Delta(RECENT,a)$ is equal to one: in that case, we update $\kappa_2 \gets \kappa_2 + x^{\delta(v,a)} \mod q$.

\paragraph{(4-5) Pair Enumeration and checking $|\delta(S'\cup\{u\},a)|=p$}

Observe that operations (4-5) need to be supported only if $|P|\geq p-1$ (otherwise, there are no pairs of fingerprints to enumerate). We can therefore assume that $PTR_R$ points to the cell of $L$ containing $L[p-1]$.

We describe a recursive procedure enumerating all subsets $S'\subseteq S$, with $|S'|=p-1$, starting from $S'=RECENT$ (of which we already know $\kappa(S') = \kappa_1$ and $\kappa(\delta(S',a)) = \kappa_2$).
Our enumeration procedure will have the following property.
Let $S'_1, S'_2, S'_3\dots,$ be the sequence of enumerated sets.
This sequence of sets is obtained by modifying incrementally $RECENT$, $\Delta(RECENT,a)$, and the corresponding fingerprints $\kappa_1$ and $\kappa_2$. 
In this sense, in the description below $S'_i$ is represented by the current state of those structures at step $i=1, 2, \dots$.
For each $i>1$, set $S'_i$ will be created from $S'_{i-1}$ by removing an element $u$ from $S'_{i-1}$, inserting a new element $v\neq u$ ($v\notin S'_{i-1}$) into it, and finally renaming $S'_{i-1}$ into $S'_i$.
After each such pair of one insertion and one deletion, 
fingerprints $\kappa_1$ and $\kappa_2$, as well as dictionaries $RECENT$ and $\Delta(RECENT,a)$, can be updated in constant time w.h.p. as described in the paragraphs above (\emph{(2) Insertions} and \emph{(3) Deletions}), so that it always holds $\kappa_1 = \kappa(S'_i)$ and $\kappa_2 = \kappa(\delta(S'_i,a))$:

\begin{itemize}
    \item Insertion of a state $u$: we insert $u$ in dictionary $RECENT$, and $\delta(u,a)$ in dictionary $\Delta(RECENT,a)$.  We update $\kappa_1 \gets \kappa_1 + x^u \mod q$. If the multiplicity of $\delta(u,a)$ in $\Delta(RECENT,a)$ is equal to 1, then we update $\kappa_2 \gets \kappa_2 + x^{\delta(u,a)} \mod q$.
    \item Deletion of a state $u$: we delete $u$ from $RECENT$ and decrease the multiplicity of $\delta(u,a)$ by one unit in dictionary $\Delta(RECENT,a)$. We update $\kappa_1 \gets \kappa_1 - x^u \mod q$ and, if the multiplicity of $\delta(u,a)$ in $\Delta(RECENT,a)$ is equal to zero, we update $\kappa_2 \gets \kappa_2 - x^{\delta(u,a)} \mod q$.
\end{itemize}

Once all these structures have been updated for set $S'_i$, notice that it is very easy also to check in $O(1)$ time w.h.p. whether $|\delta(S'_i\cup\{u\},a)|=p$ (with $u\notin S'_i$), i.e. to solve operation (5): this condition is true if and only if (1) $|\Delta(RECENT,a)|=p-1$, and $\delta(u,a)\notin \Delta(RECENT,a)$. 

All we are left to do is to show how to enumerate all subsets $S'_1(=S'), S'_2, S'_3\dots,$ of $S$ (i.e. of list $L$) of cardinality $p-1$ in constant amortized time (w.h.p.) per subset. Intuitively, starting from $S'_1 = L[1,p-1] = RECENT$, we iteratively remove element $L[j]$ from the current set $S'_i$ and insert into it element $L[j+1]$, for $j = p-1, p, \dots, |L|-1$. After this insertion/removal pair, we recursively enumerate all subsets of cardinality $p-2$ of $L[1,j]$ (the base case of the recursion being the enumeration of empty subsets from some range of $L$). In Algorithm \ref{algo: enumerate} we give the pseudocode of this recursive procedure and later argue that it runs in $O(1)$ amortized delay (w.h.p.) per enumerated subset. 
In order to solve Operations (4-5), this procedure must be called as \textsc{EnumerateSubsets}$(L, p-1, PTR_R, RECENT)$.

While in Algorithm \ref{algo: enumerate} we use an array-like interface for $L$ for readability reasons, the reader should keep in mind that $L$ is a doubly-linked list. 
In particular, the \texttt{for} loop in Line \ref{line:for loop enumerate}, accessing $L[j]$ for $j=j,\dots, |L|-1$, is simulated on the linked list by starting from the cell of $L$ pointed by $PTR$, i.e. $L[k]$, and moving iteratively to the next cells $PTR.next$ until we reach the end of the list. 

In Line \ref{line:adjust enumerate} we restore $S'$ to its initial state: before this line, it holds $S' \cap L = L[|L|-k+1,|L|]$ ($S'$ contains the last $k$ elements of $L$). After the operation in Line \ref{line:adjust enumerate}, $S'$ satisfies again $S' \cap L = L[1,k]$ as at the beginning of the procedure (in the main Algorithm \ref{algo: optimized algo}, we also perform the corresponding updates on $\Delta(S',a)$, $\kappa_1$, $\kappa_2$, also in $O(|L|-k)$ time w.h.p.). 
This operation can be easily implemented on the linked list $L$ in $O(|L|-k)$ time, as follows. If $k\le|L|/2$, then $k\le|L|-k$ so we can afford spending $O(k) \subseteq O(|L|-k)$ time deleting the last $k$ elements of $L$ from $S'$, and adding the first $k$ elements of $L$ to $S'$ (recall that we always keep a pointer to the first and last cells of $L$). If, on the other hand, $k>|L|/2$, then the elements $L[|L|-k+1,k]$ belong to $L[|L|-k+1,|L|] \cap L[1,k]$ so they do not need to be touched (i.e. removed and then re-inserted from/into $S'$). As a result, we only need to remove $L[k+1,|L|]$ from $S'$ and insert $L[1,|L|-k]$ into $S$, in total $O(|L|-k)$ time. 
To conclude, we observe that the running time $O(|L|-k)$ of Line \ref{line:adjust enumerate} can be charged to the $|L|-k$ iterations of the \texttt{for} loop. 
We conclude that Algorithm \ref{algo: enumerate} spends $O(1)$ amortized time per enumerated subset of cardinality $k$ of $L$.

\begin{algorithm2e}[ht!]
    \caption{\textsc{EnumerateSubsets}($L, k, PTR, S'$)}
	\label{algo: enumerate}
	\Input{Doubly-linked list $L$ whose elements are all distinct, integer $0<k\leq |L|$, pointer $PTR$ to the cell 
    of $L$ containing $L[k]$, and set $S'$ such that $S'\cap L = L[1,k]$. 
    When the procedure is called for the first time, $S'$ must be equal to $L[1,k]$.
    All arguments are passed by reference.}
	\Output{Enumerate (by reference) all subsets of $L$ of cardinality $k$ in constant amortized delay per subset.}
    \medskip

    \textbf{output} $S'$\;

    \If{$k = 0$ \textbf{or} $k = |L|$}{   
      \Return{}
    }
    
    \For{$j=k,\dots,|L|-1$\label{line:for loop enumerate}}{

        $S' \gets (S' \setminus \{L[j]\})\cup \{L[j+1]\}$\;
        \textsc{EnumerateSubsets}$(L[1,j],  k-1, PTR.prev, S')$\tcp*{$PTR.prev$ is the previous cell of $PTR$.}\label{line:recursive call enumerate}
        
    }

    $S' \gets (S' \setminus L[|L|-k+1,|L|] ) \cup L[1,k]$\;\label{line:adjust enumerate}
    
\end{algorithm2e}

\paragraph{Putting Everything Together}

We can conclude the analysis of Algorithm \ref{algo: optimized algo}. 
As noted previously, if $p' = \width(\mathcal A_{\min}) < p$ then $\mathcal B$ is empty and the running time of Algorithm \ref{algo: optimized algo} is dominated by Lines \ref{line: hopcroft} (minimization) and  \ref{line: intervals} (intervals computation). 
We finally obtain:

\begin{theorem}\label{thm: parameterized running time}
    Let $\mathcal A$ be an input DFA with $m$ transitions and $n$ states, and let 
    $p'=\width(\mathcal{A}_{\min})$ be the width of the minimum automaton equivalent to $\mathcal A$.
    In the word-RAM model with word size $w$, we can solve \textsc{DfaDetWidth} with parameter $p\geq 2$ on $\mathcal A$ 
    in $O\left(m\binom{p'-1}{p-1} + m\log n\right)$ time w.h.p.  and $O\left(m\binom{p'-1}{p-1}\right)$ words of working space, provided that $p\le p'$. If, on the other hand,  $p> p'$, then the running time and working space become $O(m\log n)$ and $O(m)$ words, respectively, both in the worst case. This result holds provided that the working space does not exceed the space budget $2^w$ of the model.
\end{theorem}

\begin{remark}
    If $p' \in O(\log\log n)$ (which includes the case where $\mathcal A_{\min}$ is Wheeler, i.e.\ $\width(\mathcal A_{\min})=1$), then the running time of Theorem \ref{thm: parameterized running time} reaches its minimum of $O(m \log n)$. In any case, observe that $m\binom{p'-1}{p-1} \in O(m(p'-1)^{p-1}) \subseteq O(m^p)$ (since $p'\leq n \leq m$), so this running time is not worse than that of Theorem \ref{thm:simple algorithm}.
\end{remark}

\begin{remark}\label{remark: exp-space in word RAM}
  One reasonable way to remove the limitation that the working space must not exceed $2^w$, is to use a word-RAM model with infinite memory and multiple addressing (Hagerup \cite{HagerupRAM}, for example, assumes infinite memory and double addressing). We stress out that, however, to the best of our knowledge the behavior of exponential-space algorithms in the word RAM model is not properly discussed in the literature (see also \cite{bille2015regular}: a common misconception is that de-referencing a pointer always takes constant time, regardless of the fact that, if the working space is exponential, a pointer takes linear space in the input's size). Since the working space of our algorithm does not exceed $O(m^p)$ and $w \geq \log m$, this model would need $p$ words to access any cell in the working space. Ultimately, this means that the space gets multiplied by a factor of $p$ (since $\mathcal B$ is stored using a pointer-based data structure), and that the running time gets multiplied by no more than $\mathtt{poly}(p)$ (e.g., multiplication between $p$-words integers --- required in the computation of $O(\log(m^p))$-bits Karp-Rabin fingerprints --- now takes no more than $O(p^2)$ time). 
  Even when multiplying the working space and running time of Theorem \ref{thm: parameterized running time} by a $\mathtt{poly}(p)$ factor, we obtain an algorithm running in $O(m^p)$ time and $O(m^p)$ words of working space. This running time still matches the conditional lower bound provided in the next section. 
\end{remark}


\section{Hardness Results}\label{sec: hardness}
In this section, we give hardness results for \textsc{DfaDetWidth} based on different assumptions, showing that the algorithm of Section \ref{sec:algorithm} is conditionally optimal. 
We start by recapping some basics of fine-grained complexity.

\paragraph{ETH and SETH}
We start by recalling the exponential time hypothesis (ETH). In the SAT problem, we are given a Boolean formula $\Phi$ in conjunctive normal form over Boolean variables, i.e., $\Phi$ is a conjunction of clauses, where a clause is a disjunction of literals that are either equal to a Boolean variable $x_i$ or its negation $\neg x_i$. In $k$-SAT each clause is restricted to contain at most $k$ literals. We call $N$ the number of Boolean variables and $M$ the number of clauses.

\begin{definition}[ETH]
    The \emph{exponential time hypothesis (ETH)} states that there exists $\delta>0$ such that 3-SAT on $N$ variables cannot be solved in time $O(2^{\delta N})$.
\end{definition}

The sparsification lemma by Impagliazzo et al.~\cite{ImpagliazzoPZ01} states that we can assume that $M=O(N)$ when building reductions from $k$-SAT. More precisely, this result by Impagliazzo et al.\ entails an algorithm that takes an arbitrary $k$-SAT formula $\Phi$ and outputs formulae $\Phi_1,\ldots, \Phi_t$ such that $\Phi$ is satisfiable if and only if one of the formulae $\Phi_1,\ldots, \Phi_t$ is satisfiable. Moreover, $t\le 2^{\eps N}$ for some $\eps>0$ and those formulae have the property that they are sparse, i.e., have at most $O(N)$ clauses (the constant hidden in the $O$ may depend on $\eps$). The crux of the construction is that an algorithm for $k$-SAT that runs in $O(2^{\delta N})$ time for some $\delta>0$ can be employed on all $t$ formulae $\Phi_i$, resulting in an overall running time of $O(2^{(\delta + \eps) N})$ for checking whether $\Phi$ is satisfiable. Hence an algorithm for sparse $k$-SAT contradicting ETH yields an algorithm for arbitrary $k$-SAT contradicting ETH. We proceed with the strong exponential time hypothesis.

\begin{definition}[SETH]
    The Strong Exponential Time Hypothesis (SETH) states that, for all $\eps>0$, there exists $k\ge 3$ such that $k$-SAT cannot be solved in time $O(2^{(1-\eps) N})$.
\end{definition}

\paragraph{Hardness Result}
Our goal in this section is to prove the following theorem.
\begin{theorem}\label{thm: hardness}
    \textsc{DfaDetWidth} 
    cannot be solved in:
    \begin{enumerate}
        \item\label{enum 1} $\poly(n,m)$ time, unless $P=NP$,
        \item\label{enum 2} $2^{o(\sqrt{m})}$ time, unless ETH fails
        \item\label{enum 3} $m^{o(p)}$ time for any $p\le m^c$, where $c<1/2$ is any constant, unless ETH fails,
        \item\label{enum 4} $O(m^{p -\epsilon})$ time for any constant $p$, where $\epsilon > 0$ is any constant, unless SETH fails.
    \end{enumerate}
    All claims even hold if the alphabet is binary.
\end{theorem}

\paragraph{Overview}
Given a $k$-SAT formula $\Phi$, we build the following input DFA $\mathcal A=(Q, \Sigma, \delta, s, F)$ for \textsc{DfaDetWidth} with parameter $p>2$. We start by describing a reduction with $\Sigma=\{0, \ldots, k, \#\}$ and later on argue how the construction can be transformed into one with a binary alphabet. We start from a $k$-SAT formula $\Phi = C_1 \wedge \ldots \wedge C_M$ with $M$ clauses $C_j = l_{j, 1} \vee \ldots \vee l_{j, k}$ for $j\in [M]$ over $N$ variables $x_1,\ldots, x_N$. For a literal $l$, we define $x(l)$ as the unique variable occurring in the literal.
We split the $N$ variables into $p$ blocks of $N/p$ variables each (we assume for simplicity that $N$ and $p$ are both powers of 2 and thus also that $N$ is divisible by $p$), more precisely, for $r\in [p]$, the $r$'th block of variables consists of the variables $x_i$ with $i\in B_r := [(r - 1)\cdot N/p + 1, r\cdot N/p]$. We sometimes abuse notation and say that the variable $x_i\in B_r$ if $i\in B_r$. We use $A:=\{a_1, \ldots, a_T\}$ to denote the set of all possible, i.e., $T=2^{N/p}$ many, assignments to the $N/p$ variables of any of the $p$ blocks. 

Now, the set of states $Q$ of $\mathcal A$ consists of disjoint sets $V^{out}, V^{in}, I, C$. The core of our reduction is the set $C$ that consists of $p \cdot T$ cycle-like subgraphs $S_{r, a_\ell}$, one for each $r\in [p]$ and for each assignment $a_\ell$ with $\ell\in [T]$. The graph $S_{r, a_\ell}$ contains $M+N+1$ nodes $v^{r, a_\ell}_0, \ldots, v^{r, a_\ell}_{M + N}$. The two parts $v^{r, a_\ell}_0, \ldots, v^{r, a_\ell}_{M}$ and $v^{r, a_\ell}_{M + 1}, \ldots, v^{r, a_\ell}_{M + N}$ are constructed in two different ways:
\begin{enumerate}
    \item For $j\in [M]$, $r\in [p]$ and an assignment $a_\ell$, there are at most $k$ transitions between the nodes $v^{r, a_\ell}_{j - 1}$ and $v^{r, a_\ell}_j$ and their possible labels are $[k]$. Let $h\in [k]$, then there are two cases: (1) Assume that $x(l_{j, h})\notin B_r$, then $\delta(v^{r, a_\ell}_{j - 1}, h) = v^{r, a_\ell}_j$. (2) If $x(l_{j, h})\in B_r$ instead, we have $\delta(v^{r, a_\ell}_{j - 1}, h) = v^{r, a_\ell}_j$ if and only if assignment $a_\ell$ makes literal $l_{j, h}$ true. In other words, the $j$'th pair of the sequence $v^{r, a_\ell}_0, \ldots, v^{r, a_\ell}_{M}$ may be only connected by transitions with labels $[k]$. They are always connected with a transition labeled $h$ if literal $l_{j, h}$ does not contain a variable from the block $B_r$, otherwise they are connected with a transition labeled $h$ if and only if the assignment makes this literal and thus the clause true.
    \item For the sequence $v^{r, a_\ell}_{M + 1}, \ldots, v^{r, a_\ell}_{M + N}$ we have the following transitions between subsequent nodes. For $i$ such that $M + i\in [M + 1, M + N]$, we have $\delta(v^{r, a_\ell}_{M + i - 1}, b) = v^{r, a_\ell}_{M + i}$ for $b\in \{0,1\}$ if $i \notin B_r$. For $i\in B_r$ instead we have $\delta(v^{r, a_\ell}_{M + i - 1}, b) = v^{r, a_\ell}_{M + i}$ if and only if $a_\ell$ assigns $b$ to variable $x_i$. As a result the $i$'th portion of the second part of the sequence spells exactly the assignment $a_\ell$.
\end{enumerate}
In addition we have $\delta(v^{i, a_\ell}_{M + N}, \#) = v^{i, a_\ell}_0$. There are no other transitions within $S_{r, a_\ell}$. 
We refer the reader to Figure~\ref{fig: one reduction} for an illustration of this construction for an example of a $k$-SAT formula. 

\begin{figure}[ht!]
  \centering{
    \resizebox{0.9\columnwidth}{!}{
      \input{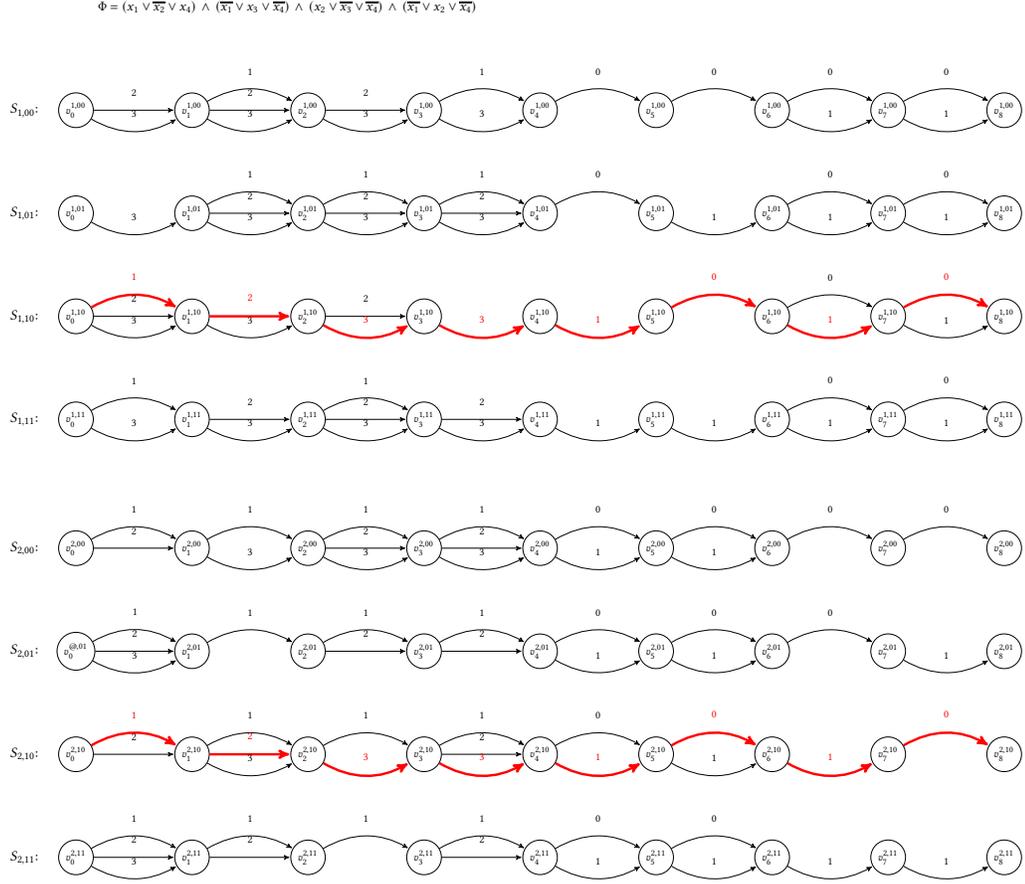}
    }
  }
\caption{The construction of the subgraphs in the set $C$ for $p=2$ for a 3-SAT formula $\Phi = (x_1 \vee \overline{x_2} \vee  x_4) \,\wedge\,(\overline{x_1} \vee x_3 \vee \overline{x_4}) \,\wedge\, (x_2 \vee \overline{x_3} \vee \overline{x_4})\,\wedge\, (\overline{x_1} \vee x_2 \vee \overline{x_4})$ with $N=4$ variables and $M=4$ clauses.
There are $p=2$ blocks, the first block consists of variables $x_1, x_2$, while the second block consists of variables $x_3, x_4$. Each block contains one subgraph for each of the $T = 2^{N/p} = 2^{4/2} = 4$ assignments. For better readability, we omit the transition labeled $\#$ from $v_{8}^{r, u}$ to $v_0^{r, u}$ for each block $r$ and assignment $u$. We observe that the assignment $x_1=1, x_2=0, x_3 = 1, x_4=0$ satisfies $\Phi$. Consequently, we find $p=2$ equally labeled cycles in the subgraphs $S_{1, 10}$ and $S_{2, 10}$. The string that is spelled by these cycles is equal to $\alpha = 1233\,1010\#$. Character $h$ at position $j$ (for $j\in [M]=[4]$) in this string indicates that clause $C_j$ is made valid by an assignment $u$ to the $h$'th literal of this clause. The assignment $u$ is that one for which we find the cycle in its corresponding subgraph. For example, the first $1$ in the string indicates that $C_1$ is made true by an assignment to its first literal; as we found the cycle labeled $\alpha$ in $S_{1, 10}$, we can conclude that the assignment to the variables of the first block that makes $C_1$ true is $10$. Indeed, the literal that makes $C_1$ true is $x_1$. For $\ell = M + i\in [M + 1, M + N]$ instead, the string $\alpha$ spells the assignment to variable $x_i$.}
\label{fig: one reduction}
\end{figure}

The rationale behind the construction is to make the following lemma hold.
\begin{lemma}\label{lemma: cycle iff kSAT}
    The set $\{S_{r, a_\ell}\}_{r\in [p], \ell\in [T]}$ contains $p$ disjoint cycles, one from each set $\{S_{1, a_\ell}\}_{\ell\in [T]}, \ldots, \{S_{p, a_\ell}\}_{\ell\in [T]}$, spelling the same string if and only if the $k$-SAT formula $\Phi$ is satisfiable. 
\end{lemma}
\begin{proof}
    Let $\Phi = C_1 \wedge \ldots \wedge C_M $ be the $k$-SAT formula. Now assume that $\Phi$ has a satisfying assignment. Let $\{u_r\}_{r\in [p]}$ be the corresponding assignments to the $p$ blocks of variables. Clearly,  
    for every clause $C_j$, $j\in [M]$, there exists a literal, say $l_{j, h_j}$, such that $C_j$ is made true by the assignment to the block that contains variable $x(l_{j, h_j})$. Let us call this block $B_{r_j}$.
    It now follows that each of the $p$ subgraphs $S_{r, u_{r}}$ for $r\in [p]$ contains a cycle labeled $h_1 \ldots h_M \, u_1 \ldots u_p\#$. In order to see this, we start with the first part, i.e., the string $h_1 \ldots h_M$. Observe that, for every $j\in [M]$, each $S_{r, u_{r}}$ with $r\neq r_j$ contains an edge labeled $h_j$ by construction. Now, let $r=r_j$, i.e., $x_{h_j} \in B_{r}$, then $S_{r, u_{r}}$ is guaranteed to contain the edge labeled $h_j$ because $C_j$ is made true by the assignment $u_{r}$ to the variables of block $B_r$. Furthermore, it is clear that by construction all subgraphs $S_{r, u_{r}}$ contain a path labeled $u_1 \ldots u_p$ in order to spell the second part of the string.

    Conversely, assume that there exist $p$ disjoint cycles spelling the same string, say $\alpha$. For a given block $r\in [p]$, observe that the construction guarantees that no two subgraphs $S_{r, u}$ and $S_{r, u'}$ spell the same string, as they contain different transitions between the nodes $v^{r, u}_{M + i - 1}$ and $v^{i, u}_{M + i}$ and between $v^{r, u'}_{M + i - 1}$ and $v^{r, u'}_{M + i}$ respectively, for $i\in B_r$, i.e., in the portion of the second part of the subgraph that spells the assignments $u$ and $u'$. The set of $p$ cycles thus contains exactly one out of the $N/p$ subgraphs of each block, i.e., one subgraph $S_{r, u_r}$ for each variable $r\in [p]$. The set of cycles thus corresponds to an assignment, namely the assignment that assigns $u_r$ to the variables of the $r$'th block.
    Let now $C_j$ be any clause in $\Phi$ and let $h\in [k]$ be the $j$'th character in $\alpha$. This implies that each out of the $p$ subgraphs $\{S_{r, u_r}\}_{r\in [p]}$ contains a transition labeled $h$ at the $j$'th position (starting to count positions from nodes with index $0$). In particular, it follows that the subgraph $S_{r, u_r}$ for which $x(l_{j, h})\in B_r$ contains a transition labeled $h$ at position $j$ and thus it follows by the construction that the assignment $u_r$ makes $C_j$ true. We conclude that all clauses are satisfied and this concludes the proof.
\end{proof}

In order to complete the reduction, it remains to describe how to embed these subgraphs $\{S_{r, a_\ell}\}_{r\in [p], \ell\in [T]}$ into a DFA $\mathcal A$ in such a way that the existence of $p$ cycles labeled by the same string become a witness of $\mathcal L(\mathcal A)$ being of width at least $p$ following Theorem~\ref{thm: main: width - cycle}.
Our construction thus has to ensure that 
\begin{enumerate} [(1)]
    \item the graph is a connected DFA,
    \item corresponding (i.e., same distance from $\#$ in the subgraphs) nodes $v,w$ in any pair of subgraphs $S_{r, a_\ell}$ and $S_{r', a_{\ell'}}$ have a non-empty co-lexicographic intersection $\mathcal I(v) \cap \mathcal I(w)$,
    \item the DFA is indeed minimum for its recognized language, and 
    \item the alphabet can be reduced to $\{0,1\}$.
\end{enumerate}

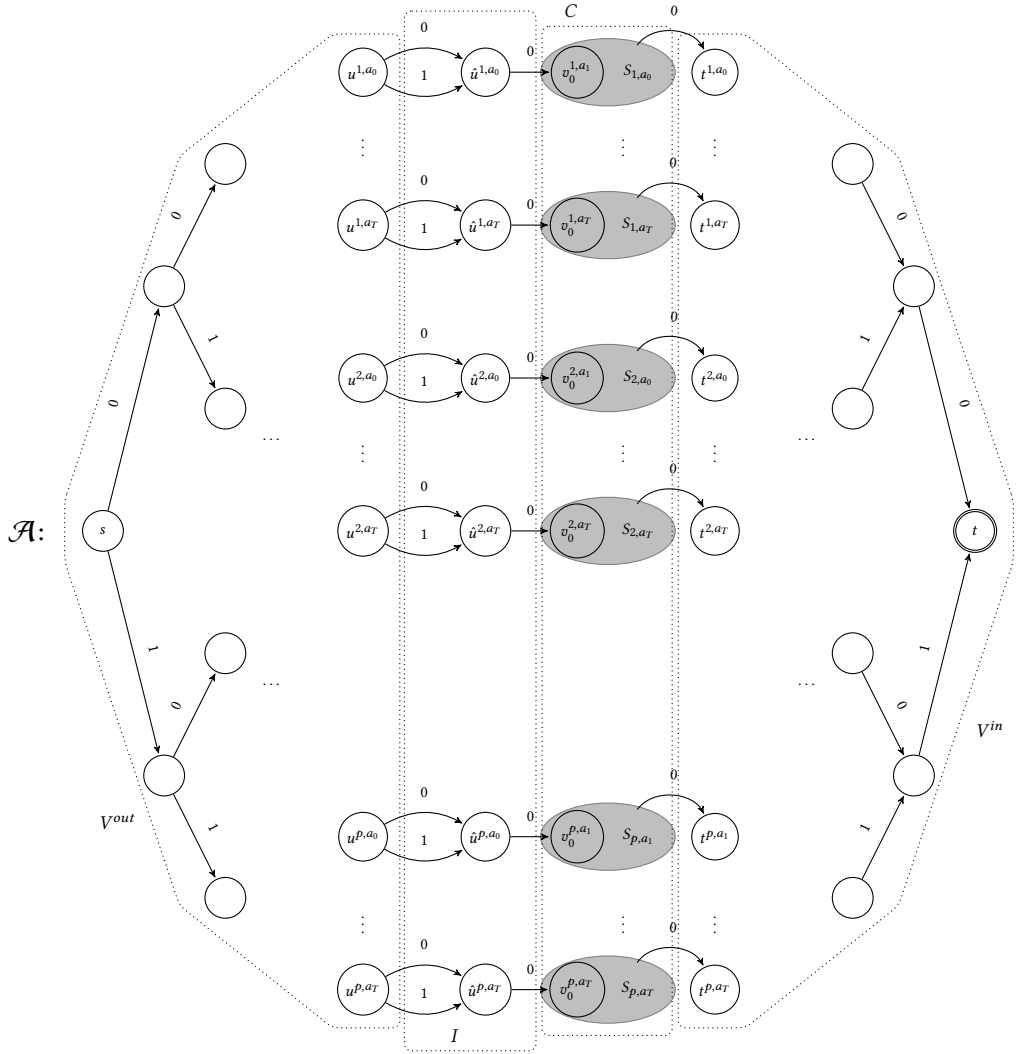
\begin{figure}[ht!]
  \centering{
    \resizebox{0.9\columnwidth}{!}{
      \begin{tikzpicture}[
          scale=.6,
          ->,
          >=stealth',
          shorten >=1pt,
          auto,
          semithick,
          every node/.style={minimum size=8mm}
    ]
    \node at (-2.0, 0) {\Huge $\mathcal A$:};

    \node[circle, draw] (s)   at (0.5,0)  {$s$};
    \node[circle, draw] (b)   at (2.5,-8)  {};
    \node[circle, draw] (a)   at (2.5,8)  {};
    \path (s) edge [sloped] node {$1$} (b);
    \path (s) edge [sloped] node {$0$} (a);
    \node[circle, draw] (a1)   at (4.5,12)  {};
    \node[circle, draw] (a2)   at (4.5,4)  {};
    \path (a) edge [sloped] node {$0$} (a1);
    \path (a) edge [sloped] node {$1$} (a2);
    \node[circle, draw] (b1)   at (4.5,-4)  {};
    \node[circle, draw] (b2)   at (4.5,-12)  {};
    \path (b) edge [sloped] node {$0$} (b1);
    \path (b) edge [sloped] node {$1$} (b2);

    \node (d) at (6, 3) {$\ldots$};
    \node (d) at (6, -5) {$\ldots$};

    \path [draw = black, rounded corners, inner sep=100pt, dotted]
        (-0.75, 1) -- (3.0, 12.25) -- (8.0, 16.25) -- (10.2, 16.25) -- (10.2, -16.25) -- (8.0, -16.25) -- (3.0, -12.25) -- (-0.75, -1) -- cycle;
    \node  (empty)    at (1.0, -9.5)  {\Large $V^{out}$};

    \node[circle, draw] (x1)   at (9, 15)  {$u^{1, a_0}$};
    \node at (9, 12.7) {$\vdots$};
    \node[circle, draw] (xN)   at (9, 10)  {$u^{1, a_{{T}}}$};
    
    \node[circle, draw] (x12)   at (9, 5)  {$u^{2, a_0}$};
    \node at (9, 2.7) {$\vdots$};
    \node[circle, draw] (xN2)   at (9, 0)  {$u^{2, a_{{T}}}$};
    
    \node[circle, draw] (y1)   at (9, -10)  {$u^{p, a_0}$};
    \node at (9, -12.7) {$\vdots$};
    \node[circle, draw] (yN)   at (9, -15)  {$u^{p, a_T}$};

    \begin{scope}[xshift=-2cm]
        \node[circle, draw] (a1p)   at (15, 15)  {$\hat u^{1, a_0}$};
        \node[circle, draw] (aNp)   at (15, 10)  {$\hat u^{1, a_T}$};

        \node[circle, draw] (c1p)   at (15, 5)  {$\hat u^{2, a_0}$};
        \node[circle, draw] (cNp)   at (15, 0)  {$\hat u^{2, a_T}$};
        
        \node[circle, draw] (b1p)   at (15, -10)  {$\hat u^{p, a_0}$};
        \node[circle, draw] (bNp)   at (15, -15)  {$\hat u^{p, a_T}$};        
    
        \begin{scope}[xshift=3cm]
            \draw[fill=lightgray, draw=gray] (16, 15) ellipse (2.2cm and 1.1cm); 
            \node () at (17, 15) {$S_{1, a_0}$};
            \node at (16.5, 12.7) {$\vdots$};
            \draw[fill=lightgray, draw=gray] (16, 10) ellipse (2.2cm and 1.1cm); 
            \node () at (17, 10) {$S_{1, a_T}$};

            \draw[fill=lightgray, draw=gray] (16, 5) ellipse (2.2cm and 1.1cm); 
            \node () at (17, 5) {$S_{2, a_0}$};
            \node at (16.5, 2.7) {$\vdots$};
            \draw[fill=lightgray, draw=gray] (16, 0) ellipse (2.2cm and 1.1cm); 
            \node () at (17, 0) {$S_{2, a_T}$};
            
            \draw[fill=lightgray, draw=gray] (16, -10) ellipse (2.2cm and 1.1cm); 
            \node () at (17, -10) {$S_{p, a_1}$};
            \node at (16.5, -12.7) {$\vdots$};
            \draw[fill=lightgray, draw=gray] (16, -15) ellipse (2.2cm and 1.1cm); 
            \node () at (17, -15) {$S_{p, a_T}$};
        
            \node[circle, draw] (a1)   at (15, 15)  {$v_0^{1, a_1}$};
            \node[circle, draw] (aN)   at (15, 10)  {$v_0^{1, a_T}$};

            \node[circle, draw] (c1)   at (15, 5)  {$v_0^{2, a_1}$};
            \node[circle, draw] (cN)   at (15, 0)  {$v_0^{2, a_T}$};
            
            \node[circle, draw] (b1)   at (15, -10)  {$v_0^{p, a_1}$};
            \node[circle, draw] (bN)   at (15, -15)  {$v_0^{p, a_T}$};
        
            \path   (y1) edge [sloped, bend left] node {$0$} (b1p)
                    (y1) edge [sloped, bend right] node {$1$} (b1p);
            \path   (yN) edge [sloped, bend left] node {$0$} (bNp)
                    (yN) edge [sloped, bend right] node {$1$} (bNp);
            
            \path   (x12) edge [sloped, bend left] node {$0$} (c1p)
                    (x12) edge [sloped, bend right] node {$1$} (c1p);
            \path   (xN2) edge [sloped, bend left] node {$0$} (cNp)
                    (xN2) edge [sloped, bend right] node {$1$} (cNp);
                    
            \path   (x1) edge [sloped, bend left] node {$0$} (a1p)
                    (x1) edge [sloped, bend right] node {$1$} (a1p);
            \path   (xN) edge [sloped, bend left] node {$0$} (aNp)
                    (xN) edge [sloped, bend right] node {$1$} (aNp);
            
            \path   (a1p) edge [sloped] node {$0$} (a1)
                    (aNp) edge [sloped] node {$0$} (aN)
                    
                    (c1p) edge [sloped] node {$0$} (c1)
                    (cNp) edge [sloped] node {$0$} (cN)
                    
                    (b1p) edge [sloped] node {$0$} (b1)
                    (bNp) edge [sloped] node {$0$} (bN);
        
            \path [draw = black, rounded corners, inner sep=100pt, dotted]
            (9.35, 17.0) -- (13.65, 17.0) -- (13.65, -17.0) -- (9.35, -17.0) -- cycle ;
            \node  (empty)    at (11, -16.5)  {\Large $I$};
        
            \path [draw = black, rounded corners, inner sep=100pt, dotted]
            (13.85, 16.5) -- (18.1, 16.5) -- (18.1, -16.5) -- (13.85, -16.5) -- cycle ;
            \node  (empty)    at (14.8, 17)  {\Large $C$};
        
            \node[circle, draw] (t1)   at (19.5, 15)  {$t^{1, a_0}$};
            \node at (19.5, 12.7) {$\vdots$};
            \node[circle, draw] (tN)   at (19.5, 10)  {$t^{1, a_T}$};

            \node[circle, draw] (t12)   at (19.5, 5)  {$t^{2, a_0}$};
            \node at (19.5, 2.7) {$\vdots$};
            \node[circle, draw] (tN2)   at (19.5, 0)  {$t^{2, a_T}$};
            
            \node[circle, draw] (z1)   at (19.5, -10)  {$t^{p, a_1}$};
            \node at (19.5, -12.7) {$\vdots$};
            \node[circle, draw] (zN)   at (19.5, -15)  {$t^{p, a_T}$};
        
            \node[] 		(CA1)   at (16.5, 15.2) {};
            \path (CA1) edge [sloped, bend left=60] node {$0$} (t1);
            \node[] 		(CAN)   at (16.5, 10.2) {};
            \path (CAN) edge [sloped, bend left=60] node {$0$} (tN);

            \node[] 		(CA12)   at (16.5, 5.2) {};
            \path (CA12) edge [sloped, bend left=60] node {$0$} (t12);
            \node[] 		(CAN2)   at (16.5, 0.2) {};
            \path (CAN2) edge [sloped, bend left=60] node {$0$} (tN2);
            
            \node[] 		(CB1)   at (16.5, -9.8) {};
            \path (CB1) edge [sloped, bend left=60] node {$0$} (z1);
            \node[] 		(CBN)   at (16.5, -14.8) {};
            \path (CBN) edge [sloped, bend left=60] node {$0$} (zN);
        
            \node[circle, draw, accepting] (t)   at (28,0)  {$t$};
            \node[circle, draw] (tb)   at (26,-8)  {};
            \node[circle, draw] (ta)   at (26,8)  {};
            \path (tb) edge [sloped] node {$1$} (t);
            \path (ta) edge [sloped] node {$0$} (t);
        
            \node (d) at (22.5, 3) {$\ldots$};
            \node (d) at (22.5, -5) {$\ldots$};
         
            \node[circle, draw] (ta1)   at (24,12)  {};
            \node[circle, draw] (ta2)   at (24,4)  {};
            \path (ta1) edge [sloped] node {$0$} (ta);
            \path (ta2) edge [sloped] node {$1$} (ta);
        
            \node[circle, draw] (tb1)   at (24,-4)  {};
            \node[circle, draw] (tb2)   at (24,-12)  {};
            \path (tb1) edge [sloped] node {$0$} (tb);
            \path (tb2) edge [sloped] node {$1$} (tb);
        
            \path [draw = black, rounded corners, inner sep=100pt, dotted]
            (29.25, 1) -- (25.5, 12.25) -- (20.5, 16.25) -- (18.3, 16.25) -- (18.3, -16.25) -- (20.5, -16.25) -- (25.5, -12.25) -- (29.25, -1) -- cycle ;
            \node  (empty)    at (28.5, -6.5)  {\Large $V^{in}$};
        \end{scope}
    \end{scope}
\end{tikzpicture}
    }
  }
\caption{Illustration of our construction of $\mathcal A$ for a $k$-SAT formula $\Phi$ involving $N$ Boolean variables $x_1, \ldots, x_N$. We split the variables into $p$ blocks, resulting in $T = 2^{N/p}$ different assignments $A = \{a_1, \ldots, a_T\}$ for each of the $p$ blocks, thus giving a total of $p\cdot T$ subgraphs $\{S_{r, a_\ell}\}_{r \in [p], \ell\in [T]}$. Those subgraphs are expanded in Figure~\ref{fig: one reduction} on a particular $k$-SAT formula.}
\label{fig: reduction overview}
\end{figure}

Before giving the detailed description of the construction and eventually the proof of Theorem~\ref{thm: hardness}, we proceed with a sketch on how we achieve the above properties, see Figure~\ref{fig: reduction overview} for an illustration. Property~(1) is achieved by connecting the above-described subgraphs to the source node $s$ through a binary out-tree with $p\cdot T$ leaves (we call these nodes $V^{out}$). This tree has depth $\log p + N/p$. Property~(2) is achieved by connecting $V^{out}$ to $C$ through the nodes in $I$ that ensure that nodes $v,w \in C$ in the same relative positions (i.e., same distance from $\#$) are reached by strings of alternating co-lexicographic order, i.e., for a suitable string $\tau$ (that is a portion of the string labeling the common cycles) and two suitable strings $\alpha$ and $\beta$, $I_v$ contains two strings suffixed by $\alpha 00\tau$ and $\alpha 10\tau$, while $I_w$ contains two strings suffixed by $\beta 00\tau$ and $\beta 10\tau$. Independent of the relative co-lexicographic order of $\alpha$ and $\beta$ this witnesses that $\mathcal I(v) \cap \mathcal I(w) \neq \emptyset$.
Property~(3) is achieved by connecting $v_{M+N}^{r, a_\ell}$ from each subgraph $S_{r, a_\ell}$ with an edge labeled 0 to a complete binary in-tree (we call these nodes $V^{in}$) with the root being the only accepting state $t$. 
As the reversed automaton (i.e., the automaton obtained by reversing the direction of all transitions) is \emph{deterministic}, any two nodes in the graph can reach $t$ through a distinct binary string, witnessing that $\mathcal A$ is indeed minimal by the Myhill-Nerode characterization of the minimum DFA~\cite{nerode1958linear}. 
Property~(4) can be satisfied by an opportune transformation that maintains both forward and backward determinism.

\paragraph{Reduction Details}
Our goal now is to give the remaining details such that we can prove the following proposition. 
\begin{proposition}\label{prop: reduction}
    Let $\Phi$ be a $k$-SAT formula with $N$ variables and $M$ clauses and let $p\in [2, N]$ be an integer. Then, we can in 
    $T(N, M, p, k) := O(p \cdot 2^{N/p} \cdot (M\cdot k + N))$ time construct a DFA $\mathcal A$ 
    with 
    $n = T(N, M, p, k)$ nodes and 
    $m = T(N, M, p, k)$ edges such that $\mathcal A$ is minimum for its language $\mathcal L(\mathcal A)$ and $\Phi$ is satisfiable if and only if 
    there exists a non-empty string $\alpha$ such that
    $\mathcal A$ contains $p$ pairwise distinct nodes $u_i\in Q$, $i \in [p]$, with $\delta(u_i, \alpha)=u_i$ for all $i\in [p]$ and $\mathcal I(u_i) \cap \mathcal I(u_j) \neq \emptyset$ for all $i,j\in [p]$.
\end{proposition}

The DFA $\mathcal A= (Q, \Sigma, \delta, s, F)$ constructed for the given $k$-SAT formula $\Phi$ involving $N$ Boolean variables $x_1, \ldots, x_N$ and for a parameter $p\in [2, N]$ is then precisely defined as follows. For simplicity, we assume that $N$ and $p$ are powers of two.
The states $Q$ of $\mathcal A$ consist of four disjoint sets $V^{out}$, $I$, $C$, and $V^{in}$. We will now define these four sets and the transitions connecting them (thus defining $\delta$). We refer the reader back to Figure~\ref{fig: reduction overview} for an illustration.
Throughout the description, for an integer $P$ that is a power of 2 and an integer $i\in [P]$, we denote by $\rho_P(i)$ the bit-string of length $\log (P)$ that represents the integer $i - 1$ in binary. We recall that $T=2^{N/p}$. 
\begin{description}
    \item[$V^{out}$:] The set $V^{out}$ is connected as a complete binary out-tree of depth $\log(p) + N/p$ with root being the source node $s$. This tree has $p \cdot T$ leaves (corresponding to $T$ different assignments to each of the $p$ blocks of variables) that we call $u^{r, a_\ell}$, for $r\in [p]$ and $\ell\in [T]$. Here, $a_1, \ldots, a_T$ are all $T$ possibles assignments to the $N/p$ variables of one block. Furthermore, the tree is such that a leaf node $u^{r, a_\ell}$ is reached from $s$ by a unique path labeled $\rho_p(r)\rho_T(\ell)$.
    \item[$I$:] The set $I$ consists of $p \cdot T$ nodes, $\hat u^{r, a_\ell}$ for $r\in [p]$ and $\ell\in [T]$. Any node $\hat u^{r, a_\ell}$ is reachable from $u^{r, a_\ell}\in V^{out}$ by two edges labeled $0$ and $1$.
    \item[$C$:] The set $C$ is composed of $p \cdot T$ disjoint sets $S_{r, a_\ell}$ for $r\in [p]$ and $\ell\in [T]$ that have been described already above. We repeat the construction for completeness. For each $r\in [p]$ and $\ell\in [T]$, we have a set $S_{r, a_\ell}$ that contains $M+N+1$ nodes $v^{r, a_\ell}_0, \ldots, v^{r, a_\ell}_{M + N}$ and we have the following transitions within $S_{r, a_\ell}$. Every node $v^{r, a_\ell}_0$ is reachable from node $\hat u^{r, a_\ell}$ by a transition labeled 0. For $j\in [M]$, $r\in [p]$, $h\in [k]$, and an assignment $a_\ell$, we have (1) $\delta(v^{r, a_\ell}_{j - 1}, h) = v^{r, a_\ell}_j$ if $x(l_{j, h})\notin B_r$ and, (2) if $x(l_{j, h})\in B_r$, we have $\delta(v^{r, a_\ell}_{j - 1}, h) = v^{r, a_\ell}_j$ if and only if assignment $a_\ell$ makes literal $l_{j, h}$ true. In addition, for $i$ such that $M + i\in [M + 1, M + N]$, we have $\delta(v^{r, a_\ell}_{M + i - 1}, b) = v^{r, a_\ell}_{M + i}$ for $b\in \{0,1\}$ if $i \notin B_r$. For $i\in B_r$ instead we have $\delta(v^{r, a_\ell}_{M + i - 1}, b) = v^{r, a_\ell}_{M + i}$ if and only if $a_\ell$ assigns $b$ to variable $x_i$.    
    In addition we have $\delta(v^\ell_{M + N}, \#) = v^\ell_0$. There are no other transitions within $S_{r, a_\ell}$, see Figure~\ref{fig: one reduction} for an example of this construction. 
    \item[$V^{in}$:] The set $V^{in}$ is connected as a complete binary in-tree of depth $\log(p) + N/p$ with root being a state $t$. This tree again has $p\cdot T$ leaves that we call $t^{r, a_\ell}$ for $r\in [p]$ and $\ell\in [T]$. This tree is such that a leaf node $t^{r, a_\ell}$ can reach the root $t$ by a unique path labeled $\rho_p(r)\rho_T(\ell)$. Every node $t^{r, a_\ell}$ is reachable from $v^{r, a_\ell}_{M + N}$ by an edge labeled $0$.
\end{description} 
Finally, we define the set of final states as $F:=\{t\}$. We proceed with the following observation: The constructed DFA $\mathcal A$ is minimal, i.e., $\mathcal A = \mathcal A_{\min}$. 
This is an easy observation that results from inspecting the reversed automaton $\mathcal A^{\rev}$ that has the same set of states, $t$ as the source node, and all transitions reversed compared to $\mathcal A$. The automaton $\mathcal A^{\rev}$ is also deterministic and thus every state in $\mathcal A^{\rev}$ is reached by a unique string from the source state $t$. It follows that every two states $u,v$ are clearly distinguishable in $\mathcal A$ by the Myhill-Nerode relation\footnote{Two states are Myhill-Nerode equivalent --- i.e.\ they can be collapsed in the minimum DFA --- if and only if they allow reaching final states with the same set of strings.} \cite{nerode1958linear} as they can reach the only final state $t$ using a unique string.

\paragraph{Transformation to Binary Alphabet}
Let $\mathcal A'$ be the automaton described above. We will now apply a transformation to $\mathcal A'$ that will result in an automaton over the binary alphabet $\{0, 1\}$. The only edges that are not already labeled with $\{0,1\}$ are within $C$, more specifically within the first part of every subgraph $S_{r,a_\ell}$ and the potential labels of the edges here are $[k]\cup \{\#\}$. The transformation that we apply here has to satisfy the properties that both forward- and reverse-determinism are maintained and that the pattern that the letter $\#$ is replaced with can not appear anywhere besides on the part that the edge labeled $\#$ was replaced by. Our idea is to transform the edge labeled $\#$ into a path of $2\log(k) + 1$ transitions labeled 1 (assume that $k$ is a power of two for simplicity) and to replace an edge labeled with $h\in[k]$ with a path of length $2\log(k) + 1$ labeled $0\rho_k(h)\rho_k(h)^{\rev}$ where $\rho_k(h)^{\rev}$ is the reversed string of $\rho_k(h)$. Consider some pair of nodes $v_{j - 1}^{r, a_\ell}, v_j^{r, a_\ell}$ in the subgraph $S_{r,a_\ell}$ for $j\in [M]$. This pair maybe connected with a set of transitions $X \subseteq [k]$. In any case however $|X| \le k$. The idea is to, instead of introducing parallel paths for edges labeled $h$, introduce these paths through a binary out- and in-trie that contains exactly the paths $\rho_k(h)\rho_k(h)^{\rev}$ for $h\in X$, see Figure~\ref{fig: reduction transformation} for an example. This allows us to maintain forward- and reverse-determinism and at the same time keeps the number of newly introduced nodes and transitions bounded by $O(k)$ for each pair of nodes that the transformation gets applied to.

\begin{figure}[ht!]
  \centering{
    \resizebox{0.85\columnwidth}{!}{
      \begin{tikzpicture}[
          scale=.9,
          ->,
          >=stealth',
          shorten >=1pt,
          auto,
          semithick,
          every node/.style={minimum size=6mm}
        ]

            \node[circle, draw] (u)   at (-2,0)  {};
            \node[circle, draw] (v)   at (2,0)  {};
            \path (u) edge node {$1,2,6,7$} (v);

            
            \begin{scope}[xshift=5cm]
                \node[circle, draw] (v1)   at (0,0)  {};
                \node[circle, draw] (v2)   at (2,0)  {};
                \node[circle, draw] (v3)   at (4,1.714)  {};
                \node[circle, draw] (v4)   at (4,-1.714)  {};
                \node[circle, draw] (v5)   at (6,2.571)  {};
                \node[circle, draw] (v7)   at (6,-0.857)  {};
                \node[circle, draw] (v8)   at (6,-2.571)  {};
                \node[circle, draw] (v9)   at (8,3)  {};
                \node[circle, draw] (v10)   at (8,2.143)  {};
                \node[circle, draw] (v14)   at (8,-1.286)  {};
                \node[circle, draw] (v15)   at (8,-2.143)  {};
                \node[circle, draw] (w5)   at (10,2.571)  {};
                \node[circle, draw] (w7)   at (10,-0.857)  {};
                \node[circle, draw] (w8)   at (10,-2.571)  {};
                \node[circle, draw] (w3)   at (12,1.714)  {};
                \node[circle, draw] (w4)   at (12,-1.714)  {};
                \node[circle, draw] (w2)   at (14,0)  {};

                \path (v1) edge [sloped] node {$0$} (v2);
                \path (v2) edge [sloped] node {$0$} (v3);
                \path (v2) edge [sloped] node {$1$} (v4);
                \path (v3) edge [sloped] node {$0$} (v5);
                \path (v4) edge [sloped] node {$0$} (v7);
                \path (v4) edge [sloped] node {$1$} (v8);
                \path (v5) edge [sloped] node {$0$} (v9);
                \path (v5) edge [sloped] node[below] {$1$} (v10);
                \path (v7) edge [sloped] node {$1$} (v14);
                \path (v8) edge [sloped] node {$0$} (v15);

                \path (v9) edge [sloped] node {$0$} (w5);
                \path (v10) edge [sloped] node[below] {$1$} (w5);
                \path (v14) edge [sloped] node {$1$} (w7);
                \path (v15) edge [sloped] node {$0$} (w8);
                \path (w5) edge [sloped] node {$0$} (w3);
                \path (w7) edge [sloped] node {$0$} (w4);
                \path (w8) edge [sloped] node {$1$} (w4);
                \path (w3) edge [sloped] node {$0$} (w2);
                \path (w4) edge [sloped] node {$1$} (w2);
            \end{scope}
      \end{tikzpicture}
    }
  }
\caption{Transformation example for a $8$-SAT formula. Assume the formula $\Phi$ contains a clause $C_j$ such that there are parallel transitions with labels $X = \{1, 2, 6, 7\}$ between a pair of nodes $v_{j - 1}^{r, a_\ell}, v_j^{r, a_\ell}$ in $\mathcal A'$ and assume that $x(l_{j,1})\in B_r$, but $x(l_{j,2}), x(l_{j,6}), x(l_{j,7})\notin B_r$. Given that a transition labeled $1$ is present we can conclude that the assignment $a_\ell$ to the variables of $B_r$ makes $l_{j, 1}$ true. We can also conclude that the variables contained in the literals $l_{j, h}$ for $h\in \{3, 4, 5, 8\}$ are contained in the block $B_r$ as well but the assignment $a_\ell$ does not make these literals true.
The parallel transitions labeled $\{1, 2, 6, 7\}$ would be replaced by a subgraph of maximal out- and in-degree 2 that contains 4 paths labeled $0\; 000\; 000$ (for $1$), $0\; 001\; 100$ (for $2$), $0\; 101\; 101$ (for $6$), $0\; 110\; 011$ (for $7$). We note that the number of newly introduced nodes and transitions is $O(k)$.}
\label{fig: reduction transformation}
\end{figure}
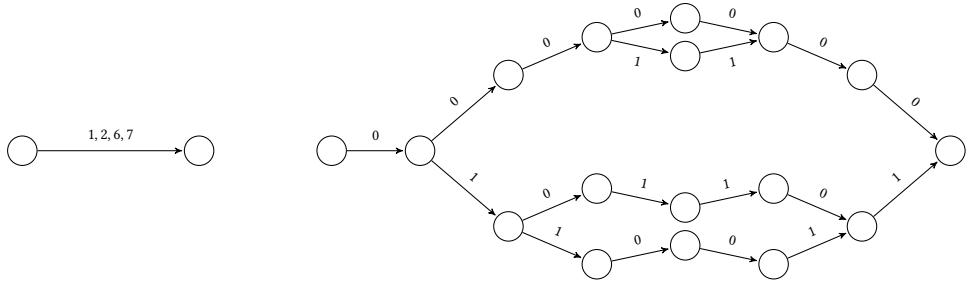

It is clear now that every path that originates from an integer $i\in [k]$ starts with a zero, while the path resulting from the letter $\#$ contains $2\log(k) + 1$ consecutive ones. As a consequence, the pattern $1^{2\log(k) + 1}$ appears only on the paths that originally corresponded to $\#$ and thus two transformed cycles match if and only if they used to match before the transformation.  
It is clear that this construction maintains both forward- and reverse- determinism and thus minimality. The nodes $v^{r,a_\ell}_0\in C$ have exactly two in-edges in $\mathcal A'$, one labeled $0$ and one labeled $\#$, in $\mathcal A$ these nodes will still have two in-edges, one labeled $0$ and one labeled $1$. Analogously, the nodes $v^{r,a_\ell}_{M + N}$ have two out-edges, one labeled $0$ and one labeled $\#$ and hence the transformation yields two out-edges, one labeled $0$ and one labeled $1$. 

We are now ready to prove Proposition~\ref{prop: reduction}.

\begin{proof}[Proof of Proposition~\ref{prop: reduction}]
    In order to bound the number of states in the constructed DFA $\mathcal A$, let us first consider the automaton $\mathcal A'$ before the transformation to the binary alphabet. We observe that both the set of nodes $V^{out}$ and $V^{in}$ are of cardinality at most $2pT$. The cardinality of $I$ is trivially bounded by $pT$ and the cardinality of $C$ by $pT \cdot (M + N)$. In total, the number of nodes in $\mathcal A'$ is thus bounded $O(pT \cdot (M + N))$. The number of transitions in $\mathcal A'$ instead is $O(pT \cdot (M\cdot k + N))$. As we have argued above the transformation of $\mathcal A'$ to $\mathcal A$ replaces $x\le k$ parallel edges between any two nodes of the first part of each subgraph $S_{r, a_\ell}$ by $O(k)$ new nodes and edges. In summary, we get that $\mathcal A$ contains at most $n=O(pT \cdot (M \cdot k + N)) = T(N, M, p, k)$ nodes and $m=T(N, M, p, k)$ transitions. 
    Clearly, the construction of $\mathcal A$ can be done in $T(N, M, p, k)$ time as well. We have argued above that $\mathcal A$ is minimal for the language $\mathcal L(\mathcal A)$ it accepts. Finally, using Lemma~\ref{lemma: cycle iff kSAT} we conclude that the set $\{S_{r, a_\ell}\}_{r\in [p], \ell\in [T]}$ and thus $\mathcal A$ contains $p$ disjoint 
    cycles spelling the same string, say $\alpha$, if and only if the $k$-SAT formula $\Phi$ is satisfiable. Due to the construction of $\mathcal A$ (the only directed cycles are the ones contained in $C$), it is clear that $p$ cycles spelling $\alpha$ exist if and only if there are $p$ nodes $u_1,\ldots, u_p$ with $\delta(u_r, \alpha) = u_r$ for each $r\in [p]$ and, w.l.o.g., we can assume that $u_r=v^{r, a_{i_r}}_{0}$ for some suitable assignment $a_{i_r}$. It remains to argue that $\mathcal I (u_i) \cap \mathcal I (u_j) \neq \emptyset$ for all $i\neq j$ holds if $\Phi$ is satisfiable. Let $i, j\in [p]$ be distinct. Notice that there is a path from the source node to $u_i$ spelling a string $\beta00$ as well as $\beta10$ -- the 0 and 1 being due to the transitions reaching the nodes in the set $I$ and $\beta$ being the string spelled by the path through $V^{out}$. Analogously, there is a path from the source node to $u_j$ spelling a string $\gamma 00$ as well as $\gamma 10$ and it must hold that $\beta \neq \gamma$. Now assume that $\beta \prec \gamma$, then $\beta 00 \prec \gamma 00 \prec \beta 10 \prec \gamma 10$ and thus $\mathcal I(u_i) \cap \mathcal I(u_j) \neq \emptyset$. Otherwise, if $\gamma \prec \beta$, we have $\gamma 00 \prec \beta 00 \prec \gamma 10\prec \beta 10$ and again $\mathcal I(u_i) \cap \mathcal I(u_j) \neq \emptyset$. This concludes the proof.
\end{proof}

We are now ready to prove Theorem~\ref{thm: hardness}. 

\begin{proof}[Proof of Theorem~\ref{thm: hardness}]
    Given a $k$-SAT formula $\Phi$ with $N$ variables and $M$ clauses, we apply the described reduction for some value $p\in [2, N]$.
    Following Proposition~\ref{prop: reduction}, the corresponding \textsc{DfaDetWidth} instance $\mathcal A$ can be built in $T(N, M, p, k) := O(p \cdot 2^{N/p} \cdot (M\cdot k + N))$ time and it has $n=T(N, M, p, k)$ states and $m=T(N, M, p, k)$ transitions. 
    Theorem~\ref{thm: main: width - cycle} implies that $\width^D(\mathcal L(\mathcal A))\ge p$ if and only if there exists a non-empty string $\alpha$ such that
    $\mathcal A$ contains $p$ pairwise distinct nodes $u_i\in Q$, $i \in [p]$, with $\delta(u_i, \alpha)=u_i$ for all $i\in [p]$ and $\mathcal I(u_i) \cap \mathcal I(u_j) \neq \emptyset$ for all $i,j\in [p]$.
    Proposition~\ref{prop: reduction} states that this is the case if and only if the $k$-SAT instance is a YES-instance. 
    Altogether it follows that the $k$-SAT instance is a YES-instance if and only if $\width^D(\mathcal L(\mathcal A))\ge p$.
    
    For the statements~\ref{enum 1} and \ref{enum 2}, choose $k=3$, $p=N$, and assume that $M=O(N)$ (this is w.l.o.g., see the sparsification lemma above). Observe that this choice implies $T(N, M, p, k) = O(N^2)$.
    \begin{enumerate}
        \item Assume that there exists an algorithm for \textsc{DfaDetWidth} with running time $\poly(n, m)$. We can employ this algorithm to check whether $\width^D(\mathcal L(\mathcal A))<p=N$ and thus solve the 3-SAT instance in $\poly(N)$ time implying that $P=NP$.
        \item Now, assume that there exists an algorithm for \textsc{DfaDetWidth} with running time $2^{o(\sqrt{m})}$. We can now employ this algorithm to check whether $\width^D(\mathcal L(\mathcal A))<N$ and thus solve the 3-SAT instance in $2^{o(N)}$ time implying that ETH fails.
        \item Choose $k=3$ and assume that $M=O(N)$ (again w.l.o.g., according to the sparsification lemma), and let $p\le m^c$ for any constant $c<1/2$. We now first show that this assumption on $p$ implies $p = O( N /\log N)$. Let $\eps>0$ be the constant such that $c = 1/2 - \eps$. 
        We now distinguish two cases. (1) Assume $\frac{N}{p} \ge \eps \cdot (\log p + \log N)$. In this case we immediately get $p \le 1/\eps \cdot N / \log N = O(N / \log N)$.
        (2) Now, assume the opposite, i.e., $\frac{N}{p} < \eps \cdot (\log p + \log N)$ and
        recall that $m\le C\cdot p \cdot 2^{N/p} \cdot N$ for a sufficiently large constant $C$. Hence 
        \[
            p
            \le m^{c}
            \le \Big(C\cdot p \cdot 2^{N/p} \cdot N\Big)^c
            = C^c \cdot 2^{c \cdot \big(\log p + \log N + \frac{N}{p}\big)}
            \le C^c \cdot 2^{c \cdot (1 + \eps) \cdot \big(\log p + \log N \big)}
            = C^c \cdot (pN)^{c \cdot (1 + \eps)}.
        \]
        This implies $p \le N^{\frac{c(1 + \eps)}{1 - c(1 + \eps)}}$. Now recall that $c = 1/2 -\eps$ and observe that thus
        \[
           \frac{c(1 + \eps)}{1 - c(1 + \eps)}
           = \frac{1/2 - \eps/2 - \eps^2}{1/2 + \eps/2 + \eps^2}
           < 1,
        \]
        and hence $p = O( N /\log N)$ also in this case.
        Now, assume that there exists an algorithm for \textsc{DfaDetWidth} with running time $ m^{o(p)}$. We can employ this algorithm to check whether $\width^D(\mathcal L(\mathcal A))<p$ and thus solve the 3-SAT instance in $O((p \cdot 2^{N/p} \cdot N)^{o(p)}) =  2^{o(N)}$ time, where we used that $p = O(N/\log N)$. This contradicts ETH.
        \item Now, let $p$ be any constant. Then, $T(N, M, p, k) = O(2^{N/p} \cdot (M\cdot k + N)$. Assume that there is an algorithm for \textsc{DfaDetWidth} with parameter $p$ that decides if $\width^D(\mathcal L(\mathcal A))<p$ in running time $O(m^{p -\epsilon})$ for some constant $\epsilon >0$. We can employ this algorithm to check whether $\width^D(\mathcal L(\mathcal A))<p$ and thus solve the $k$-SAT instance in
        \[
            O(m^{p -\epsilon}) 
            = O((2^{N/p} \cdot (M\cdot k + N))^{p - \epsilon})
            = O(2^{N \cdot (1 - \epsilon / p)} \cdot \poly(M, N))
        \]
        time, contradicting SETH.\qedhere
    \end{enumerate}
\end{proof}

\section*{Acknowledgments}

Ruben Becker, Davide Cenzato, Sung-Hwan Kim, Bojana Kodric and Nicola Prezza: Funded by ERC StG ``REGINDEX: Compressed indexes for regular languages with applications to computational pan-genomics'' grant nr 101039208. Views and opinions expressed are however those of the author(s) only and do not necessarily reflect those of the European Union or the European Research Council Executive Agency. Neither the European Union nor the granting authority can be held responsible for them. Alberto Policriti: Supported by project National Biodiversity Future Center-NBFC (CN\_00000033, CUP G23C22001110007) under the National Recovery and Resilience Plan of Italian Ministry of University and Research funded by European Union-NextGenerationEU.

\bibliographystyle{ACM-Reference-Format}
\bibliography{main}


\begin{thebibliography}{22}


\ifx \showCODEN    \undefined \def \showCODEN     #1{\unskip}     \fi
\ifx \showDOI      \undefined \def \showDOI       #1{#1}\fi
\ifx \showISBNx    \undefined \def \showISBNx     #1{\unskip}     \fi
\ifx \showISBNxiii \undefined \def \showISBNxiii  #1{\unskip}     \fi
\ifx \showISSN     \undefined \def \showISSN      #1{\unskip}     \fi
\ifx \showLCCN     \undefined \def \showLCCN      #1{\unskip}     \fi
\ifx \shownote     \undefined \def \shownote      #1{#1}          \fi
\ifx \showarticletitle \undefined \def \showarticletitle #1{#1}   \fi
\ifx \showURL      \undefined \def \showURL       {\relax}        \fi
\providecommand\bibfield[2]{#2}
\providecommand\bibinfo[2]{#2}
\providecommand\natexlab[1]{#1}
\providecommand\showeprint[2][]{arXiv:#2}

\bibitem[Alanko et~al\mbox{.}(2021)]%
        {alanko:iac21:wheeler}
\bibfield{author}{\bibinfo{person}{Jarno Alanko}, \bibinfo{person}{Giovanna D'Agostino}, \bibinfo{person}{Alberto Policriti}, {and} \bibinfo{person}{Nicola Prezza}.} \bibinfo{year}{2021}\natexlab{}.
\newblock \showarticletitle{Wheeler languages}.
\newblock \bibinfo{journal}{\emph{Information and Computation}}  \bibinfo{volume}{281} (\bibinfo{year}{2021}), \bibinfo{pages}{104820}.
\newblock
\urldef\tempurl%
\url{https://doi.org/10.1016/j.ic.2021.104820}
\showDOI{\tempurl}


\bibitem[Alanko et~al\mbox{.}(2024)]%
        {alanko_et_al:LIPIcs.CPM.2024.1}
\bibfield{author}{\bibinfo{person}{Jarno~N. Alanko}, \bibinfo{person}{Davide Cenzato}, \bibinfo{person}{Nicola Cotumaccio}, \bibinfo{person}{Sung-Hwan Kim}, \bibinfo{person}{Giovanni Manzini}, {and} \bibinfo{person}{Nicola Prezza}.} \bibinfo{year}{2024}\natexlab{}.
\newblock \showarticletitle{{Computing the LCP Array of a Labeled Graph}}. In \bibinfo{booktitle}{\emph{35th Annual Symposium on Combinatorial Pattern Matching (CPM 2024)}} \emph{(\bibinfo{series}{Leibniz International Proceedings in Informatics (LIPIcs)}, Vol.~\bibinfo{volume}{296})}, \bibfield{editor}{\bibinfo{person}{Shunsuke Inenaga} {and} \bibinfo{person}{Simon~J. Puglisi}} (Eds.). \bibinfo{publisher}{Schloss Dagstuhl -- Leibniz-Zentrum f{\"u}r Informatik}, \bibinfo{address}{Dagstuhl, Germany}, \bibinfo{pages}{1:1--1:15}.
\newblock
\showISBNx{978-3-95977-326-3}
\showISSN{1868-8969}
\urldef\tempurl%
\url{https://doi.org/10.4230/LIPIcs.CPM.2024.1}
\showDOI{\tempurl}


\bibitem[Becker et~al\mbox{.}(2023)]%
        {BeckerCCKKOP23}
\bibfield{author}{\bibinfo{person}{Ruben Becker}, \bibinfo{person}{Manuel C{\'{a}}ceres}, \bibinfo{person}{Davide Cenzato}, \bibinfo{person}{Sung{-}Hwan Kim}, \bibinfo{person}{Bojana Kodric}, \bibinfo{person}{Francisco Olivares}, {and} \bibinfo{person}{Nicola Prezza}.} \bibinfo{year}{2023}\natexlab{}.
\newblock \showarticletitle{Sorting Finite Automata via Partition Refinement}. In \bibinfo{booktitle}{\emph{31st Annual European Symposium on Algorithms, {ESA} 2023, September 4-6, 2023, Amsterdam, The Netherlands}} \emph{(\bibinfo{series}{LIPIcs}, Vol.~\bibinfo{volume}{274})}, \bibfield{editor}{\bibinfo{person}{Inge~Li G{\o}rtz}, \bibinfo{person}{Martin Farach{-}Colton}, \bibinfo{person}{Simon~J. Puglisi}, {and} \bibinfo{person}{Grzegorz Herman}} (Eds.). \bibinfo{publisher}{Schloss Dagstuhl - Leibniz-Zentrum f{\"{u}}r Informatik}, \bibinfo{pages}{15:1--15:15}.
\newblock
\urldef\tempurl%
\url{https://doi.org/10.4230/LIPICS.ESA.2023.15}
\showDOI{\tempurl}


\bibitem[Bille(2015)]%
        {bille2015regular}
\bibfield{author}{\bibinfo{person}{Philip Bille}.} \bibinfo{year}{2015}\natexlab{}.
\newblock \showarticletitle{On Regular Expression Matching and Deterministic Finite Automata}.
\newblock \bibinfo{journal}{\emph{Tiny ToCS}}  \bibinfo{volume}{3} (\bibinfo{year}{2015}), \bibinfo{pages}{1}.
\newblock


\bibitem[Cotumaccio(2022)]%
        {Cotumaccio22}
\bibfield{author}{\bibinfo{person}{Nicola Cotumaccio}.} \bibinfo{year}{2022}\natexlab{}.
\newblock \showarticletitle{Graphs can be succinctly indexed for pattern matching in {$O(|E|^2 + |V|^{5/2})$} time}. In \bibinfo{booktitle}{\emph{Data Compression Conference, {DCC} 2022, Snowbird, UT, USA, March 22-25, 2022}}, \bibfield{editor}{\bibinfo{person}{Ali Bilgin}, \bibinfo{person}{Michael~W. Marcellin}, \bibinfo{person}{Joan Serra{-}Sagrist{\`{a}}}, {and} \bibinfo{person}{James~A. Storer}} (Eds.). \bibinfo{publisher}{{IEEE}}, \bibinfo{pages}{272--281}.
\newblock
\urldef\tempurl%
\url{https://doi.org/10.1109/DCC52660.2022.00035}
\showDOI{\tempurl}


\bibitem[Cotumaccio(2023)]%
        {cotumaccio2023prefix}
\bibfield{author}{\bibinfo{person}{Nicola Cotumaccio}.} \bibinfo{year}{2023}\natexlab{}.
\newblock \showarticletitle{{Prefix Sorting DFAs: A Recursive Algorithm}}. In \bibinfo{booktitle}{\emph{34th International Symposium on Algorithms and Computation (ISAAC 2023)}} \emph{(\bibinfo{series}{Leibniz International Proceedings in Informatics (LIPIcs)}, Vol.~\bibinfo{volume}{283})}, \bibfield{editor}{\bibinfo{person}{Satoru Iwata} {and} \bibinfo{person}{Naonori Kakimura}} (Eds.). \bibinfo{publisher}{Schloss Dagstuhl -- Leibniz-Zentrum f{\"u}r Informatik}, \bibinfo{address}{Dagstuhl, Germany}, \bibinfo{pages}{22:1--22:15}.
\newblock
\showISBNx{978-3-95977-289-1}
\showISSN{1868-8969}
\urldef\tempurl%
\url{https://doi.org/10.4230/LIPIcs.ISAAC.2023.22}
\showDOI{\tempurl}


\bibitem[Cotumaccio et~al\mbox{.}(2023)]%
        {CotumaccioJACM23}
\bibfield{author}{\bibinfo{person}{Nicola Cotumaccio}, \bibinfo{person}{Giovanna D’Agostino}, \bibinfo{person}{Alberto Policriti}, {and} \bibinfo{person}{Nicola Prezza}.} \bibinfo{year}{2023}\natexlab{}.
\newblock \showarticletitle{{Co-Lexicographically Ordering Automata and Regular Languages - Part I}}.
\newblock \bibinfo{journal}{\emph{J. ACM}} \bibinfo{volume}{70}, \bibinfo{number}{4}, Article \bibinfo{articleno}{27} (\bibinfo{date}{aug} \bibinfo{year}{2023}), \bibinfo{numpages}{73}~pages.
\newblock
\showISSN{0004-5411}
\urldef\tempurl%
\url{https://doi.org/10.1145/3607471}
\showDOI{\tempurl}


\bibitem[Cotumaccio and Prezza(2021)]%
        {cotumaccio:soda21:psortable}
\bibfield{author}{\bibinfo{person}{Nicola Cotumaccio} {and} \bibinfo{person}{Nicola Prezza}.} \bibinfo{year}{2021}\natexlab{}.
\newblock \showarticletitle{{On Indexing and Compressing Finite Automata}}. In \bibinfo{booktitle}{\emph{Proceedings of the 32nd Annual ACM-SIAM Symposium on Discrete Algorithms (SODA)}}. \bibinfo{pages}{2585--2599}.
\newblock
\urldef\tempurl%
\url{https://doi.org/10.1137/1.9781611976465.153}
\showDOI{\tempurl}


\bibitem[D'Agostino et~al\mbox{.}(2023)]%
        {DAgostinoMP23}
\bibfield{author}{\bibinfo{person}{Giovanna D'Agostino}, \bibinfo{person}{Davide Martincigh}, {and} \bibinfo{person}{Alberto Policriti}.} \bibinfo{year}{2023}\natexlab{}.
\newblock \showarticletitle{Ordering regular languages and automata: Complexity}.
\newblock \bibinfo{journal}{\emph{Theor. Comput. Sci.}}  \bibinfo{volume}{949} (\bibinfo{year}{2023}), \bibinfo{pages}{113709}.
\newblock
\urldef\tempurl%
\url{https://doi.org/10.1016/j.tcs.2023.113709}
\showDOI{\tempurl}


\bibitem[Dietzfelbinger et~al\mbox{.}(1992)]%
        {Dietzfelbinger1992}
\bibfield{author}{\bibinfo{person}{M. Dietzfelbinger}, \bibinfo{person}{J. Gil}, \bibinfo{person}{Y. Matias}, {and} \bibinfo{person}{N. Pippenger}.} \bibinfo{year}{1992}\natexlab{}.
\newblock \showarticletitle{Polynomial hash functions are reliable}. In \bibinfo{booktitle}{\emph{Automata, Languages and Programming}}, \bibfield{editor}{\bibinfo{person}{W.~Kuich}} (Ed.). \bibinfo{publisher}{Springer Berlin Heidelberg}, \bibinfo{address}{Berlin, Heidelberg}, \bibinfo{pages}{235--246}.
\newblock


\bibitem[Dietzfelbinger and Meyer auf~der Heide(1990)]%
        {Dietzfelbinger90}
\bibfield{author}{\bibinfo{person}{Martin Dietzfelbinger} {and} \bibinfo{person}{Friedhelm Meyer auf~der Heide}.} \bibinfo{year}{1990}\natexlab{}.
\newblock \showarticletitle{A new universal class of hash functions and dynamic hashing in real time}. In \bibinfo{booktitle}{\emph{International Conference on Automata, Languages and Programming}}. \bibinfo{publisher}{Springer Berlin Heidelberg}, \bibinfo{address}{Berlin, Heidelberg}, \bibinfo{pages}{6--19}.
\newblock
\showISBNx{978-3-540-47159-2}


\bibitem[Eizenga et~al\mbox{.}(2020)]%
        {pangenomeGraphs}
\bibfield{author}{\bibinfo{person}{Jordan~M. Eizenga}, \bibinfo{person}{Adam~M. Novak}, \bibinfo{person}{Jonas~A. Sibbesen}, \bibinfo{person}{Simon Heumos}, \bibinfo{person}{Ali Ghaffaari}, \bibinfo{person}{Glenn Hickey}, \bibinfo{person}{Xian Chang}, \bibinfo{person}{Josiah~D. Seaman}, \bibinfo{person}{Robin Rounthwaite}, \bibinfo{person}{Jana Ebler}, \bibinfo{person}{Mikko Rautiainen}, \bibinfo{person}{Shilpa Garg}, \bibinfo{person}{Benedict Paten}, \bibinfo{person}{Tobias Marschall}, \bibinfo{person}{Jouni Sir\'{e}n}, {and} \bibinfo{person}{Erik Garrison}.} \bibinfo{year}{2020}\natexlab{}.
\newblock \showarticletitle{Pangenome Graphs}.
\newblock \bibinfo{journal}{\emph{Annual Review of Genomics and Human Genetics}} \bibinfo{volume}{21}, \bibinfo{number}{1} (\bibinfo{year}{2020}), \bibinfo{pages}{139--162}.
\newblock
\urldef\tempurl%
\url{https://doi.org/10.1146/annurev-genom-120219-080406}
\showDOI{\tempurl}
\showeprint{https://doi.org/10.1146/annurev-genom-120219-080406}
\newblock
\shownote{PMID: 32453966}.


\bibitem[Gagie et~al\mbox{.}(2017)]%
        {gagie:tcs17:wheeler}
\bibfield{author}{\bibinfo{person}{Travis Gagie}, \bibinfo{person}{Giovanni Manzini}, {and} \bibinfo{person}{Jouni Sir\'en}.} \bibinfo{year}{2017}\natexlab{}.
\newblock \showarticletitle{{Wheeler graphs: A framework for BWT-based data structures}}.
\newblock \bibinfo{journal}{\emph{Theoretical Computer Science}}  \bibinfo{volume}{698} (\bibinfo{year}{2017}), \bibinfo{pages}{67--78}.
\newblock
\urldef\tempurl%
\url{https://doi.org/10.1016/j.tcs.2017.06.016}
\showDOI{\tempurl}


\bibitem[Gibney and Thankachan(2022)]%
        {gibney2022complexity}
\bibfield{author}{\bibinfo{person}{Daniel Gibney} {and} \bibinfo{person}{Sharma~V Thankachan}.} \bibinfo{year}{2022}\natexlab{}.
\newblock \showarticletitle{{On the complexity of recognizing Wheeler graphs}}.
\newblock \bibinfo{journal}{\emph{Algorithmica}} \bibinfo{volume}{84}, \bibinfo{number}{3} (\bibinfo{year}{2022}), \bibinfo{pages}{784--814}.
\newblock


\bibitem[Hagerup(1998)]%
        {HagerupRAM}
\bibfield{author}{\bibinfo{person}{Torben Hagerup}.} \bibinfo{year}{1998}\natexlab{}.
\newblock \showarticletitle{Sorting and searching on the word RAM}. In \bibinfo{booktitle}{\emph{STACS 98}}, \bibfield{editor}{\bibinfo{person}{Michel Morvan}, \bibinfo{person}{Christoph Meinel}, {and} \bibinfo{person}{Daniel Krob}} (Eds.). \bibinfo{publisher}{Springer Berlin Heidelberg}, \bibinfo{address}{Berlin, Heidelberg}, \bibinfo{pages}{366--398}.
\newblock


\bibitem[Hopcroft(1971)]%
        {hopcroft1971n}
\bibfield{author}{\bibinfo{person}{John Hopcroft}.} \bibinfo{year}{1971}\natexlab{}.
\newblock \showarticletitle{An n log n algorithm for minimizing states in a finite automaton}.
\newblock In \bibinfo{booktitle}{\emph{Theory of machines and computations}}. \bibinfo{publisher}{Elsevier}, \bibinfo{pages}{189--196}.
\newblock


\bibitem[Impagliazzo et~al\mbox{.}(2001)]%
        {ImpagliazzoPZ01}
\bibfield{author}{\bibinfo{person}{Russell Impagliazzo}, \bibinfo{person}{Ramamohan Paturi}, {and} \bibinfo{person}{Francis Zane}.} \bibinfo{year}{2001}\natexlab{}.
\newblock \showarticletitle{Which Problems Have Strongly Exponential Complexity?}
\newblock \bibinfo{journal}{\emph{J. Comput. Syst. Sci.}} \bibinfo{volume}{63}, \bibinfo{number}{4} (\bibinfo{year}{2001}), \bibinfo{pages}{512--530}.
\newblock
\urldef\tempurl%
\url{https://doi.org/10.1006/jcss.2001.1774}
\showDOI{\tempurl}


\bibitem[Karp and Rabin(1987)]%
        {KR}
\bibfield{author}{\bibinfo{person}{Richard~M. Karp} {and} \bibinfo{person}{Michael~O. Rabin}.} \bibinfo{year}{1987}\natexlab{}.
\newblock \showarticletitle{Efficient randomized pattern-matching algorithms}.
\newblock \bibinfo{journal}{\emph{IBM Journal of Research and Development}} \bibinfo{volume}{31}, \bibinfo{number}{2} (\bibinfo{year}{1987}), \bibinfo{pages}{249--260}.
\newblock
\urldef\tempurl%
\url{https://doi.org/10.1147/rd.312.0249}
\showDOI{\tempurl}


\bibitem[Kim et~al\mbox{.}(2023)]%
        {KimOP23}
\bibfield{author}{\bibinfo{person}{Sung{-}Hwan Kim}, \bibinfo{person}{Francisco Olivares}, {and} \bibinfo{person}{Nicola Prezza}.} \bibinfo{year}{2023}\natexlab{}.
\newblock \showarticletitle{Faster Prefix-Sorting Algorithms for Deterministic Finite Automata}. In \bibinfo{booktitle}{\emph{34th Annual Symposium on Combinatorial Pattern Matching, {CPM} 2023, June 26-28, 2023, Marne-la-Vall{\'{e}}e, France}} \emph{(\bibinfo{series}{LIPIcs}, Vol.~\bibinfo{volume}{259})}, \bibfield{editor}{\bibinfo{person}{Laurent Bulteau} {and} \bibinfo{person}{Zsuzsanna Lipt{\'{a}}k}} (Eds.). \bibinfo{publisher}{Schloss Dagstuhl - Leibniz-Zentrum f{\"{u}}r Informatik}, \bibinfo{pages}{16:1--16:16}.
\newblock
\urldef\tempurl%
\url{https://doi.org/10.4230/LIPIcs.CPM.2023.16}
\showDOI{\tempurl}


\bibitem[Leiserson et~al\mbox{.}(1994)]%
        {leiserson1994introduction}
\bibfield{author}{\bibinfo{person}{Charles~Eric Leiserson}, \bibinfo{person}{Ronald~L Rivest}, \bibinfo{person}{Thomas~H Cormen}, {and} \bibinfo{person}{Clifford Stein}.} \bibinfo{year}{1994}\natexlab{}.
\newblock \bibinfo{booktitle}{\emph{Introduction to algorithms}}. Vol.~\bibinfo{volume}{3}.
\newblock \bibinfo{publisher}{MIT press Cambridge, MA, USA}.
\newblock


\bibitem[Manzini et~al\mbox{.}(2024)]%
        {manzini_et_al:LIPIcs.CPM.2024.23}
\bibfield{author}{\bibinfo{person}{Giovanni Manzini}, \bibinfo{person}{Alberto Policriti}, \bibinfo{person}{Nicola Prezza}, {and} \bibinfo{person}{Brian Riccardi}.} \bibinfo{year}{2024}\natexlab{}.
\newblock \showarticletitle{{The Rational Construction of a Wheeler DFA}}. In \bibinfo{booktitle}{\emph{35th Annual Symposium on Combinatorial Pattern Matching (CPM 2024)}} \emph{(\bibinfo{series}{Leibniz International Proceedings in Informatics (LIPIcs)}, Vol.~\bibinfo{volume}{296})}, \bibfield{editor}{\bibinfo{person}{Shunsuke Inenaga} {and} \bibinfo{person}{Simon~J. Puglisi}} (Eds.). \bibinfo{publisher}{Schloss Dagstuhl -- Leibniz-Zentrum f{\"u}r Informatik}, \bibinfo{address}{Dagstuhl, Germany}, \bibinfo{pages}{23:1--23:15}.
\newblock
\showISBNx{978-3-95977-326-3}
\showISSN{1868-8969}
\urldef\tempurl%
\url{https://doi.org/10.4230/LIPIcs.CPM.2024.23}
\showDOI{\tempurl}


\bibitem[Nerode(1958)]%
        {nerode1958linear}
\bibfield{author}{\bibinfo{person}{Anil Nerode}.} \bibinfo{year}{1958}\natexlab{}.
\newblock \showarticletitle{Linear automaton transformations}.
\newblock \bibinfo{journal}{\emph{Proc. Amer. Math. Soc.}} \bibinfo{volume}{9}, \bibinfo{number}{4} (\bibinfo{year}{1958}), \bibinfo{pages}{541--544}.
\newblock


\end{thebibliography}

\appendix

\section{Sorting infima and suprema strings}
\label{sec: appendix: infsup}

In \cite{BeckerCCKKOP23}, Becker et al. show how to compute the (co-lex) order of infima and suprema strings of any NFA in $O(m\log n)$ time. Their algorithm works by pruning transitions of the input automaton and produces as output a pseudo-forest (i.e., every state has an indegree of at most 1) with the property that the unique backward walk entering each state $v$ encodes $\inf I_v$: the infimum of state $v$ in the original automaton. The procedure for suprema strings is symmetric. Ultimately, infima and suprema strings can be computed and sorted by calling twice the algorithm of \cite{BeckerCCKKOP23}; the output of the whole process are two pseudo-forests encoding infima and suprema strings, respectively.  
Sorting the infima strings and suprema strings all together can be done using the suffix doubling algorithm in \cite{KimOP23}, which runs in $O(n\log n)$ time in this particular case, as also mentioned in \cite{BeckerCCKKOP23}.

One issue to be concerned about is that the algorithm \cite{BeckerCCKKOP23} assumes that (i) the source state of the input automaton does not have any incoming transition, and (ii)  transitions leading to the same state are labeled with the same character (the so-called \emph{input-consistency}), which it does not necessarily hold after minimizing the DFA in our application. We solve these issues on in the infima case; the suprema case is symmetric.
(i) The assumption on the source state $s$ with no incoming transition can be resolved by adding a new source state $s'$, a dummy state $s''$ (this dummy node will be useful when we want to obtain a reachable DFA as explained in the next paragraph), a transition from $s'$ to $s''$ and another transition from $s''$ to $s$, both labeled with a special character $\$\notin\Sigma$ where $\$$ is assumed to be smaller than every character in $\Sigma$. Note that this pre-processing just prepends $\$\$$ to the beginning of finite strings (i.e. those starting in the source node), and does not change the order of infima strings.
(ii) The issue related with input-consistency can be resolved simply by removing incoming transitions with non-minimum labels before computing a pseudo-forest for infima strings.
To see why this removal procedure works, consider (not necessarily distinct) nodes $u',u'',v\in Q$ such that $v=\delta(u',a)=\delta(u'',b)$ with $a\prec b\in\Sigma$. Then the infimum string $\inf I_v$ of $v$ cannot end with $b$ because for every string $\alpha\in I_{u'}$, $\alpha a\in I_v$ reaches $v$, and this string is always co-lexicographically smaller than every string ending with $b$. Thus, removing the transition labeled with $b$ does not affect the infimum string of $v$. In other words, for each node we can remove all transitions but the ones bearing the smallest label, thereby obtaining an input-consistent DFA preserving the infima strings of the original DFA.

Note that the above procedure could disconnect the DFA;  
some states might not be reachable from the source state after removing transitions. 
In fact, the pruning algorithm \cite{BeckerCCKKOP23} still works regardless of the reachability of the automaton although it was not clearly stated therein. Nevertheless, in order to ensure the correctness, we shall also show how to enforce the reachability without affecting the infima and suprema strings.
The following modification ensures the reachability of each of the two DFAs (one for infima and one for suprema strings), while preserving the infima. The procedure for suprema strings is symmetric, so here we do not describe it. 
Let $v_1,v_2,\cdots,v_n\in Q$ be any arbitrary ordering of the $n$ states, and define $v_{0}:=s''$ as the dummy state $s''$ added in the above procedure. Note that, for $0\le i\le n$, the label $\lambda(v_i)$ of each incoming transition is uniquely defined. For every $1\le i\le n$, we add a new state $v_{i-1,i}$ and add a transition from $v_{i-1}$ to $v_{i-1,i}$ with label $\#\in\Sigma$, and another transition from $v_{i-1,i}$ to $v_i$ with label $\lambda(v_i)$. Here, $\#$ is chosen to be larger than every character in $\Sigma\cup\{\$\}$. Now, every state of the resulting automaton is reachable from the (new dummy) source state. Note that the transition from $v_{i-1,i}$ to $v_{i}$ cannot contribute the infima. To see why, observe that, for every $1\le i\le n$, there exists a state $u_i$ such that $\delta(u_i,\lambda(v_i))=v_i$ and $u_i\ne v_{i-1,i}$. Since $\lambda(u_i)\prec\#=\lambda(v_{i-1,i})$, every string reaching $v_{i-1,i}$ is larger than every string reaching $u_i$, thus $v_{i-1,i}$ cannot contribute to the infimum string of $v_i$.

\end{document}